\documentclass[11pt]{article}

\usepackage[a4paper, margin=2.5cm]{geometry}
\usepackage{amsmath}
\usepackage{amssymb}
\usepackage{amsthm}
\usepackage{graphicx}
\usepackage{xurl}
\usepackage{multirow}
\usepackage{accents}
\usepackage{natbib}
\usepackage[hidelinks]{hyperref}

\DeclareMathOperator{\Var}{Var}
\DeclareMathOperator{\Cov}{Cov}
\DeclareMathOperator{\sgn}{sign}
\newcommand{\spn}{\mathrm{span}}

\theoremstyle{plain}
\newtheorem{theorem}{Theorem}[section]
\newtheorem{lemma}[theorem]{Lemma}
\newtheorem{proposition}[theorem]{Proposition}
\newtheorem{corollary}[theorem]{Corollary}
\theoremstyle{definition}
\newtheorem{definition}[theorem]{Definition}
\newtheorem{assumption}[theorem]{Assumption}
\newtheorem{remark}[theorem]{Remark}

\numberwithin{equation}{section}
\numberwithin{figure}{section}
\numberwithin{table}{section}

\newenvironment{keywords}{\par\medskip\noindent\textbf{Key words.} \ignorespaces}{\par\medskip}
\newenvironment{MSCcodes}{\par\noindent\textbf{AMS subject classifications.} \ignorespaces}{\par\medskip}
\newcommand{\email}[1]{\texttt{#1}}

\title{The Science and Practice of Trend-Following Systems}

\author{Artur Sepp\thanks{LGT Bank, Zurich, Switzerland (\email{artursepp@gmail.com}).}
\and Vladimir Lucic\thanks{Marex Solutions, London, and Imperial College London, London, United Kingdom (\email{v.lucic@imperial.ac.uk}).}}

\date{This version: July 20, 2026}

\begin{document}

\maketitle

\begin{abstract}
	We present a unified approach to designing trend-following (TF) systems and classify them into European, American, and Time Series Momentum categories. For European TF systems, we derive an exact relationship between profit-and-loss, autocorrelation, and drift in volatility-normalized returns. We analyze the expected return under fractional ARFIMA processes and show that TF systems are profitable when the long-term autocorrelation is positive, even under short-term mean reversion. In the frequency domain, the expected return is represented as a Poisson-kernel reading of the analytical or empirical spectrum of the volatility-normalized returns: the system profits at zero drift when the kernel-weighted spectral mass exceeds one, so trend-following alpha is excess spectral mass at low frequencies. Longer lookbacks benefit in addition from the squared drift of the return process. We derive the closed-form Sharpe ratio, with the excess kurtosis of the innovations entering through a single loading, and the net Sharpe ratio and cost-optimal span under trading costs. Under white noise, we derive the closed-form skewness of aggregated TF returns, which is positive at every horizon and peaks near half the filter span. Monte Carlo experiments confirm the analytical results. We show that the positive skewness of TF returns is structural under various model assumptions. Empirically, we evaluate the systems on liquid contracts, and show that all TF systems are strongly correlated and our analytical results can be applied for their performance attribution.  Our results enable design, simulation, and performance attribution of TF systems from trend persistence, mean reversion, drift, and skewness.
\end{abstract}



\section{Introduction}

Trend-following systems have been employed by many quantitative and discretionary funds, also known as commodity trading advisors (CTAs), or managed futures, since the early 1980s. The industry traces back to the turtle-trading classes of Richard Dennis in the early 1980s (\citealp{Angrist1989}), whose graduates created quantitatively driven CTA funds. According to BarclayHedge data, the assets under management (AUM) of managed futures funds stood at USD $340$ billion at the end of $2024$. A sizeable additional allocation runs through quantitative investment strategies (QIS) offered and executed by major investment banks.

\cite{Lintner1983} provided the first evidence that managed futures deliver better risk-adjusted returns and offer strong diversification benefits for long-only portfolios.
Subsequent studies (see, among others, \cite{Hamill2016}, \cite{Bouchaud2017}, \cite{Harvey2019}, \cite{Sepp2019}, and \cite{Sepp2020}) reinforce the diversification role of managed futures, so many private and institutional portfolios now carry some exposure to them.

Some large fund managers publish research papers on the analysis of their models, including \cite{Hurst2013} at AQR, \cite{Baz2015} at Man Group, and \cite{Bouchaud2017} at Capital Fund Management. Replicating an industry benchmark such as the SG Trend Index is reasonably straightforward, as we show in Section \ref{sc:impl}, and managers pursue outperformance through second-order proprietary features, including style factors, risk management, and other risk premia (\citealp{Carver2023}).

On the modeling side, the profitability of trend-following (TF) systems connects to the return-predictability literature in financial economics: \citet{PesaranTimmermann1995} document the economic significance of stock-return predictability, \citet{Cooper2004} document the dependence of momentum profits on the market state, and \citet{Bottazzi2019} derive short-term momentum and long-term reversal in equilibrium, which gives an economic origin for the serial dependence that we take as exogenous.

From a purely theoretical viewpoint, serial dependence under the statistical measure is not inconsistent with the absence of arbitrage: by the Dalang--Morton--Willinger theorem (\citealp{DalangMortonWillinger1990}), the absence of arbitrage in discrete time is equivalent to the existence of an equivalent martingale measure. Under that measure and the appropriate integrability conditions, the discounted price process has uncorrelated increments, which in general is not the case under the physical measure. Furthermore, the asset drift under the two measures can differ substantially.
Our net Sharpe ratio analysis sharpens the point quantitatively: at the empirically observed first-lag autocorrelations, the proxy break-even cost in equation \eqref{eq:sr_net_ar1} lies below realistic trading costs, so the short-memory predictability is not exploitable net of frictions.

This paper provides methodological and practical results on TF systems. We unify three approaches to the construction of TF systems, which we term European, American, and Time Series Momentum (TSMOM) systems, and we find a set of parameters for which the three systems are strongly correlated and generate similar performance to the SG Trend Index. In the theoretical part, we link the expected return of the European TF system to the autocorrelation function and drift of a generic process driving the dynamics of returns. In related analysis, \cite{Bouchaud2017} show that the performance of the risk-managed trend system is given as a difference between a long-term variance and a short-term variance of the returns of the underlying instrument. This difference-of-variances decomposition originates with \cite{BruderGaussel2011}, and our sample-path identity gives its exact EWMA-weighted sample form with the explicit boundary term. Quantitative aspects of TF systems in the Gaussian framework are studied by \cite{Grebenkov2014}, \cite{Ferreira2018}, \cite{Zakamulin2020}, \cite{Valeyre2025}. \cite{Jusselin2017} derive the mean, the variance, and the skewness of EWMA trend following under trend-plus-noise dynamics, and \cite{Martin2012} compute moments of linear momentum strategies under Gaussian dynamics. \cite{Martin2023} derives the term structure of the skewness of momentum trading returns in closed form, and treats the expected return by simulation. In the cross-section, \cite{LoMacKinlay1990} decompose trading profits into weighted autocovariances and a squared-mean term. \cite{GouldingHarveyMazzoleni2023} and \cite{MoreiraMuir2017} attribute the alpha of a risk-scaled strategy to the covariance between the risk signal and future returns.

We instead decompose the return of the raw signal into an autocorrelation term and a drift term, before any risk scaling. We extend the available results to derive an exact analytical formula linking the performance of the TF system to the autocorrelation of the volatility-normalized returns under generic processes. We show that the TF system is expected to be profitable when the autocorrelation of returns is positive even if the drift is zero. We also show that the TF system is expected to be profitable for any process with a large positive or negative drift if the filter span is large. We derive a closed-form proxy for the expected volatility-normalized turnover of the European TF system, exact for the signal component under independent Gaussian returns. Further, we derive the Sharpe ratio of the European TF system in closed form, using the Isserlis theorem for the variance of its daily returns. To our knowledge, the closed-form Sharpe ratio linking the risk-adjusted performance of a TF system to the autocorrelation function of the volatility-normalized returns is new. The two results hold at two levels of generality. The expected-return formula requires only a stationary mean, variance, and autocorrelation function, so it holds for any stationary autocorrelation function. The Sharpe ratio additionally assumes a causal linear process with independent innovations, which fixes the fourth-moment structure. The generating function admits a spectral representation in which the span selects the bandwidth of a Poisson kernel applied to the spectrum of the volatility-normalized returns (Section \ref{sec:sr}).

We illustrate the analytical results under white noise, AR-1, and fractional ARFIMA processes. The ARFIMA process carries both short-term and long-term mean reversion and trend features, and it extends the moving-average analysis of \cite{Brock1992}. In the empirical part, we apply the systems to liquid futures contracts and verify the analytical Sharpe ratios against the backtests. The scope of the present paper is the second and third moments of the strategy: the expected return, the variance, the gross and net Sharpe ratios, and, under white noise, the skewness of the aggregated returns, whose term structure \cite{Martin2023} derives in closed form for linear systems and which we obtain within our unified framework and take to the futures data. The convexity profile of TF overlays and the diversification benefits for long-only portfolios are analyzed in a companion paper.

Compared to the cited results, our contributions are the exact sample-path identity for the cumulative return and the closed-form solution for the gross Sharpe ratio under a generic stationary autocorrelation function within the linear-process class, including the long-short filter, together with the closed-form skewness of the aggregated returns under white noise. We further derive the leading-order net Sharpe ratio under proportional trading costs, with a closed-form break-even cost for short-memory dynamics and a power-law optimal span under long memory. We validate the formula on market data: in an in-sample reconstruction, the sample autocorrelation functions of $84$ futures contracts reproduce the realized Sharpe ratios with a pooled correlation of $0.99$ and a regression slope of $0.96$ (Section \ref{sec:attr}).

Our analysis complements the optimal-control strand of trend following. \cite{DaiZhangZhu2010} derive the optimal trading rule for a trend follower under a continuous-time regime-switching model, solving an optimal stopping problem for the entry and exit times. We take the trading rules as given by industry practice and characterize their performance analytically, so our evaluation results complement the optimal-design strand.

Our paper is organized as follows. In Section \ref{sc:frm}, we introduce the framework and necessary tools for the construction of TF systems. In Section \ref{sc:tfs}, we introduce the three approaches for building TF systems. In Section \ref{sc:eft}, we develop the European TF system and derive key analytical results. In Section \ref{sec:sr}, we derive the closed-form Sharpe ratio of the European TF system. In Section \ref{sec:sp}, we analyze the expected return and the Sharpe ratio of the European TF system under specific return-generating processes. In Section \ref{sc:impl}, we apply the TF systems considered to real-world data of liquid futures and consider several implementation aspects. All proofs and auxiliary results are provided in the appendices.

\section{Framework and Necessary Tools}\label{sc:frm}

The underlying instruments of most TF systems are futures. Futures markets provide leveraged exposure to equity indices, government bonds, foreign exchange (FX) rates, and commodities. Since futures contracts have set expiration dates, existing positions must be rolled over. Continuous time series of futures prices (termed as continuous futures) for historical analysis are obtained by stitching the front and second contracts at roll events, which occur at some time prior to the first notice day of the contract. The stitching adjusts the price level at each roll, so the relative returns of the continuous series carry no roll-related jumps and equal the excess returns of the held contract (\citealp{Carver2023}). We assume that the base currency for all systems is the United States Dollar (USD).

\begin{definition}[Futures markets and their returns]\label{eq:fm}
	We consider the set $\{t\}$ of observation times (without loss of generality, we assume weekday observations) and the set $\{s_{t}\}$ of sample prices of continuous futures incorporating excess USD returns.
	We define the lag-1 difference $d_{t}$ as follows:
	\begin{equation}\label{eq:int1} 
		d_{t} = s_{t} - s_{t-1},
	\end{equation}
	and the relative return $r_{t}$ as follows:
	\begin{equation} \label{eq:int2} 
		r_{t} = s_{t} / s_{t-1} - 1.
	\end{equation}
	
	We convert the daily return of a non-USD contract to USD by the most common futures convention, in which the local return is scaled by the FX ratio:
	\begin{equation} \label{eq:int3} 
		r_{t} = \left(s^{*}_{t}/s^{*}_{t-1}-1\right)\left(x_{t}/x_{t-1}\right),
	\end{equation}
	where $s^{*}_{t}$ is the local price and $x_{t}$ is the exchange rate of the local currency X to USD. Returns on futures contracts are excess returns, because futures margins are small and funding costs are embedded into futures returns by arbitrage.
	
\end{definition}

\subsection{EWMA Filter}

The design of TF systems is typically based on inference of trend indicators from noisy price and returns data. One of the most commonly used filters for the design of TF systems is the exponentially weighted moving average (EWMA) filter.

We consider a stochastic process $y_t$ on the grid of discrete observation times $\{t\}$ with ordered sequence $\{y_{t}\}$. In our context, $\{y_{t}\}$ represents the realizations of the price or returns of a financial instrument, such as a financial future contract. We introduce the statistical measure $\mathbb{P}$ for computing population moments of $\{y_{t}\}$.

\begin{definition}[EWMA Filter, \cite{Holt1957}] The exponentially weighted moving average (EWMA) filter with smoothing parameter $\nu$, $0<\nu<1$, on sequence $\{y_{t}\}$ is defined as:
	\begin{equation} \label{eq:ewma1} 
		\begin{split}
			\mathcal{L}^{(\nu)} (y_{t}) &  \equiv   \left(1-\nu \right)  \left(  y_{t} + \nu^{1} y_{t-1} + \nu^{2} y_{t-2}+... \right) \\
			&  =  \left(1-\nu \right)  \sum^{\infty}_{m=0} \nu^{m} y_{t-m}\\
			&  =  \left(1-\nu \right)y_{t} +  \nu\mathcal{L}^{(\nu)} (y_{t-1}).
		\end{split}
	\end{equation}
\end{definition}

The smoothing parameter can be specified in terms of the span period as follows:
\begin{equation} \label{eq:i2} 
	\nu =1 - \frac{2}{\spn+1},
\end{equation}
where $\spn$ is the lookback period of the EWMA filter measured in units of observation times.

The EWMA filter in equation \eqref{eq:ewma1} for serially independent variables $\{y_{t}\}$ can be interpreted as an estimator of the running mean with the following variance:
\begin{equation} \label{eq:v7} 
	\Var\left[ \mathcal{L}^{(\nu)} (y_{n}) \right]  \equiv   \frac{ 1-\nu   }{ 1+\nu} \vartheta_{0} =  \frac{ 1}{ \spn } \vartheta_{0}.
\end{equation}
Thus, for serially independent data, the EWMA filter reduces the variance as much as the moving-average (MA) filter with the window equal to the span.

\begin{definition}[Variance-preserving EWMA filter] We define the adjusted filter:
	\begin{equation} \label{eq:ewma2} 
		\tilde{\mathcal{L}}^{(\nu)} (y_{t})   = \sqrt{ \frac{ 1+\nu}{ 1-\nu} } \mathcal{L}^{(\nu)} (y_{t}). 
	\end{equation}
\end{definition}

\begin{proposition}\label{prop1} Let $\{y_{t}\}$ be serially independent under $\mathbb{P}$ with variance $\Var[y_{t}]=\vartheta_{0}$. The variance-preserving filter in equation \eqref{eq:ewma2} satisfies:
	\begin{equation} \label{eq:ewma2a} 
		\Var[\widetilde{\mathcal{L}}^{(\nu)} (y_{t})]   = \vartheta_0.
	\end{equation}
\end{proposition}
The proof is provided in Appendix \ref{sec:prop1}.

Finally, we introduce the long-short EWMA filter, termed a double filter by \citet{Martin2012} and \citet{Martin2023}.

\begin{definition}[Variance-preserving long-short filter] We define the long-short EWMA filter as follows:
	\begin{equation} \label{eq:ewma3} 
		\begin{split}
			\widetilde{\mathcal{LS}}^{(\nu_{1}, \nu_{2})}(y_{n})   & \equiv \widetilde{l}_{1} \widetilde{\mathcal{L}}^{(\nu_{1})} (y_{n}) - \widetilde{l}_{2}\widetilde{\mathcal{L}}^{(\nu_{2})} (y_{n}) \\
			&   = \widetilde{l}_{1}\sqrt{\frac{ 1+\nu_{1}   }{ 1-\nu_{1}}} \mathcal{L}^{(\nu_{1})} (y_{n}) - \widetilde{l}_{2}\sqrt{\frac{ 1+\nu_{2}   }{ 1-\nu_{2}}} \mathcal{L}^{(\nu_{2})} (y_{n})\\
			&   \equiv l_{1} \mathcal{L}^{(\nu_{1})} (y_{n}) - l_{2} \mathcal{L}^{(\nu_{2})} (y_{n}),
		\end{split}
	\end{equation}
	where loading coefficients $l_{1}$, $l_{2}$ and $\widetilde{l}_{1}$, $\widetilde{l}_{2}$ are defined by:
	\begin{equation} \label{eq:ewma3a} 
		\begin{split}
			& l_{1} =\frac{1}{ 1-\nu_{1} } q, \ l_{2} =\frac{1}{1-\nu_{2}} q,\\
			& \widetilde{l}_{1} =\frac{1}{\sqrt{ 1-\nu^{2}_{1}} } q, \ \widetilde{l}_{2} =\frac{1}{\sqrt{ 1-\nu^{2}_{2}} } q, \\ 
			& q = \left( \frac{1}{ 1-\nu^{2}_{1}}  + \frac{1}{ 1-\nu^{2}_{2} } - 2 \frac{1}{1-\nu_{1}\nu_{2}} \right)^{-1/2}.
		\end{split}
	\end{equation}
\end{definition}

\begin{proposition} \label{prop2} Assuming a sequence of serially independent and identically distributed random variables $\{y_t\}$ with $\Var[y_t]=\vartheta_{0}$ and distinct spans $\nu_{1}\neq\nu_{2}$, the variance of the filter in equation \eqref{eq:ewma3} is given by:
	\begin{equation} \label{eq:ewma3c} 
		\Var\left[ \widetilde{\mathcal{LS}}^{(\nu_{1}, \nu_{2})}(y_{n})  \right] = \vartheta_0.
	\end{equation}
\end{proposition}
The proof is provided in Appendix \ref{sec:prop2}.

Figure \ref{fig:signal_weight} illustrates the filters for the unit-impulse sequence $(1,0,0,\ldots)$ at three combinations of spans.

\begin{figure}[tph]
	\begin{center}\hspace*{-3\baselineskip}\vspace*{-\baselineskip} 
		\includegraphics[width=0.78\textwidth, angle=0] {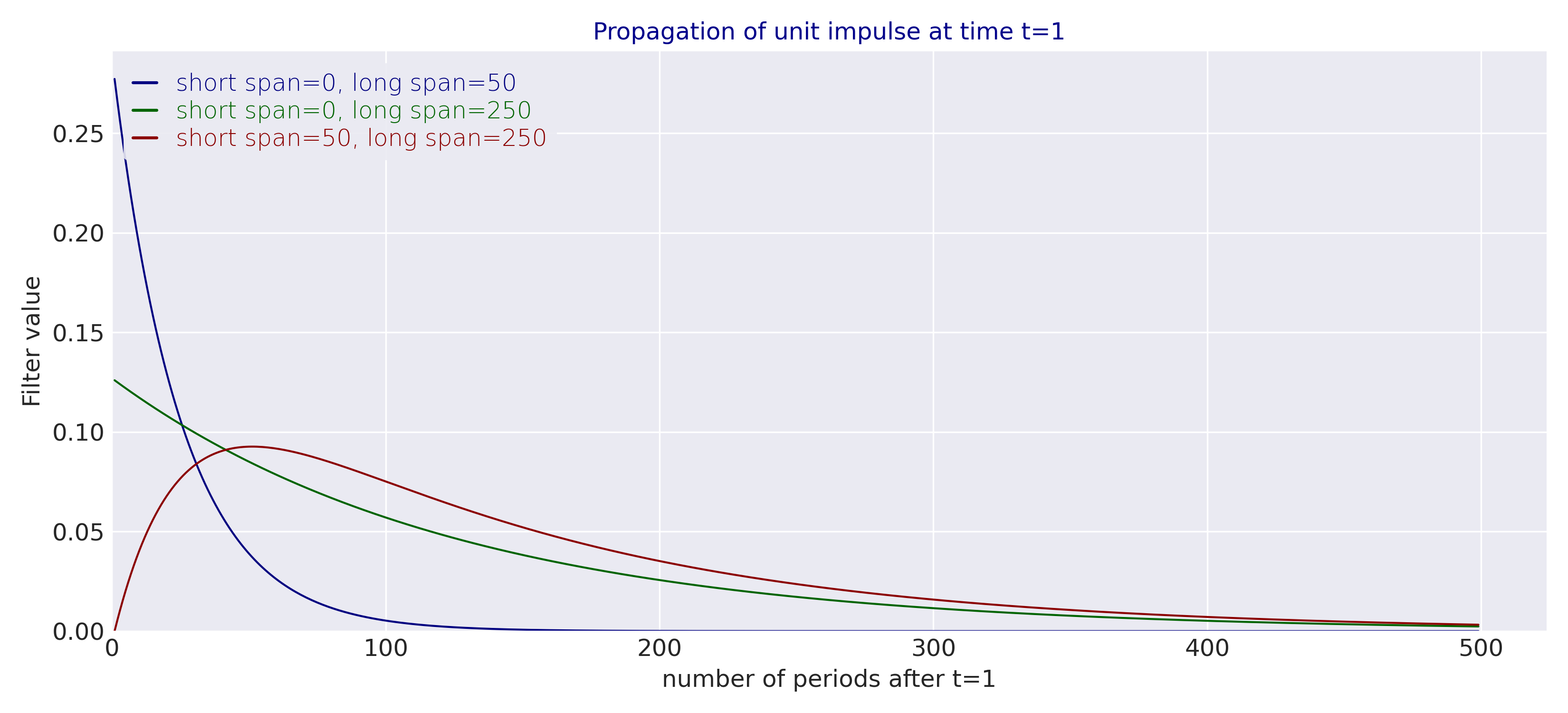}\hspace*{-2.5\baselineskip} 
	\end{center}
	\caption{Values of EWMA long-short filter applied to the unit-impulse sequence $(1,0,0,\ldots)$ using long span $N_{1}$ such that $\nu_{1}=1-2/(N_{1}+1)$ and short span $N_{2}$ such that $\nu_{2}=1-2/(N_{2}+1)$ for three combinations: $(N_{1}=50, N_{2}=0)$, $(N_{1}=250, N_{2}=0)$, $(N_{1}=250, N_{2}=50)$, where $N_{2}=0$ denotes the single-filter case without the short leg}
	\label{fig:signal_weight}
\end{figure}

\subsection{Volatility Estimation and Volatility-Normalized Returns}

Volatility is the key input for position sizing. Most often, it is estimated with the EWMA filter without mean adjustment, applied to the price changes in equation \eqref{eq:int1}:
\begin{equation} \label{eq:pricevol}
	\sigma^{price}_{t}=\sqrt{\mathcal{L}^{(\nu_{\sigma})}\left(d^{2}_{t}\right)},
\end{equation}
or for relative returns in equation \eqref{eq:int2} by:
\begin{equation} \label{eq:ewmavol}
	\sigma^{return}_{t}=\sqrt{\mathcal{L}^{(\nu_{\sigma})}\left(r^{2}_{t}\right)},
\end{equation}
where $\nu_{\sigma}$ is the corresponding smoothing parameter.

The American TF system measures volatility by the average true range (ATR) of daily prices, which we define in Appendix \ref{app:tfs}. The ATR-based estimates are highly correlated with the standard deviations of returns, with an average $R^{2}$ of $90\%$ across our core futures contracts, so both risk measures lead to similar system behavior.

We compute the volatility-normalized daily returns $z_{t}$ as follows:
\begin{equation} \label{eq:eu1} 
	z_{t} = \frac{r_{t}}{\sigma_{t-1}},
\end{equation}
where $\sigma_{t-1}$ is lag-1 volatility of daily returns. 

The returns $z_{t}$ can be interpreted as the returns of a strategy with volatility targeting in equation \eqref{eq:eu3}  with $\sigma_{\mathrm{target}}=\sqrt{a}$. Empirically, the realized variance of $z_{t}$ is close to one when the lagged volatility estimator is well calibrated, sufficiently stable, and tracks the conditional second moment of the returns. For S\&P 500 index futures, the EWMA volatility averages $17\%$, the volatility-target weight from equation \eqref{eq:eu3} with $\sigma_{\mathrm{target}}=15\%$ averages $1.09$, and the volatility-normalized returns from equation \eqref{eq:eu1} have a standard deviation of $0.96$.

\section{Trend-Following Systems}\label{sc:tfs}

We differentiate between three approaches for signal generation of TF systems, based on our experience and communications with CTA managers.

\begin{definition}[Trend-following systems] ~
	
	\textbf{1. European TF} system is based on continuous weights computed using the EWMA filter (other filters can also be applied) as in equations \eqref{eq:eu2} and \eqref{eq:eu3}. This approach is formulated by a number of large European CTA managers (see \cite{Bouchaud2017} and \cite{Baz2015}). European systems assume continuous position sizes proportional to the signal strength, so the number of contracts fluctuates day-to-day, up to discretization given contract sizes.
	
	\textbf{2. American TF} originates from early adopters of CTAs in North America\footnote{The labels European and American follow industry usage and do not refer to option exercise styles. The European system updates the position continuously, and the American system trades at discrete breakout and stop levels.}, in particular, from the so-called turtle traders in the 1980s (\citealp{Covel2009}, \citealp{Curtis2007}). Channel and breakout systems are tested by \cite{Lukac1988} and \cite{Szakmary2010}, \cite{ParkIrwin2007} survey technical trading rules, \cite{KaminskiLo2014} analyze the stop-loss exit leg, and \cite{LevinePedersen2016} and \cite{Brock1992} study moving-average crossover rules. 
	We generalize this approach in Definition \ref{as:am} following the design of the turtle systems. The American system trades break-outs from ranges with binary position sizes, so the exposure is fully allocated when the signal is on and zero otherwise.

	\textbf{3. Time Series Momentum (TSMOM) TF} system is based on the price momentum as in equation \eqref{eq:tsmom1} or, more generally, on the momentum of the returns adjusted for volatility as in equation \eqref{eq:tsmom2}. This approach is commonly used in academic studies (see \cite{Moskowitz2012}, \cite{Hurst2013}, \cite{Baltas2015}, and \cite{Dudler2015}). The rebalancing frequency is typically monthly in the cited literature. In equation \eqref{eq:tsmom2}, the frequency is set by the period length $L$ in days, and in practice we recommend weekly to monthly rebalancing with $L\in(5, 30)$.
	
\end{definition}

Although the three systems generate trend-following signals differently, Section \ref{sec:comp} finds parameter specifications under which they deliver similar risk-adjusted performances and closely correlated returns, matching CTA benchmarks such as the SG Trend Index.

\subsection{European TF System}

Our approach to generate the trend-following signal is close to that of \cite{Bouchaud2017} who specify the allocation to each contract by:
\begin{equation} \label{eq:etf1} 
	w_{t} = \frac{1}{\sigma_{t}}  \mathcal{L}^{(\nu)} \left(\frac{d_{t}}{\sigma_{t-1}}\right), 
\end{equation}
where $\mathcal{L}^{(\nu)}$ is the EWMA operator defined in equation \eqref{eq:ewma1}, $\sigma_{t}$ is the volatility estimated using the EWMA filter in equation \eqref{eq:pricevol}. We differ by using relative returns for the volatility in equation \eqref{eq:ewmavol} and volatility-normalized returns for the signal, with the full construction in Definition \ref{as:eu}. The meta-parameters are the smoothing value $\nu$ of the long filter or the pair $(\nu_{1}, \nu_{2})$ of a long-short filter.

\subsection{American and TSMOM TF Systems}

The American TF system generates discrete long and short positions from the crossover of two EWMA filters applied to futures prices. The entry rules apply a signal buffer, and the exit rules apply trailing stop-losses, with both scaled by the ATR of daily prices. Position sizes are binary and fixed at trade inception, so the system trades at discrete breakout and stop levels. We define the system in Appendix \ref{app:tfs}.

The TSMOM system sizes positions by the momentum of past returns, following \cite{Moskowitz2012}. We generalize the signal to the normalized sum of the signs of volatility-normalized returns over $M$ periods of $L$ days each, mirroring the European signal. We define the system in Appendix \ref{app:tfs}.

The three systems monetize the same underlying trend signal, and Section \ref{sec:comp} documents their strong pairwise correlations and their $80\%$ average correlation with the SG Trend Index. We focus the analysis on the European system, because its continuous signal admits exact closed-form results, and Section \ref{sec:attr} shows empirically that the European closed form ranks the realized performance of the American and TSMOM implementations.

\section{Construction of European TF System}\label{sc:eft}

\begin{definition}[European TF system]\label{as:eu} ~

	\begin{enumerate}
		
		\item We consider the sequence of daily returns $\{r_{t}\}$ of a continuous futures contract. We introduce an estimator for the daily (non-annualized) volatility of the daily returns that produces a sequence of estimated volatilities $\{\sigma_{t}\}$ of daily returns. We use volatility-normalized daily returns $z_{t}$ as computed in equation \eqref{eq:eu1}.

		\item The signal $S_{t}$ is computed using either the variance-preserving EWMA filter in equation \eqref{eq:ewma2} or the long-short filter in equation \eqref{eq:ewma3} as follows:
		\begin{equation} \label{eq:eu2} 
			S_{t} = \widetilde{\mathcal{L}}^{(\nu)} (z_{t}) , \ S_{t} = \widetilde{\mathcal{LS}}^{(\nu_{1}, \nu_{2})} (z_{t}). 
		\end{equation}
		
		\item The position size (or the position weight per trade level) is linearly proportional to the signal $S_{t}$ and to the volatility-target weight $w^{vt}_{t}$ as follows:
		\begin{equation} \label{eq:eu3} 
			w_{t} =  S_{t}w^{vt}_{t}, \ w^{vt}_{t}=\frac{\sigma_{\mathrm{target}}}{\sqrt{a}\sigma_{t}}, 
		\end{equation}
		where $\sigma_{\mathrm{target}}$ is the annualized volatility target and $a$ is the annualization factor ($a=260$ for daily returns, the weekday count of Definition \ref{eq:fm}).
		
	\end{enumerate}
	
\end{definition}

We stress that the goal of targeting volatility in equation \eqref{eq:eu3} is not risk management. Instead, we seek volatility-normalized returns for signal generation and for the performance of the TF system on the instrument level.

\begin{corollary}[Daily return of European TF system] Using equation \eqref{eq:eu3} for computing the end-of-day weights of the European TF system, the daily return of the system for each instrument is computed by\footnote{We assume end-of-day rebalancing at the close. In practice, orders are sampled and submitted shortly before the session close.}:
	\begin{equation} \label{eq:ret} 
		f_{t}  \equiv w_{t-1} r_{t} =  S_{t-1}  \frac{\sigma_{\mathrm{target}}}{\sqrt{a}\sigma_{t-1}}r_{t}= \frac{\sigma_{\mathrm{target}}}{\sqrt{a}}  S_{t-1} z_{t}.
	\end{equation}
	
\end{corollary}

As a result, the day-$t$ return of the trend-following strategy is proportional to the volatility-normalized return weighted by the signal observed on the previous day, so the daily return is a path-dependent function of the past volatility-normalized returns.

In the idealized model in which the volatility-normalized returns are independent with zero mean and unit variance and the volatility denominator is treated as fixed, the variance-preserving signal has unit variance at every signal span, and the strategy variance equals $\sigma^{2}_{\mathrm{target}}/a$. With an estimated volatility, the variance of $z_{t}$ and of the weights depends on the span of the volatility estimator, so the invariance is approximate. The long-term volatility of the trend-following system is then close to $\sigma_{\mathrm{target}}$.

In Figure \ref{fig:ES1_short_signals}, we illustrate the European TF system on the S\&P 500 index future for fast, medium, and slow long-short EWMA filters. The standard deviation of all three filters is close to one in panel (A). The average weights in panel (B) are positive, reflecting the positive long-term drift of the index, and the weights become extreme during periods of low realized volatility, which CTA managers mitigate with floors in practice. The realized volatility of daily returns on the TF system in panel (C) is $18\%$ on average, close to the $15\%$ target.

\begin{figure}[t]
	\begin{center}\hspace*{-3\baselineskip}\vspace*{-\baselineskip} 
		\includegraphics[width=0.80\textwidth, angle=0] {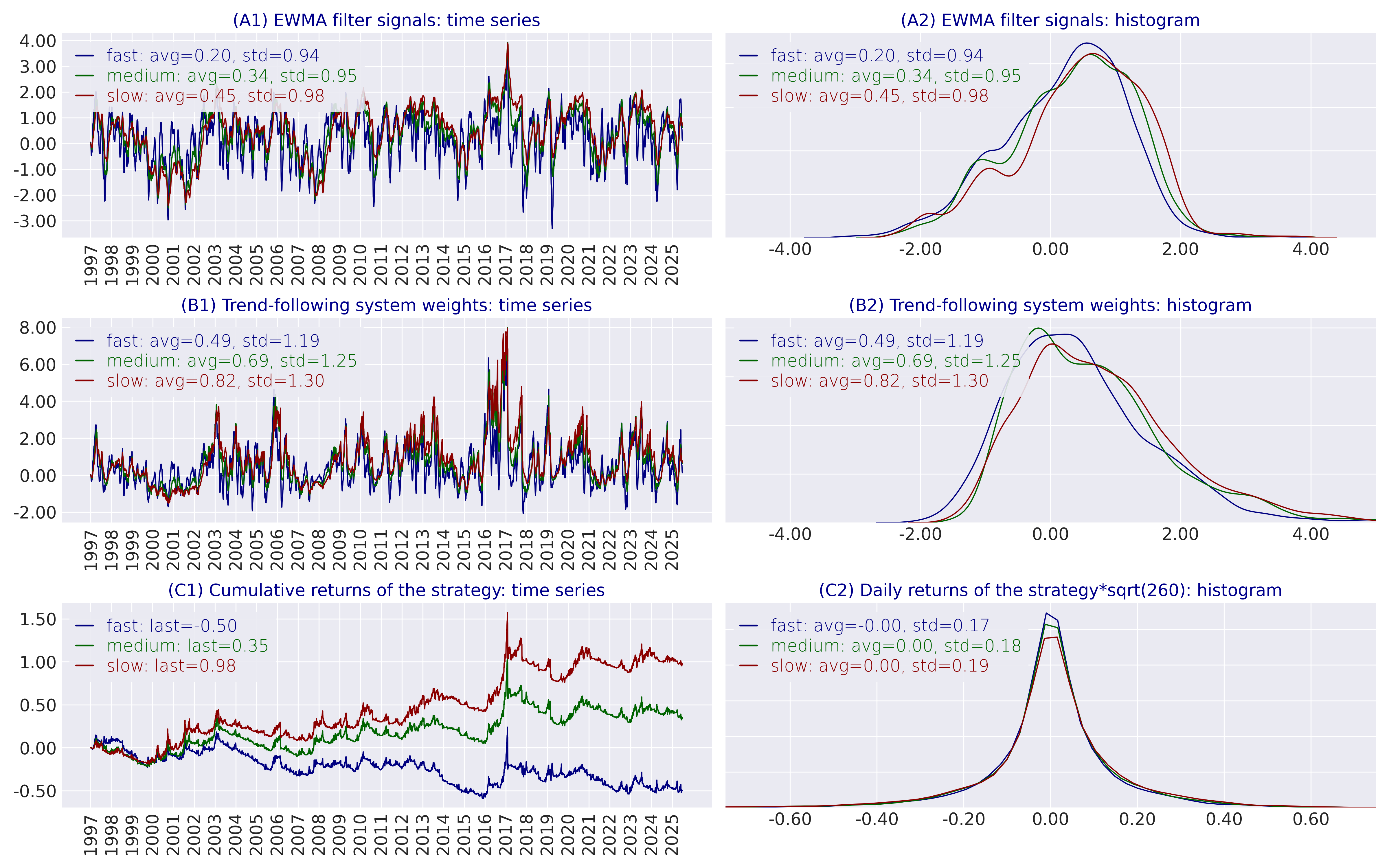}\hspace*{-2.5\baselineskip} 
	\end{center}
	\caption{Panels (A1) and (A2) show time series and histogram of EWMA long-short filters applied on volatility-normalized returns. Panels (B1) and (B2) show time series of weights with volatility target of $15\%$. Panels (C1) and (C2) show the cumulative P\&L and daily returns multiplied by $\sqrt{a}$ with $a=260$. Fast, medium, and slow long-short EWMA filters are computed using long spans of $30$, $125$, and $250$ days, respectively. The short span for all filters is $20$ days. The labels show average values (avg) and standard deviation (std). The underlying contract is S\&P 500 index futures with data taken from 31 December 1997 to 10 July 2026}
	\label{fig:ES1_short_signals}
\end{figure}

\subsection{Connection to Merton Problem and Kelly Sizing}\label{sc:kelly}

The connection between the position sizing of the European TF system in equation \eqref{eq:eu3} and the Merton problem for the optimization of long-term portfolio growth (\citealp{Merton1969}) offers a useful interpretation. Practitioners describe volatility targeting as fractional Kelly sizing (\citealp{Carver2023}), and we make the link explicit. \citet{Watt2026} revisits the discrete-time growth-optimal allocation over one risky and one risk-free asset, which is the setting of our daily grid with excess returns.

Returns on futures contracts are excess returns, so the Kelly rule of \citet{Merton1969}, $w^{*}=\tilde{\mu}/\sigma^{2}$, becomes $w^{*}=\left(\tilde{\mu}/\sigma\right)\times\left(1/\sigma\right)$, where the first term is the Sharpe ratio of the continuous futures and the second term is the volatility sizing. Volatilities and Sharpe ratios are time-inhomogeneous and cyclical, so we replace the two terms with their rolling estimates, $\widehat{w}^{*}_{t}=\widehat{SR}_{t}/\widehat{\sigma}_{t}$. The weight of the European TF system in equation \eqref{eq:eu3} has this form up to the filter loading, because the EWMA filter of the volatility-normalized returns estimates the Sharpe ratio, $\widehat{SR}_{t} = \mathcal{L}^{(\nu)} (z_{t})$. The variance-preserving filter applies the loading $l=\sqrt{\spn}$, so the weight is a span-dependent multiple of the Kelly fraction. The Kelly rule also assumes independent returns, while we apply the filter to serially dependent returns. For the white noise model with drift, equation \eqref{eq:wn3} gives the expected annual return of the TF system proportional to the square of the annualized drift, which is in line with the expected growth rate of one-half of the squared Sharpe ratio in the continuous-time Merton solution.

\subsection{Cumulative Return of European TF System}

By construction, the daily return of the European TF system is given in equation \eqref{eq:ret}. For generality, we consider the linear scaling of the EWMA filter for signal generation in equation \eqref{eq:eu2}:
\begin{equation} \label{eq:e1} 
	\mathcal{L}^{(\nu, l)} (z_{t})  \equiv l \mathcal{L}^{(\nu)} (z_{t}),
\end{equation}
where loading $l$ is $l=1$ for the EWMA filter $\mathcal{L}^{(\nu)}$ in equation \eqref{eq:ewma1}, $l=\sqrt{(1+\nu)/(1-\nu)}$ for the variance-preserving EWMA filter in equation \eqref{eq:ewma2}, and $l=l_{1}$ or $l=l_{2}$ for weights of the long-short EWMA filter in equation \eqref{eq:ewma3}. We develop analytical results for the single EWMA filter and then generalize results for the long-short filter using its linearity feature.

We present the daily return in equation \eqref{eq:ret} as follows:
\begin{equation} \label{eq:e2} 
	f_{t} = w_{t-1} r_{t} = \left[ \frac{\sigma_{\mathrm{target}}}{\sqrt{a}}\mathcal{L}^{(\nu, l)} (z_{t-1}) \right]\frac{ r_{t} }{\sigma_{t-1}}=   \frac{ l\sigma_{\mathrm{target}}}{\nu\sqrt{a}} \mathcal{E}_{t},
\end{equation}
where
\begin{equation} \label{eq:e3} 
	\mathcal{E}_{t} =    \nu z_{t}  \mathcal{L}^{(\nu)} (z_{t-1}).
\end{equation}

We consider the time period $T$ and the cumulative return over the period $T$, defined as the cumulative arithmetic P\&L rather than the compounded wealth return, as:
\begin{equation} \label{eq:e4} 
	F_{T} \equiv \sum^{T}_{t=1}  f_{t} = \frac{l\sigma_{\mathrm{target}} }{\nu\sqrt{a}} \sum^{T}_{t=1}  \mathcal{E}_{t}  = \frac{ l\sigma_{\mathrm{target}} }{\nu\sqrt{a}} E(T)
\end{equation}
where 
\begin{equation} \label{eq:e4a} 
	E(T)= \sum^{T}_{t=1}  \mathcal{E}_{t}.
\end{equation}

\begin{proposition} [Cumulative return]\label{pr:cumreturn} Assume that the EWMA series converges absolutely, for which the finite second moments of Assumption \ref{as:pr} are sufficient. The cumulative return of the European TF system over the period $T$ can be presented as follows:
	\begin{equation} \label{eq:cumreturn} 
		\begin{split}
			& F_{T} = \frac{ l\sigma_{\mathrm{target}} }{\nu\sqrt{a}} E(T)\\
			& E(T) =   \left(1-\nu \right)  \sum^{\infty}_{m=0} \nu^{m} \sum^{T}_{t=1}  \left( z_{t} -\overline{z}_{T} \right) \left(z_{t-m} -\overline{z}_{T} \right) -  \left(1-\nu \right) \sum^{T}_{t=1} \left(z_{t}-\overline{z}_{T} \right)^{2} + \nu T (\overline{z}_{T})^{2} + R_{T},
		\end{split}
	\end{equation}
	where the boundary term $R_{T}$ collects the interaction of the sample mean with the window shift of the lagged returns:
	\begin{equation} \label{eq:cumreturn_rt}
		R_{T} = \overline{z}_{T}\left(1-\nu \right) \sum^{\infty}_{m=1} \nu^{m} \Delta_{m}, \qquad \Delta_{m} = \sum^{T}_{t=1} z_{t-m} - T\overline{z}_{T}.
	\end{equation}
	and $\overline{z}_{T}$ is the sample mean:
	\begin{equation} \label{eq:mean} 
		\overline{z}_{T} = \frac{1}{T} \sum^{T}_{t=1} z_{t}.
	\end{equation}
	
	In equation \eqref{eq:cumreturn} for the cumulative return $E(T)$, the  terms are the sample autocorrelation term, the sample variance term, the squared sample mean, and residual, respectively.
	
\end{proposition}

Proof is provided in Appendix \ref{sec:cumreturn}.

\subsection{Cumulative Return through Sample Autocorrelation}

We introduce the sample autocorrelation of volatility-normalized returns $\{z_{t}\}$ realized over the period $T$ as follows:
\begin{equation} \label{eq:e11} 
	\begin{split}
		\hat{\rho}_{T}(m) & = \frac{\hat{\gamma}_{T}(m)}{\hat{\gamma}_{T}(0)},\\
		\hat{\gamma}_{T}(0)  & = \frac{1}{T} \sum^{T}_{t=1} \left( z_{t} -\overline{z}_{T}\right)^{2}, \\
		\hat{\gamma}_{T}(m)  & = \frac{1}{T} \sum^{T}_{t=1} \left(z_{t}-\overline{z}_{T}\right) \left(z_{t-m}-\overline{z}_{T}\right).
	\end{split}
\end{equation}

Two conventions and one boundary term make the identity in Proposition \ref{pr:cumreturn} exact. First, the filter uses the observations $z_{t-m}$ for all $m\geq 0$, so the data extend before the sample start by the memory of the filter. Second, the estimator $\hat{\gamma}_{T}(m)$ in equation \eqref{eq:e11} sums over $t=1,\ldots,T$ and uses the pre-sample observations, while the conventional estimator sums over $t=m+1,\ldots,T$, and the ratio $\hat{\rho}_{T}(m)$ normalizes by $\hat{\gamma}_{T}(0)$ alone, so it is an identity-specific lag product and need not lie in $[-1,1]$. Third, the boundary term $R_{T}$ collects the cross products of the sample mean with the shifted windows, because the sample mean centers the window $t=1,\ldots,T$ and not its lagged copies, and it vanishes when $\overline{z}_{T}=0$. Under Assumption \ref{as:pr}, $\mathbb{E}\left|\Delta_{m}\right|$ grows at most linearly in $m$, so $R_{T}$ is of the order $\spn/T$ in expected magnitude relative to $E(T)$ at a nonzero expected return, and the centered decomposition without $R_{T}$ is an approximation for $T \gg \spn$. Truncating the filter at the sample start adds a separate term of the same order. The population statement in Corollary \ref{col:cumreturn} is free of all three effects under Assumption \ref{as:pr}.

\begin{corollary} The sample cumulative return $\hat{F}_{T}$ of equation \eqref{eq:e4} follows from the sample estimators in equation \eqref{eq:e11}:
	\begin{equation} \label{eq:e12} 
		\begin{split}
			\hat{E}(T)& =   \left(1-\nu \right)   T  \hat{\gamma}_{T}(0)\left(  \left[ \sum^{\infty}_{m=0} \nu^{m} \hat{\rho}_{T}(m) -  1 \right] + \frac{\nu}{1-\nu}\frac{(\overline{z}_{T})^{2}}{\hat{\gamma}_{T}(0)}\right) + R_{T},
		\end{split}
	\end{equation}
	\begin{equation} \label{eq:e14} 
		\begin{split}
			\hat{F}_{T} &   =\left( \frac{l\sigma_{\mathrm{target}} }{\sqrt{a}} \frac{\left(1-\nu \right)}{\nu}   T \right)    \hat{\gamma}_{T}(0)  \left( \left[ \sum^{\infty}_{m=0} \nu^{m} \hat{\rho}_{T}(m) -  1 \right]  + \frac{\nu}{1-\nu}\frac{(\overline{z}_{T})^{2}}{\hat{\gamma}_{T}(0)} \right) + \frac{l\sigma_{\mathrm{target}}}{\nu\sqrt{a}}\, R_{T}.
		\end{split}
	\end{equation}
	
\end{corollary}

The boundary term $R_{T}$ enters equation \eqref{eq:e14} scaled by the prefactor of equation \eqref{eq:e4}. The cumulative return is proportional to the realized variance $\hat{\gamma}_{T}(0)$ and to the sum of two leading drivers: the moment of realized autocorrelation weighted by $\nu^{m}$, $\sum^{\infty}_{m=0} \nu^{m} \hat{\rho}_{T}(m) -  1$, and the square of the daily sample Sharpe ratio, $(\overline{z}_{T})^{2}/\hat{\gamma}_{T}(0)$, whose annualized counterpart carries the factor $a$. For volatility-normalized returns, $\hat{\gamma}_{T}(0)$ is close to one, so the realized autocorrelation is the key driver of the cumulative return. The $\nu$-weighted sum of the sample autocorrelations at positive lags must be positive for the strategy to profit at zero drift. The boundary term $R_{T}$ in equation \eqref{eq:cumreturn_rt} is bilinear in the sample mean $\overline{z}_{T}$ and the shifted-window discrepancies $\Delta_{m}$.

\subsection{Turnover Analysis}

The average volatilities of futures returns range from $2\%$ for short rates to above $30\%$ for most commodities, so, after volatility-normalized position sizing, the realized turnover of managed futures programs may be dominated by a few low-volatility contracts with inflated notional exposures. We introduce the volatility-normalized turnover $U_{t}$ as a stable measure of the turnover of futures contracts.

\begin{definition}[Volatility-normalized turnover] We define $U_{t}$ as follows:
	\begin{equation} \label{eq:tur1}  
		U_{t}  \equiv \sqrt{a}\sigma_{t}\left| w_{t} -  w_{t-1}\right|,
	\end{equation}
	where $\sigma_{t}$ is the volatility of daily returns and $a$ is the annualization factor.
\end{definition}

\begin{proposition}\label{pr:turnover1} Assuming serially independent Gaussian $z_{t}$ with unit variance $\vartheta_{0}=1$ and the specification of the weight in equation \eqref{eq:eu3} for the European TF system with the variance-preserving EWMA filter in equation \eqref{eq:ewma2}, the expected annualized signal turnover $U^{(signal)}_{t} = \sigma_{\mathrm{target}}\left|S_{t}-S_{t-1}\right|$ is given by:
	\begin{equation} \label{eq:tur2} 
		\mathbb{E}\left[ a U^{(signal)}_{t} \right]    =  \frac{2a}{\sqrt{\pi}}\sigma_{\mathrm{target}} \sqrt{1-\nu}.
	\end{equation}
\end{proposition}

For the European TF system, the leading order of the expected turnover is inversely proportional to the square root of the EWMA filter span. The formula is free of process parameters only under the independence assumption, because the variance of the signal increment loads on the autocorrelation function when the returns are serially dependent. The full turnover $U_{t}$ in equation \eqref{eq:tur1} contains both the signal-increment and the volatility-update components, which the absolute value does not separate additively, so the closed forms serve as the independence-based signal-turnover proxy for $\mathbb{E}[aU_{t}]$.

\begin{proposition}\label{pr:turnover2}
	
	Assuming serially independent Gaussian $z_{t}$ with unit variance $\vartheta_{0}=1$ and the long-short EWMA filter in equation \eqref{eq:ewma3} for the specification of the weight in equation \eqref{eq:eu3}, the expected annualized signal turnover of the long-short filter is given by:
	\begin{equation} \label{eq:tur_ls} 
		\mathbb{E}\left[ a U^{(signal)}_{t} \right]   =  \frac{2a}{\sqrt{\pi}}\sigma_{\mathrm{target}} \sqrt{\zeta},
	\end{equation}
	where $\zeta$ is given through the normalization $q$ of equation \eqref{eq:ewma3a} by
	\begin{equation*}
		\zeta = \frac{q^{2}}{2}\left( \frac{1-\nu_{1}}{1+\nu_{1}}  +  \frac{1-\nu_{2}}{1+\nu_{2}}  - \frac{2\left( 1-\nu_{1}\right) \left( 1-\nu_{2}\right)}{1-\nu_{1}\nu_{2}}\right).
	\end{equation*}
	
\end{proposition}

Proofs are given in Appendix \ref{sec:ta1} and \ref{sec:ta2}. The closed forms are exact for the signal turnover under the stated assumptions and serve as leading-order proxies for the full turnover. The derivation retains the signal-increment component, omits the volatility-update component, and applies the Gaussian mean-absolute-value formula to the signed signal increment, which is exact for Gaussian innovations. Across the white-noise, AR-1, and ARFIMA exhibits, the proxy stays within $8.4\%$ of the signal-level MC turnover, with the largest gap at the shortest span, and the residual gap arises from the serial dependence, the non-Gaussianity, and the volatility-normalization pipeline. Within the full pipeline at the volatility span of $33$ days, the single-filter turnover exceeds the proxy by about $4\%$. For the long-short filter, the small signal increment amplifies the volatility-update term, and the proxy understates the pipeline turnover by a factor of $1.6$ to $2.3$ across our processes, and by at most $29\%$ at the volatility span of $250$ days (Section \ref{ssec:sr_mc}).

At the spans of the empirical section, the signal-turnover proxy gives $393\%$ per year for the single 250-day filter and $88\%$ for LS(250,20), which is a significant reduction for paractical purposes, because the two legs carry the identical loading $q$ on the contemporaneous innovation, which cancels exactly in the signal increment.

\subsection{Expected Cumulative Return}

We state the assumption on the volatility-normalized returns that underlies our analytical population-based results on the performance of the European TF system.
\begin{assumption}[Characteristics of volatility-normalized returns $z_{t}$]\label{as:pr} ~
	We consider the volatility-normalized returns $z_{t}$ in equation \eqref{eq:eu1} with the following population-based characteristics for the expected mean, the variance, and the autocorrelation function, respectively, under statistical measure $\mathbb{P}$:
	\begin{equation} \label{eq:e15} 
		\begin{split}
			& \mathbb{E}^{\mathbb{P}} \left[ z_{t} \right] = \mu,\\
			& \Var^{\mathbb{P}} \left[ z_{t} \right] = \vartheta > 0,\\
			& \mathbb{E}^{\mathbb{P}} \left[ \frac{1}{\vartheta} (z_{t}-\mu)(z_{t-m}-\mu) \right] = \rho(m), \ m = 1, 2, \ldots.
		\end{split}
	\end{equation}
	We set $\rho(0)=1$, extending the autocorrelation function to the zero lag.
\end{assumption}

Assumption \ref{as:pr} requires no specification of the raw returns $r_{t}$, because the systems trade the volatility-normalized returns $z_{t}$ and every population formula depends on the process only through the moments of $z_{t}$. When the volatility is constant, the autocorrelation functions of $z_{t}$ and $r_{t}$ coincide, and the annualized mean of $z_{t}$ equals the annualized Sharpe ratio of the instrument. More generally, if $r_{t} = \sigma_{t-1} x_{t}$ with a standardized stationary process $x_{t}$ independent of the predictable volatility, then $\rho_{r}(m) = \rho_{x}(m)\,\mathbb{E}[\sigma_{t}\sigma_{t-m}]/\mathbb{E}[\sigma^{2}_{t}]$, so time-varying volatility dampens the autocorrelation of the raw returns, and the normalization recovers it. Under the estimated EWMA volatility, the normalization keeps $z_{t}$ close to a stationary process with unit variance, and its autocorrelation is a dampened version of the raw-return autocorrelation, which Section \ref{ssec:sr_mc} quantifies. The empirical inputs of Section \ref{sec:attr} are estimated directly from $z_{t}$, so the formulas are applied to the process that they assume.

\begin{corollary}[Expected cumulative return of European TF system]\label{col:cumreturn} Under Assumption \ref{as:pr}, the expected cumulative return of equation \eqref{eq:e4} under $\mathbb{P}$ is given by:
	\begin{equation} \label{eq:ft1} 
		\bar{F}_{T} \equiv \mathbb{E}^{\mathbb{P}} \left[ F_{T}\right] = \left( \frac{l\sigma_{\mathrm{target}} }{\sqrt{a}}  \frac{1-\nu }{\nu}    \vartheta T \right)   \left[ \sum^{\infty}_{m=0} \nu^{m}\rho(m) -  1 \right]  +  \left( \frac{l\sigma_{\mathrm{target}} }{\sqrt{a}}  T \right)   \mu^{2},  
	\end{equation}
	where the loading $l$ is specified in equation \eqref{eq:e1}.

	For the long-short EWMA filter, the expected cumulative return is given by:
	\begin{equation} \label{eq:ft2} 
		\begin{split}
			\bar{F}_{T}  &   = \left( \frac{\sigma_{\mathrm{target}} }{\sqrt{a}}  \vartheta T\right) \left( l_{1}\frac{1-\nu_{1} }{\nu_{1}}  \left[ \sum^{\infty}_{m=0} \nu^{m}_{1}\rho(m) -  1 \right]  -  l_{2}\frac{1-\nu_{2} }{\nu_{2}}  \left[ \sum^{\infty}_{m=0} \nu^{m}_{2}\rho(m) -  1 \right]\right)\\
			& +  \left( \frac{\sigma_{\mathrm{target}} }{\sqrt{a}}  \mu^{2}   T \right)    \left(l_{1}-l_{2}\right),
		\end{split}
	\end{equation}
	where the loadings $l_{1}$ and $l_{2}$ are specified in equation \eqref{eq:ewma3a}.
	
\end{corollary}

Further, we consider the annual performance period with $T=a$. We set the variance of volatility-normalized returns to unit variance $\vartheta=1$. Since we apply volatility-normalization to raw returns, the ``natural'' volatility of the returns process is an auxiliary variable and it can be set to one. In equation \eqref{eq:e15}, $\mu$ stands for the population mean of the daily volatility-normalized return. The annual population mean of volatility-normalized return denoted by $\mu^{z}_{an}$ is obtained by applying the annualization factor $a$:
\begin{equation} \label{eq:er2a} 
	\mu^{z}_{an} \equiv \frac{a \mu }{\sqrt{a}} = \sqrt{a} \mu. 
\end{equation}
Here, $\mu^{z}_{an}$ can be interpreted as the Sharpe ratio of the process with unit variance.

\begin{corollary}[Expected annual return of European TF system] Using Corollary \ref{col:cumreturn} with $T=a$, $\vartheta=1$, $\mu^{z}_{an}=\sqrt{a} \mu$, we obtain:
	\begin{equation} \label{eq:er2} 
		\begin{split}
			\bar{F}_{1y} = h_{1y} \left[ \sum^{\infty}_{m=0} \nu^{m}\rho(m) -  1 \right]  +  \left( \frac{l\sigma_{\mathrm{target}} }{\sqrt{a}} \right)  (\mu^{z}_{an})^{2}  , \ h_{1y} = l\sigma_{\mathrm{target}}\sqrt{a}\frac{1-\nu }{\nu},
		\end{split}
	\end{equation}
	where the loading $l$ is specified in equation \eqref{eq:e1} and $\rho(m)$ is the population autocorrelation function of the volatility-normalized returns in Assumption \ref{as:pr}.

	For the long-short EWMA filter, the expected annual return is given by:
	\begin{equation} \label{eq:er3} 
		\begin{split}
			\bar{F}_{1y}  &   =\widetilde{h}_{1} \left[ \sum^{\infty}_{m=0} \nu^{m}_{1}\rho(m) -  1 \right]  -  \widetilde{h}_{2}  \left[ \sum^{\infty}_{m=0} \nu^{m}_{2}\rho(m) -  1 \right] + \left( \frac{\sigma_{\mathrm{target}} }{\sqrt{a}}  (\mu^{z}_{an})^{2} \right)    \left(l_{1}-l_{2}\right)\\
			& \widetilde{h}_{1} = \sigma_{\mathrm{target}} \sqrt{a}  l_{1}\frac{1-\nu_{1} }{\nu_{1}}, \ \widetilde{h}_{2} =  \sigma_{\mathrm{target}} \sqrt{a} l_{2}\frac{1-\nu_{2} }{\nu_{2}}, 
		\end{split}
	\end{equation}
	where the loadings $l_{1}$ and $l_{2}$ are specified in equation \eqref{eq:ewma3a}.
	
\end{corollary}

\section{Sharpe Ratio of European TF System}\label{sec:sr}

Corollary \ref{col:cumreturn} of Section \ref{sc:eft} links the expected return of the European TF system to the autocorrelation function and the drift of the volatility-normalized returns. The expected return alone, however, cannot rank filter spans, because the risk of the strategy changes with the span as well. In this section, we derive the variance of the daily strategy return in closed form and obtain the annualized Sharpe ratio. We define the annualized gross Sharpe ratio\footnote{The daily returns $f_{t}$ are serially dependent because the signal is persistent. We define the annualized Sharpe ratio on the moments of daily returns, which is the convention of our backtests and of industry reporting for liquid managed futures. The horizon-based Sharpe ratio of \cite{Lo2002}, defined on the variance of aggregated returns, differs from our definition when returns are autocorrelated. Under the white noise with zero mean assumption, the covariances of daily strategy returns vanish at all lags and the two definitions coincide. For persistent signals, the daily strategy returns are typically positively autocorrelated, so the daily-moment annualization overstates the horizon-based Sharpe ratio.} of a strategy with daily returns $f_{t}$ as
\begin{equation} \label{eq:sr_def}
	SR = \sqrt{a} \, \frac{\mathbb{E}^{\mathbb{P}}[f_{t}]}{\sqrt{\Var^{\mathbb{P}}[f_{t}]}},
\end{equation}
where $a$ is the number of trading days per year, and $\mathbb{E}^{\mathbb{P}}[f_{t}]$ and $\Var^{\mathbb{P}}[f_{t}]$ are the mean and the variance of the daily returns under $\mathbb{P}$. Returns on futures contracts are excess returns, so we do not subtract a risk-free rate. The daily return in equation \eqref{eq:e2} is a product of the volatility-normalized return $z_{t}$ and the lagged filter value, and both variables are linear in past returns. The variance of such a product is available in closed form through an extension of the Isserlis theorem. As a result, the Sharpe ratio depends on the autocorrelation moments of the returns and on the excess kurtosis of the innovations.

The gross Sharpe ratio excludes trading costs. We define the net Sharpe ratio under a proportional cost $c$ per unit of the volatility-normalized turnover $U_{t}$ in equation \eqref{eq:tur1}:
\begin{equation} \label{eq:sr_net_def}
	SR^{net} = \sqrt{a} \, \frac{\mathbb{E}^{\mathbb{P}}[f_{t}] - c\,\mathbb{E}^{\mathbb{P}}[U_{t}]}{\sqrt{\Var^{\mathbb{P}}[f_{t}]}} = SR - c\, \frac{\mathbb{E}^{\mathbb{P}}[a U_{t}]}{\sqrt{a}\,\sqrt{\Var^{\mathbb{P}}[f_{t}]}}.
\end{equation}
The second form divides the expected annual cost by the annualized volatility of the strategy. The denominator keeps the gross volatility, so the ratio is a cost-adjusted gross-risk ratio rather than the Sharpe ratio of the net return $f_{t}-c\,U_{t}$, whose variance adds the covariance term $-2c\,\Cov[f_{t},U_{t}]$ of the first order in $c$ and the term $c^{2}\Var[U_{t}]$ of the second order. The turnover input applies the signal-turnover proxy of Propositions \ref{pr:turnover1} and \ref{pr:turnover2}: the independence-based Gaussian closed form for the signal component of the turnover. Serial dependence changes the signal turnover itself, and the full turnover contains the volatility-update component as well, so the net results are leading-order approximations under Assumption \ref{as:pr}, while the gross results are exact. We refer to the resulting closed forms as the net Sharpe ratio under the signal-turnover proxy, and we write them with the approximation sign. The net Sharpe ratio remains independent of the volatility target, because the turnover and the strategy returns are both proportional to $\sigma_{\mathrm{target}}$.

\subsection{Autocorrelation Generating Function}

Under Assumption \ref{as:pr}, we define the autocorrelation generating function $\Phi_{\nu}$ and its centered version $\Psi_{\nu}$ as follows:
\begin{equation} \label{eq:sr_phi}
	\Phi_{\nu} = \sum^{\infty}_{m=0} \nu^{m}\rho(m), \qquad \Psi_{\nu} = \Phi_{\nu} - 1 = \sum^{\infty}_{m=1} \nu^{m}\rho(m).
\end{equation}
In this notation, the autocorrelation term in the expected return in equation \eqref{eq:er2} equals $\Psi_{\nu}$: zero for white noise and $\nu\phi/(1-\nu\phi)$ for the AR-1 process. For ARFIMA processes, we evaluate $\Psi_{\nu}$ numerically from the autocorrelation function in equations \eqref{eq:lm3} and \eqref{eq:lm4} truncated at $2000$ lags, where the weights $\nu^{m}$ bound the truncation remainder relative to $\Psi_{\nu}$ below $10^{-4}$ at the longest span for the fractional orders $d \leq 0.1$ used in this paper.

The generating function admits a spectral representation. By the Herglotz theorem, the autocorrelation function of any stationary process is the Fourier coefficient sequence of a normalized spectral measure $F$ on $[-\pi, \pi]$, which requires only the positive definiteness of the autocorrelation function, and the symmetrized sum gives:
\begin{equation} \label{eq:sr_spectral}
	2\Phi_{\nu} - 1 = \int^{\pi}_{-\pi} \frac{1-\nu^{2}}{1 - 2\nu\cos\lambda + \nu^{2}}\, dF(\lambda),
\end{equation}
where the integrand is the Poisson kernel with the concentration parameter $\nu$. The expected return of the TF system is therefore a Poisson-kernel reading of the spectrum of the volatility-normalized returns. The system profits at zero drift if and only if the kernel-weighted spectral mass exceeds one, so trend-following alpha is excess spectral mass at low frequencies. The span sets the bandwidth of the kernel, so short spans read the whole spectrum while long spans read the spectral mass near zero, where the long memory resides. The long-short filter reads the spectrum through a signed contrast of two Poisson-kernel readings. The contrast is not a band-pass window, because the level signal retains a nonzero gain at frequency zero under the normalization in equation \eqref{eq:ewma3a}. The one-day signal increment has zero gain at frequency zero, so the increment, and not the level, acts as a band-pass difference of the two filters.

\subsection{Variance of the Daily Return}

The expected return in Corollary \ref{col:cumreturn} requires only the first two moments of the process. The variance of the daily return involves fourth moments, so we need an assumption on the distribution of the volatility-normalized returns.

\begin{assumption}[Linear process for volatility-normalized returns]\label{as:lp}
	The volatility-normalized returns $\{z_{t}\}$ satisfy Assumption \ref{as:pr} and admit a linear representation in the form of a Wold decomposition (see \citep[Chapter 4]{Hamilton1994}):
	\begin{equation} \label{eq:sr_lp}
		z_{t} = \mu + \sqrt{\vartheta}\sum^{\infty}_{s=0}\psi_{s}\epsilon_{t-s}, \qquad \sum^{\infty}_{s=0}\psi^{2}_{s}=1,
	\end{equation}
	where the innovations $\{\epsilon_{t}\}$ are independent and identically distributed with zero mean, unit variance, zero third moment, and finite excess kurtosis $\kappa=\mathbb{E}[\epsilon^{4}_{t}]-3$.
\end{assumption}

The moving-average weights $\psi_{s}$ determine the autocorrelation function through $\rho(m)=\sum^{\infty}_{s=0}\psi_{s}\psi_{s+m}$. The assumption does not restrict the serial dependence between the signal and the return, because both variables load on the common innovations. The zero third moment removes cross terms between the drift and the squared innovations, and these terms are of the order of the drift times the skewness of the innovations. The excess kurtosis enters the variance of the strategy return through a single loading, which Proposition \ref{prop:dailyvar} defines. Jointly Gaussian returns correspond to the case $\kappa=0$, because a stationary Gaussian process is a linear process with Gaussian innovations\footnote{The Wold decomposition separates a deterministic component, which we exclude. The processes of Section~\ref{sec:sp} are linear by construction.}.

Volatility normalization can remove most of the heavy tails and the asymmetry of raw returns, so the restrictions on the innovations are mild. Under GARCH-type dynamics, the EWMA variance recursion approximates the conditional variance, so the volatility-normalized returns are approximately the standardized innovations: in a GARCH(1,1) simulation with a raw excess kurtosis of $7$, the pipeline $z_{t}$ has an excess kurtosis of $0.06$ and a first-lag autocorrelation below $0.001$. Positive excess kurtosis adds a positive term to the variance of the strategy return and reduces the absolute Sharpe ratio, lowering a positive ratio and moving a negative ratio toward zero. The expected return in Corollary \ref{col:cumreturn} is unaffected, because the first moment is distribution-free. We quantify the impact of the full volatility-normalization pipeline, which the linear representation does not capture, in Section \ref{ssec:sr_mc}.

\begin{lemma}[Variance of a product of variables linear in common innovations]\label{lem:isserlis}
	Let $X=\mu_{X}+\sum^{\infty}_{s=0}a_{s}\epsilon_{t-s}$ and $Y=\mu_{Y}+\sum^{\infty}_{s=0}b_{s}\epsilon_{t-s}$ be linear in the innovations of Assumption \ref{as:lp}, with variances $\sigma^{2}_{X}$ and $\sigma^{2}_{Y}$ and covariance $\sigma_{XY}$. The variance of the product is given by:
	\begin{equation} \label{eq:sr_isserlis}
		\Var\left[ XY \right] = \sigma^{2}_{X}\sigma^{2}_{Y} + \sigma^{2}_{XY} + \mu^{2}_{X}\sigma^{2}_{Y} + \mu^{2}_{Y}\sigma^{2}_{X} + 2\mu_{X}\mu_{Y}\sigma_{XY} + \kappa\sum^{\infty}_{s=0}a^{2}_{s}b^{2}_{s}.
	\end{equation}
	For $\kappa=0$, equation \eqref{eq:sr_isserlis} reduces to the Isserlis theorem for jointly Gaussian variables.
\end{lemma}
Proof is given in Appendix \ref{sec:isserlis}.

We apply Lemma \ref{lem:isserlis} to $X=z_{t}$ and $Y=\mathcal{L}^{(\nu)}(z_{t-1})$, because the daily return in equation \eqref{eq:e2} is proportional to their product. The filter is a linear combination of past returns, so its moments follow from Assumption \ref{as:pr} by direct summation.

\begin{proposition}[Mean and variance of the daily return]\label{prop:dailyvar}
	Under Assumption \ref{as:lp}, the mean and the variance of the daily return $f_{t}$ of the European TF system in equation \eqref{eq:e2} are given by:
	\begin{equation} \label{eq:sr_moments}
		\begin{split}
			\mathbb{E}\left[ f_{t} \right] & = \frac{l\sigma_{\mathrm{target}}}{\sqrt{a}} \left( \vartheta A_{\nu} + \mu^{2} \right), \\
			\Var\left[ f_{t} \right] & = \left(\frac{l\sigma_{\mathrm{target}}}{\sqrt{a}}\right)^{2} \vartheta \left( \vartheta \left( B_{\nu} + A^{2}_{\nu} + \kappa K_{\nu}\right) + \mu^{2}\left( 1 + B_{\nu} + 2A_{\nu} \right) \right),
		\end{split}
	\end{equation}
	where the autocorrelation loadings $A_{\nu}$, $B_{\nu}$ and the kurtosis loading $K_{\nu}$ are defined by:
	\begin{equation} \label{eq:sr_ab}
		A_{\nu} = \frac{1-\nu}{\nu}\Psi_{\nu}, \qquad B_{\nu} = \frac{1-\nu}{1+\nu}\left( 1 + 2\Psi_{\nu} \right),
	\end{equation}
	\begin{equation} \label{eq:sr_k}
		K_{\nu} = \sum^{\infty}_{s=1} \psi^{2}_{s}\, g^{2}_{s-1}, \qquad g_{u} = \left(1-\nu\right)\sum^{u}_{m=0}\nu^{m}\psi_{u-m}.
	\end{equation}
	Under a proportional cost $c$ per unit of the volatility-normalized turnover $U_{t}$ in equation \eqref{eq:tur1}, the expected net daily return is given at the leading order by:
	\begin{equation} \label{eq:sr_netmean}
		\mathbb{E}\left[ f_{t} - c U_{t} \right] \approx \mathbb{E}\left[ f_{t} \right] - \frac{2c}{\sqrt{\pi}}\,\sigma_{\mathrm{target}}\sqrt{(1-\nu)\vartheta},
	\end{equation}
	where the expected turnover applies the signal-turnover proxy of Proposition \ref{pr:turnover1}, scaled by $\sqrt{\vartheta}$ from its unit-variance statement, because the signal increment is linear in $z_{t}$.
\end{proposition}

Proof is given in Appendix \ref{sec:dailyvar}.

The coefficient $g_{u}$ is the loading of the lagged filter on the innovation $\epsilon_{t-1-u}$. Only innovations that enter both the return $z_{t}$ and the lagged filter contribute to $K_{\nu}$, so the kurtosis correction requires serial dependence. Under white noise, the return loads only on the current innovation and the loading $K_{\nu}$ vanishes. 

\begin{corollary}[Moments and Sharpe ratio of the long-short system]\label{col:sr_ls} The result extends to the long-short EWMA filter in equation \eqref{eq:ewma3} by linearity, because the signal remains a linear combination of past returns. The signal $S_{t-1}=l_{1}\mathcal{L}^{(\nu_{1})}(z_{t-1})-l_{2}\mathcal{L}^{(\nu_{2})}(z_{t-1})$, which restates equation \eqref{eq:ewma3} through the raw filters, has the following moments under Assumption \ref{as:pr}:
	\begin{equation} \label{eq:sr_ls}
		\begin{split}
			& \mathbb{E}\left[ S_{t-1} \right] = \left(l_{1}-l_{2}\right)\mu, \qquad
			\Cov\left[ z_{t}, S_{t-1} \right] = \vartheta \left( l_{1} \frac{1-\nu_{1}}{\nu_{1}}\Psi_{\nu_{1}} - l_{2} \frac{1-\nu_{2}}{\nu_{2}}\Psi_{\nu_{2}} \right), \\
			&  \Var\left[ S_{t-1} \right] = \vartheta \left( l^{2}_{1} B_{\nu_{1}} + l^{2}_{2} B_{\nu_{2}} - 2l_{1}l_{2} \frac{\left(1-\nu_{1}\right)\left(1-\nu_{2}\right)\left(1+\Psi_{\nu_{1}}+\Psi_{\nu_{2}}\right)}{1-\nu_{1}\nu_{2}} \right),
		\end{split}
	\end{equation}
	where the cross-covariance of the two filters follows by the same double-sum manipulation as in equation \eqref{eq:sr_varL}. Under Assumption \ref{as:lp}, substituting these moments into Lemma \ref{lem:isserlis} gives the Sharpe ratio of the long-short system. The kurtosis loading of the long-short system replaces $g_{u}$ in equation \eqref{eq:sr_k} with $l_{1}g^{(\nu_{1})}_{u}-l_{2}g^{(\nu_{2})}_{u}$, where $g^{(\nu_{i})}_{u}$ denotes the loading $g_{u}$ evaluated at the span $\nu_{i}$. The net Sharpe ratio of the long-short system replaces the expected turnover with the signal-turnover proxy of Proposition \ref{pr:turnover2}.
\end{corollary}

\subsection{Annualized Sharpe Ratio}

\begin{corollary}[Sharpe ratio of European TF system]\label{col:sr}
	Under Assumption \ref{as:lp} with $\vartheta=1$ and the annualized mean $\mu^{z}_{an}=\sqrt{a}\mu$, the annualized Sharpe ratio of the European TF system is given by:
	\begin{equation} \label{eq:sr_main}
		SR = \frac{ \sqrt{a}\, A_{\nu} + \left(\mu^{z}_{an}\right)^{2}/\sqrt{a} }{ \sqrt{ B_{\nu} + A^{2}_{\nu} + \kappa K_{\nu} + \left( \left(\mu^{z}_{an}\right)^{2}/a\right) \left( 1 + B_{\nu} + 2A_{\nu} \right) } }.
	\end{equation}
	The Sharpe ratio does not depend on the volatility target $\sigma_{\mathrm{target}}$ and does not depend on the filter loading $l$, because the factor $l\sigma_{\mathrm{target}}/\sqrt{a}$ multiplies the mean and the volatility of $f_{t}$ equally. For $\kappa=0$, equation \eqref{eq:sr_main} gives the Sharpe ratio under jointly Gaussian returns. Dividing the proxy turnover of Proposition \ref{pr:turnover1} by the annualized volatility of Proposition \ref{prop:dailyvar} and using the identity $l\sqrt{1-\nu}=\sqrt{1+\nu}$, the net Sharpe ratio of equation \eqref{eq:sr_net_def} follows at the leading order as:
	\begin{equation} \label{eq:sr_net}
		SR^{net} \approx SR - \frac{2ac}{\sqrt{\pi}}\, \frac{1-\nu}{\sqrt{ (1+\nu) \left( B_{\nu} + A^{2}_{\nu} + \kappa K_{\nu} + \left( \left(\mu^{z}_{an}\right)^{2}/a\right) \left( 1 + B_{\nu} + 2A_{\nu} \right) \right)} }.
	\end{equation}
	Under zero-drift white noise, the bracket reduces to $B_{\nu}=(1-\nu)/(1+\nu)$ and the cost drag equals $(2ac/\sqrt{\pi})\sqrt{1-\nu}$. The net Sharpe ratio remains independent of $\sigma_{\mathrm{target}}$ and depends on the filter only through the span.
\end{corollary}

Equation \eqref{eq:sr_main} reduces the TF design problem for a single EWMA filter to the choice of the span, which can be selected from the autocorrelation structure and the excess kurtosis of the returns. The white-noise result in equation \eqref{eq:sr_wn} requires no distributional assumption beyond serial independence and a finite second moment. Under white noise, $z_{t}$ is independent of the signal built from $z_{t-1}, z_{t-2}, \ldots$, so the variance of the product factorizes into moments of order two. For autocorrelated returns, the excess kurtosis reduces the absolute Sharpe ratio through the term $\kappa K_{\nu}$ in Proposition \ref{prop:dailyvar}. For the AR-1 process of Section \ref{sec:sp}, the correction is below the Monte Carlo confidence intervals in Figure \ref{fig:expected_return_ar}.

Table \ref{tab:notation} summarizes the notation of the theoretical sections.

\begin{table}[!htb]
	\footnotesize
	\centering
	\caption{Notation of the theoretical sections}\label{tab:notation}
	\begin{tabular}{llc}
		\hline
		Symbol & Meaning & Equation \\
		\hline
		$r_{t}$ & daily return of the continuous futures contract & \eqref{eq:int2} \\
		$\sigma_{t}$ & estimate of the daily volatility of $r_{t}$ & \eqref{eq:ewmavol} \\
		$z_{t}$ & volatility-normalized return, $z_{t}=r_{t}/\sigma_{t-1}$ & \eqref{eq:eu1} \\
		$a$ & annualization factor, $a=260$ trading days & \eqref{eq:eu3} \\
		$\nu$ & filter smoothing parameter, $\nu=1-2/(\spn+1)$ & \eqref{eq:i2} \\
		$l$ & loading of the variance-preserving filter, $l=\sqrt{(1+\nu)/(1-\nu)}$ & \eqref{eq:e1} \\
		$\mathcal{L}^{(\nu,l)}$, $\widetilde{\mathcal{L}}^{(\nu)}$ & EWMA filter and its variance-preserving version & \eqref{eq:ewma1}, \eqref{eq:ewma2} \\
		$\widetilde{\mathcal{LS}}^{(\nu_{1},\nu_{2})}$ & variance-preserving long-short EWMA filter & \eqref{eq:ewma3} \\
		$S_{t}$, $w_{t}$ & signal and portfolio weight of the system & \eqref{eq:eu2}, \eqref{eq:eu3} \\
		$f_{t}$, $F_{T}$ & daily return and cumulative return of the system & \eqref{eq:ret}, \eqref{eq:e4} \\
		$\mu$, $\vartheta$ & mean and variance of $z_{t}$ & \eqref{eq:e15} \\
		$\rho(m)$, $\hat{\rho}_{T}(m)$ & population autocorrelation and identity-specific lag product of $z_{t}$ & \eqref{eq:e15}, \eqref{eq:e11} \\
		$\Phi_{\nu}$, $\Psi_{\nu}$ & autocorrelation generating function and its centered version & \eqref{eq:sr_phi} \\
		$A_{\nu}$, $B_{\nu}$ & autocorrelation loadings of Proposition \ref{prop:dailyvar} & \eqref{eq:sr_ab} \\
		$\psi_{s}$, $\epsilon_{t}$ & moving-average weights and innovations of the linear process & \eqref{eq:sr_lp} \\
		$\kappa$, $K_{\nu}$ & excess kurtosis of the innovations and kurtosis loading & \eqref{eq:sr_k} \\
		$\mu^{z}_{an}$ & annualized drift of $z_{t}$, $\mu^{z}_{an}=\sqrt{a}\mu$ & \eqref{eq:er2a} \\
		$U_{t}$, $U^{(signal)}_{t}$ & volatility-normalized turnover and its signal component & \eqref{eq:tur1}, \eqref{eq:tur2} \\
		$\zeta$ & turnover loading of the long-short filter & \eqref{eq:tur_ls} \\
		$c$, $SR^{net}$ & cost per unit of volatility-normalized turnover and net Sharpe ratio & \eqref{eq:sr_net_def} \\
		$\varsigma(T)$, $w_{j}$, $R(h)$ & skewness of the $T$-day return, signal weights, and signal autocorrelation & \eqref{eq:skew_t}, \eqref{eq:skew_gen} \\
		\hline
	\end{tabular}
\end{table}

\section{Expected Return and Sharpe Ratio under Specific Processes}\label{sec:sp} 

In this section, we apply different stochastic processes to derive analytic solutions for the expected return and the Sharpe ratio of the trend-following strategy. We also compare analytical results with estimates obtained using Monte Carlo simulations\footnote{The Python package \texttt{trendfollowing}, which reproduces all analytical results and exhibits of this paper, is available at \url{https://github.com/ArturSepp/TrendFollowingSystems}.}.

For verification of analytical results, we generate $1000$ Monte Carlo paths of $50$ years of daily returns ($13000$ observations per path). We simulate the daily returns $\{r_{t}\}$ assuming the time step $dt=1/a$, with $a=260$, and the noise process with variance equal to $dt$. We compute volatility-normalized returns $\{z_{t}\}$ in equation \eqref{eq:eu1} using the EWMA volatility estimator in equation \eqref{eq:ewmavol} with $\spn=33$\footnote{The confidence intervals improve with a longer volatility span. We use the shorter span in alignment with the implementation of TF strategies and industry practice (\cite{Bouchaud2017}, \cite{Robertson2025}).} and with the volatility target $\sigma_{\mathrm{target}}=15\%$. The processes of this section specify the raw returns with constant innovation variance, so the volatility-normalized returns inherit the autocorrelation function and the drift of the process up to the estimation effect of the volatility filter, which Section \ref{ssec:sr_mc} quantifies. The gross results assume zero trading costs. For the net Sharpe ratio, we apply the cost of $c=20$bp per unit of the volatility-normalized turnover\footnote{This corresponds to about $2\%$ and $1\%$ of annual trading costs for higher and medium frequency TF systems, respectively, see Figure \ref{fig:tf_sg_backtest_paper}.} to each simulated path, and we compute the net Sharpe ratio per path as the annual return minus the cost times the annualized turnover, divided by the gross annualized volatility, matching equation \eqref{eq:sr_net_def}. For each path, we compute the average annual return to generate the sample of $1000$ MC estimates $\{\hat{F}_{1y}\}$. The process exhibits report the mean of the per-path estimates, with the $95\%$ confidence interval $1.96\times stdev(\hat{F}_{1y})/\sqrt{1000}$ across the paths. The verification table of Section \ref{ssec:sr_mc} reports the pooled estimators over all paths, with confidence intervals from whole-path blocks.

We compute the expected return of the European TF system in Definition \ref{as:eu} as a function of the EWMA filter span in days, labeled as follows:
\begin{equation} \label{eq:fs1} 
	1w \equiv 5, \ 2w \equiv 10, \ 1m \equiv 21, \ 3m \equiv 63, \ 6m \equiv 125, \ 1y \equiv 250, \ 2y \equiv 500.
\end{equation}

The simulated returns enter the pipeline in two roles. They feed the volatility estimator $\sigma_{t}$ in equation \eqref{eq:ewmavol}, and they define the volatility-normalized returns $z_{t}$ in equation \eqref{eq:eu1}. All performance calculations then operate on $z_{t}$, because the identity in equation \eqref{eq:ret} expresses the daily strategy return as $f_{t} = (\sigma_{\mathrm{target}}/\sqrt{a})\, S_{t-1} z_{t}$ with the signal computed from $z_{t}$. The simulation therefore applies the strategy to the pipeline-normalized process $z_{t}$, and every performance statistic is a functional of the volatility-normalized returns. Dividing by an estimated volatility does not by itself enforce the stationarity or the population moments of Assumption \ref{as:pr}, so the Monte Carlo measures the discrepancy between the pipeline process and the analytical assumptions. The analytic values in the figures and in Table \ref{tab:t6} apply the population moments of $r_{t}$ in place of those of $z_{t}$, which the inheritance justifies up to the estimation effect.

\subsection{White Noise}

We first consider a white noise process with mean $\mu$ and unit annualized variance:
\begin{equation} \label{eq:wn1} 
	r_{t} = \mu + \epsilon_{t},
\end{equation}
where $\epsilon_{t}$ is normal with zero mean and variance $1/a$, and $\mu=\mu^{z}_{an}/a$, so the volatility-normalized returns have the daily mean $\mu^{z}_{an}/\sqrt{a}$ under a constant volatility estimate, and approximately under the estimated EWMA volatility, because the random denominator changes the expectation of $1/\sigma_{t-1}$. The autocorrelation is the Kronecker delta at $m=0$, $\rho(m) = \delta_{m=0}$, so that the expected return $\overline{F}_{1y}$ in equation \eqref{eq:er2} becomes as follows:
\begin{equation} \label{eq:wn3} 
	\bar{F}_{1y} = \left( \frac{l\sigma_{\mathrm{target}} }{\sqrt{a}} \right)   (\mu^{z}_{an})^{2}.
\end{equation}
Thus, a TF system under white noise can be profitable only if the mean $\mu$ is not zero. Figure \ref{fig:expected_return_white_noise} shows the expected return in panel (A) from equation \eqref{eq:wn3}, the gross Sharpe ratio in panel (B) from equation \eqref{eq:sr_wn}, and the net Sharpe ratio at $c=20$bp in panel (C) from equation \eqref{eq:sr_net}. The zero-drift Sharpe ratio is zero at every span, and the Sharpe ratio for non-zero drift grows with the span. The analytic results are within the MC confidence intervals.

\begin{figure}[tph]\hspace*{-3\baselineskip}
	\begin{center}\hspace*{-3\baselineskip}\vspace*{-\baselineskip} 
		\includegraphics[width=0.9\textwidth, angle=0] {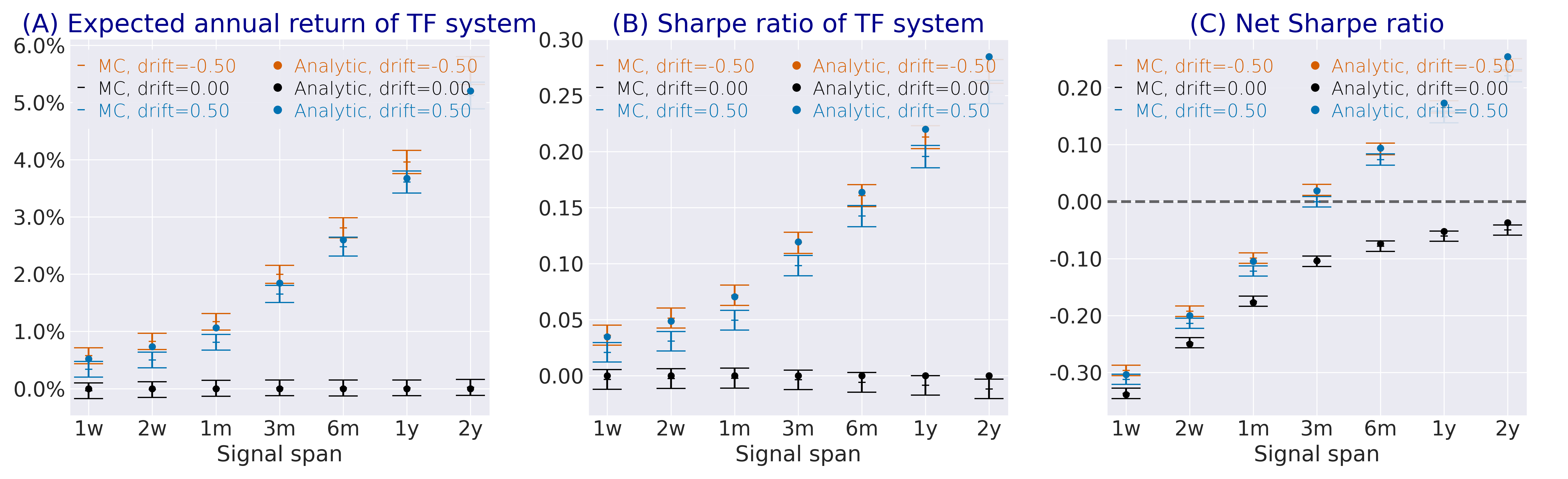}\hspace*{-2.5\baselineskip} 
	\end{center}
	\caption{Expected annual return, gross Sharpe ratio, and net Sharpe ratio of the European TF system on the white noise process in equation \eqref{eq:wn1}, as functions of the annualized drift $\mu^{z}_{an}\in\{-0.5,0.0,0.5\}$ and the filter span, with $\sigma_{\mathrm{target}}=15\%$. Panel (A) shows the expected annual return, panel (B) shows the Sharpe ratio, and panel (C) shows the net Sharpe ratio at the cost of $c=20$bp per unit of volatility-normalized turnover, with analytic values from equation \eqref{eq:sr_net}. Analytic values use the closed-form formulas, and MC midpoints with $95\%$ confidence intervals are obtained by Monte Carlo simulation of $1000$ paths of $50$ years}
	\label{fig:expected_return_white_noise}
\end{figure}

We now turn to the Sharpe ratio, shown in panel (B) of Figure \ref{fig:expected_return_white_noise}. Under white noise with drift, we have $\Psi_{\nu}=0$, so $A_{\nu}=0$ and $B_{\nu}=(1-\nu)/(1+\nu)=1/\spn$ using the span convention in equation \eqref{eq:i2}. Equation \eqref{eq:sr_main} becomes:
\begin{equation} \label{eq:sr_wn}
	SR = \frac{ \left(\mu^{z}_{an}\right)^{2} \sqrt{\spn/a} }{ \sqrt{ 1 + \left( \left(\mu^{z}_{an}\right)^{2}/a \right)\left( \spn + 1 \right) } } \approx \left(\mu^{z}_{an}\right)^{2} \sqrt{\frac{\spn}{a}},
\end{equation}
where the approximation drops the term of order $\left(\mu^{z}_{an}\right)^{2}/a$. The leading order gives a simple rule. At the one-year span with $\spn=a$, the Sharpe ratio of the TF system approximately equals the squared Sharpe ratio of the underlying instrument, because the exact denominator in equation \eqref{eq:sr_wn} exceeds one. The net Sharpe ratio follows from equation \eqref{eq:sr_net} in closed form:
\begin{equation} \label{eq:sr_net_wn}
	SR^{net} \approx \frac{ \left(\mu^{z}_{an}\right)^{2} \sqrt{\spn/a} - \left(2ac/\sqrt{\pi}\right) \sqrt{2/(\spn+1)} }{ \sqrt{ 1 + \left( \left(\mu^{z}_{an}\right)^{2}/a \right)\left( \spn + 1 \right) } } \approx \left(\mu^{z}_{an}\right)^{2} \sqrt{\frac{\spn}{a}} - 2ac\sqrt{\frac{2}{\pi\, \spn}}.
\end{equation}
The drift alpha grows with $\sqrt{\spn}$ while the cost drag declines with $1/\sqrt{\spn}$, so, under the signal-turnover proxy, any proportional cost is beaten at a long enough span. The net Sharpe ratio in panel (C) reaches $0.25$ at the two-year span for $|\mu^{z}_{an}|=0.5$, and the zero-drift curve isolates the pure cost drag, which declines from $0.34$ at the one-week span to $0.04$ at the two-year span. In the long-span limit, the direction of the signal converges to the sign of the drift and the system converges to a buy-and-hold exposure up to a deterministic scale, so the net Sharpe ratio under the proxy approaches the absolute Sharpe ratio $|\mu^{z}_{an}|$ of the underlying instrument for any proportional cost (see Remark \ref{rem:asym_wn}).
\subsection{AR-1 process}

Next, we consider the AR-1 process with $|\phi|<1$ and the dynamics:
\begin{equation} \label{eq:a3} 
	\begin{split}
		r_{t} & = \mu + x_{t}, \qquad x_{t} = \phi x_{t-1} + \epsilon_{t},
	\end{split}
\end{equation}
where $\left|\phi\right|<1$, $\epsilon_{t}$ is normal with variance $1/a$, and the raw drift parameter $\mu_{an}$ sets $\mu=\mu_{an}/a$ outside the autoregression without the intercept amplification $1/(1-\phi)$. The stationary variance of $x_{t}$ is $1/(a(1-\phi^{2}))$, so the normalized annualized drift is $\mu^{z}_{an} = \mu_{an}\sqrt{1-\phi^{2}}$, which equals $\mu_{an}$ to the leading order in $\phi^{2}$, a $0.1\%$ effect at $|\phi|=0.05$.

We recall the following features of the AR-1 process. If $\phi>0$, the volatility-normalized returns are positively autocorrelated and the markets are divergent: the volatility of aggregated returns settles above the square-root-of-time benchmark by the factor $\sqrt{(1+\phi)/(1-\phi)}$ at long horizons. If $\phi<0$, the volatility-normalized returns are negatively autocorrelated and the markets are mean-reverting: the aggregated volatility settles below the benchmark by the same factor.

The autocorrelation of the AR-1 process is the power function:
\begin{equation} \label{eq:a4} 
	\rho(m) = \phi^{m}.
\end{equation}

Thus, the expected return in equation \eqref{eq:er2} is given by:
\begin{equation} \label{eq:a5} 
	\bar{F}_{1y} = h_{1y} \left[  \frac{\phi\nu}{1-\phi\nu}   \right] + \left( \frac{l\sigma_{\mathrm{target}} }{\sqrt{a}} \right)   \left(\mu^{z}_{an}\right)^{2}.  
\end{equation}

We reach two conclusions. First, the expected return at zero drift is positive for $\phi>0$, when the markets are diverging, and negative for $\phi<0$, when the markets are mean-reverting. Second, the shorter the filter span, the larger the autocorrelation contribution to the expected return. Figure \ref{fig:expected_return_ar} shows the analytical results and the MC confidence intervals for the zero-drift AR-1 process with $\phi=\{-0.05, 0.05\}$. The gross Sharpe ratio in panel (B) is positive for $\phi=0.05$, negative for $\phi=-0.05$, and declines in magnitude with the span.  Panel (C) reports the net Sharpe ratio at $c=20$bp. The analytical results are within the MC intervals.

\begin{figure}[tph]\hspace*{-3\baselineskip}
	\begin{center}\hspace*{-3\baselineskip}\vspace*{-\baselineskip}
		\includegraphics[width=0.9\textwidth, angle=0] {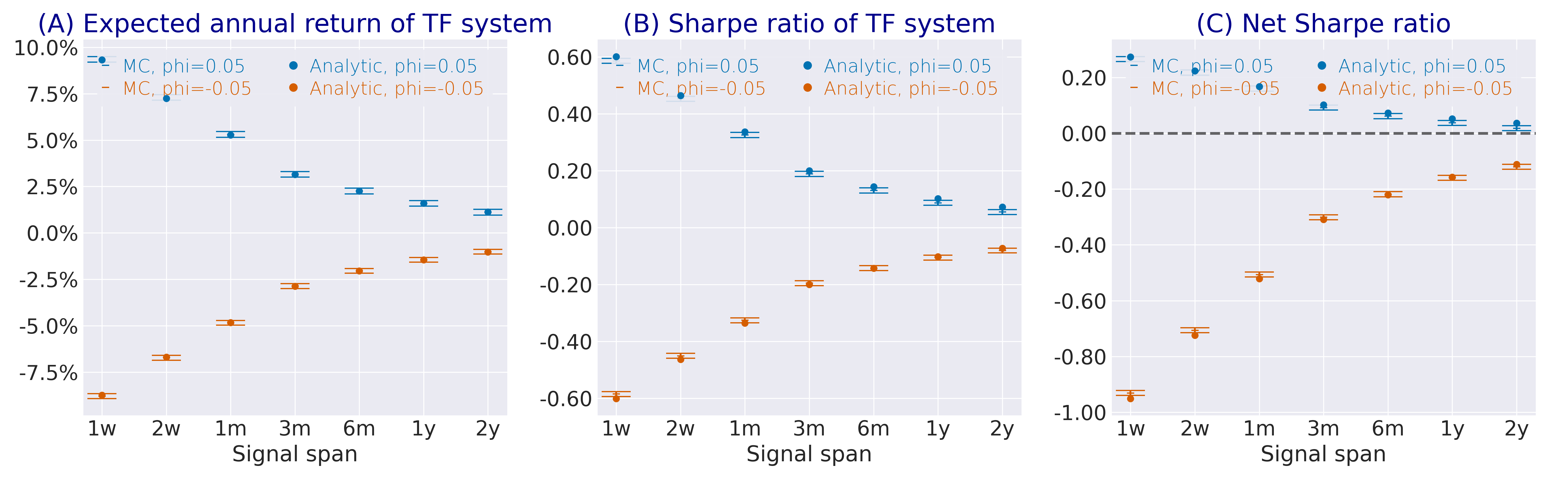}\hspace*{-2.5\baselineskip}
	\end{center}
	\caption{Expected annual return, gross Sharpe ratio, and net Sharpe ratio of the European TF system for the zero-drift AR-1 process in equation \eqref{eq:a3} with $\phi=\{-0.05, 0.05\}$. Panel layout, cost level, and analytic and MC conventions as in Figure \ref{fig:expected_return_white_noise}}
	\label{fig:expected_return_ar}
\end{figure}

Under the zero-drift AR-1 process in equation \eqref{eq:a3}, the loadings and the Sharpe ratio of panel (B) follow from equation \eqref{eq:sr_main} in closed form for $\kappa=0$:
\begin{equation} \label{eq:sr_ar1}
	A_{\nu} = \frac{(1-\nu)\phi}{1-\nu\phi}, \qquad B_{\nu} = \frac{(1-\nu)(1+\nu\phi)}{(1+\nu)(1-\nu\phi)}, \qquad SR = \frac{\sqrt{a}\,A_{\nu}}{\sqrt{B_{\nu}+A^{2}_{\nu}}}.
\end{equation}
The AR-1 process is the linear process with the weights $\psi_{s}=\sqrt{1-\phi^{2}}\,\phi^{s}$, so the kurtosis loading in equation \eqref{eq:sr_k} has the closed form:
\begin{equation} \label{eq:sr_ar1_k}
	K_{\nu} = \frac{(1-\nu)^{2}(1-\phi^{2})^{2}}{(\phi-\nu)^{2}} \left( \frac{\phi^{4}}{1-\phi^{4}} - \frac{2\nu\phi^{3}}{1-\nu\phi^{3}} + \frac{\nu^{2}\phi^{2}}{1-\nu^{2}\phi^{2}} \right).
\end{equation}
At $\phi=\nu$, equation \eqref{eq:sr_ar1_k} is interpreted by continuity.
For $\phi=0.05$ and $\kappa=3$, the correction $\kappa K_{\nu}$ lowers the Sharpe ratio by at most $0.2\%$ in relative terms across the spans from $5$ to $250$ days. We therefore evaluate the AR-1 exhibits with $\kappa=0$.

For a small AR-1 coefficient, the expansion in $\phi$ yields a compact approximation:
\begin{equation} \label{eq:sr_ar1_approx}
	SR \approx \phi \sqrt{a \left(1-\nu^{2}\right)} = \frac{2\phi\sqrt{a}\sqrt{\spn}}{\spn+1} \approx 2\phi\sqrt{\frac{a}{\spn}}.
\end{equation}
The approximation is accurate for the parameter ranges that we consider. The net Sharpe ratio admits the same expansion:
\begin{equation} \label{eq:sr_net_ar1_approx}
	SR^{net} \approx \phi \sqrt{a\left(1-\nu^{2}\right)} - \frac{2ac}{\sqrt{\pi}}\sqrt{1-\nu} \approx 2\sqrt{\frac{a}{\spn}} \left( \phi - c\sqrt{\frac{2a}{\pi}} \right).
\end{equation}
Both terms decay with $1/\sqrt{\spn}$, so the sign of the net Sharpe ratio is span-invariant at leading order and positive if and only if $c < \phi\sqrt{\pi/(2a)}$, which is the small-$\phi$ form of the threshold in equation \eqref{eq:sr_net_ar1}. The net Sharpe ratio reveals a knife-edge property of the AR-1 alpha. The gross Sharpe ratio and the cost drag both decline with the square root of the span, so the break-even cost is nearly span-invariant, with the closed form of Proposition \ref{pr:asym_ar1} and the long-span limit:
\begin{equation} \label{eq:sr_net_ar1}
	c^{*}_{\infty} = \sqrt{\frac{\pi}{2a}}\,\frac{\phi}{1-\phi}.
\end{equation}
For $\phi=0.05$ and $a=260$, the proxy break-even cost varies only between $37$bp and $41$bp across the spans from one week to two years. These values lie around the lower edge of realistic costs of $40$bp to $60$bp, and the pipeline turnover of the single filter at the volatility span of $33$ days exceeds the proxy by about $4\%$, which places the break-even costs slightly below that range. To first order, costs determine the viability of a short-memory alpha. Near the break-even level, costs also push the cost-optimal span toward longer horizons.

Proposition \ref{pr:asym_ar1} derives the interior optimal span in the limit $c \to c^{*}_{\infty}$:
\begin{equation} \label{eq:sr_net_span}
	\spn^{*} \approx \frac{6\phi/(1-\phi) + 3c/(2c^{*}_{\infty})}{1 - c/c^{*}_{\infty}}.
\end{equation}
As the cost approaches the break-even level, the optimal span diverges and the attainable net Sharpe ratio declines toward zero. The optimal response to a marginally profitable short-memory alpha is therefore to trade it at ever longer spans. Equation \eqref{eq:sr_net_span} is asymptotic in the limit $c \to c^{*}_{\infty}$ and is not intended for quantitative use at costs well below the threshold. For such costs, the interior candidate falls below the one-week minimum of our span grid, and the optimum reverts to the corner solution at the fastest span.

The closed forms expose a fundamental trade-off in the span selection. The drift contribution to the Sharpe ratio in equation \eqref{eq:sr_wn} grows with $\sqrt{\spn}$, while the short-term autocorrelation contribution in equation \eqref{eq:sr_ar1_approx} decays with $1/\sqrt{\spn}$. A short filter monetizes short-term autocorrelation and a long filter monetizes the drift. For processes that combine both features, the Sharpe ratio in equation \eqref{eq:sr_main} may attain its maximum at an interior span, depending on the parameter magnitudes and signs, and the formula turns the span selection into an explicit function of the process parameters.

\subsection{Long-memory ARFIMA Process}

\cite{Granger1980} and \cite{Hosking1981} introduce the autoregressive fractionally integrated moving-average (ARFIMA) process. \citet{SchadnerLang2023} estimate option-implied Hurst exponents for S\&P 500 stocks and document time-varying persistence of returns, which supports long memory as a priced feature. The ARFIMA $(0,d,0)$ process is a discrete-time analog of fractional Gaussian noise, with the infinite moving-average representation:
\begin{equation} \label{eq:lm1} 
	\tilde{r}_{t} =   \sum^{\infty}_{j=0}\psi_{j}\epsilon_{t-j}, \ \psi_{j} = \frac{\Gamma(d+j)}{\Gamma(1+j)\Gamma(d)},
\end{equation}
where $\{\epsilon_{t}\}$ are independent innovations with zero mean and unit variance, Gaussian in our exhibits, and $\Gamma(x)$ is the Gamma function. The process is covariance stationary and invertible for $-1/2 < d < 1/2$, which we assume throughout. At $d=0$, the Gamma-function representation of the weights is interpreted by its white-noise limit $\psi_{0}=1$ and $\psi_{j}=0$ for $j \geq 1$.

The variance and autocorrelation of $\tilde{r}_{t}$ are given respectively by:
\begin{equation} \label{eq:lm3} 
	\gamma^{\tilde{r}}_{0} = \frac{\Gamma(1-2d)}{(\Gamma(1-d))^{2}}, \ \rho^{\tilde{r}}_{k} = \frac{\Gamma(1-d)\Gamma(k+d)}{\Gamma(d)\Gamma(k-d+1)}.
\end{equation}

Further, we consider the ARFIMA $(1,d,0)$ process, which incorporates both the short-term autocorrelation produced by the AR-1 process and the long-term autocorrelation produced by the ARFIMA $(0,d,0)$ process:
\begin{equation} \label{eq:lm2} 
	r_{t} = \mu + x_{t}, \qquad x_{t} = \phi x_{t-1} + \tilde{r}_{t}/\sqrt{a},
\end{equation}
where $|\phi|<1$ and the raw drift parameter $\mu_{an}$ sets $\mu=\mu_{an}/a$ outside the autoregression, as in equation \eqref{eq:a3}. The stationary variance of $x_{t}$ equals $V_{\phi,d}/a$ with $V_{\phi,d}$ in equation \eqref{eq:lm4a}, so the normalized annualized drift is $\mu^{z}_{an} = \mu_{an}/\sqrt{V_{\phi,d}}$ rather than $\mu_{an}$. The rescaling is $1.0\%$ at $(\phi, d)=(0, 0.1)$ and below $0.1\%$ at $d=0.02$, and the analytic values apply the normalized drift.

The autocorrelation of the ARFIMA $(1,d,0)$ process $r_{t}$ is given by:
\begin{equation} \label{eq:lm4} 
	\rho^{r}_{k} = \rho^{\tilde{r}}_{k}  \frac{ F(1,d+k,1-d+k; \phi) + F(1,d-k,1-d-k; \phi) - 1}{(1-\phi)F(1,1+d,1-d;\phi)},
\end{equation}
with the variance of the unscaled auxiliary process $\tilde{x}_{t} = \phi \tilde{x}_{t-1} + \tilde{r}_{t}$ given by:
\begin{equation} \label{eq:lm4a} 
	V_{\phi,d} = \gamma^{\tilde{r}}_{0} \frac{F(1,1+d,1-d;\phi)}{1+\phi},
\end{equation}
where $F(a,b,c;x)$ is the hypergeometric function.

The AR-1 autocorrelations decay geometrically, while the long-lag ARFIMA tail decays slowly, with the positive sign for $d>0$ and the negative sign for $d<0$.

We evaluate the autocorrelation sums in equations \eqref{eq:er2} and \eqref{eq:er3} by direct summation using equation \eqref{eq:lm4}. For heavy-tailed innovations, the kurtosis loading $K_{\nu}$ follows numerically from the moving-average weights in equation \eqref{eq:lm1}, rescaled to the unit-norm convention of Assumption \ref{as:lp}. Figure \ref{fig:expected_return_arfima1} shows the expected return, the gross Sharpe ratio, and the net Sharpe ratio at $c=20$bp for the ARFIMA process with $d=0.02$ and $\phi\in\{-0.05, 0.0, 0.05\}$. The calibration matches the observed first-lag autocorrelation through $\rho(1)=d/(1-d)$ for the pure fractional process, with $d=0.02$ at the center of the $0.01$ to $0.04$ range of Section \ref{sec:attr}. Long memory alone makes the TF system profitable, and it stays profitable under short-term mean reversion when the filter span exceeds one month.

\begin{figure}[!htb]\hspace*{-3\baselineskip}
	\begin{center}
		\includegraphics[width=0.86\textwidth, angle=0] {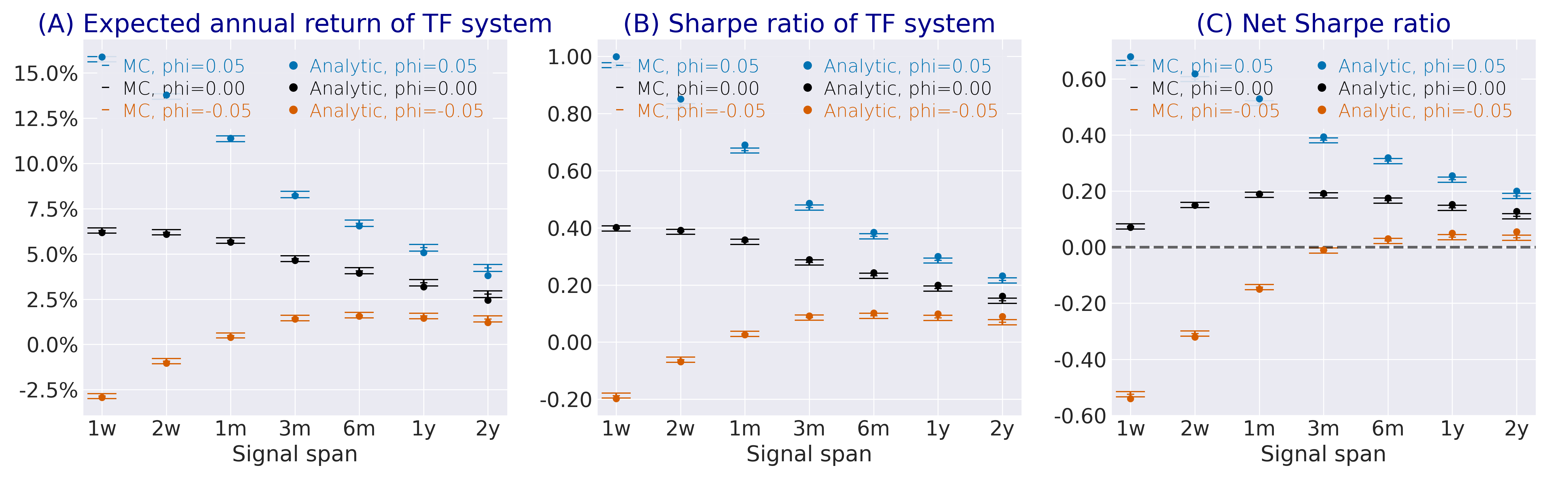}
	\end{center}
	\caption{Expected annual return, gross Sharpe ratio, and net Sharpe ratio of the European TF system for the ARFIMA process with positive long-term memory, fractional order $d=0.02$, AR-1 feature $\phi\in\{-0.05, 0.0, 0.05\}$, and zero drift. Panel layout, cost level, and analytic and MC conventions as in Figure \ref{fig:expected_return_white_noise}}
	\label{fig:expected_return_arfima1}
\end{figure}

The Sharpe ratio in panel (B) is highest at the shortest spans for the pure long-memory case. Panel (C) displays the three regimes of the cost-adjusted span selection in a single figure. For $\phi=0.05$, the alpha dominates the cost at every span and the net Sharpe ratio declines from $0.68$ at the one-week span, so the fast corner remains optimal. For the pure long-memory case $\phi=0$, the net Sharpe ratio is hump-shaped with an interior maximum of $0.19$ at the one-to-three-month spans, because the ratio of the long-memory alpha to the cost drag grows with $\spn^{2d}$ while both decline in the span\footnote{In the large-cost limit, the optimal span scales as $c^{1/(2d)}$ and the net-to-gross ratio at the optimum converges to $4d/(1+2d)$ (Proposition \ref{pr:asym_arfima}). For small $d$, the limit lies beyond practical cost levels: at $d=0.02$ and $c=20$bp, the continuous-span optimum lies at $37$ days with a net-to-gross ratio of $0.61$, against the limiting value of $0.08$. We therefore locate the cost-optimal span from the non-asymptotic proxy objective in equation \eqref{eq:sr_net}, rather than from the large-cost asymptotic law.}. For $\phi=-0.05$, the net Sharpe ratio turns positive between the three-month and the six-month spans, so short-term mean reversion and costs together push the viable region into slow spans. 

For $\phi \neq 0$, the moving-average weights entering the kurtosis loading $K_{\nu}$ are the convolution $h_{j} = \sum^{j}_{k=0} \phi^{k}\psi_{j-k}$ of the fractional weights with the AR weights, normalized to unit variance. For the pure fractional process ARFIMA$(0,d,0)$, using the closed-form autocorrelations derived by \citet{Hosking1981}, the autocorrelation generating function can be written as:
\begin{equation} \label{eq:sr_agf_hyp}
	\Phi_{\nu} = F(d, 1, 1-d; \nu),
\end{equation}
where $F$ is the hypergeometric function of equation \eqref{eq:lm4}, and the identity follows term by term from the Gauss series, because the series coefficients equal the autocorrelations in equation \eqref{eq:lm3}. At $\kappa=0$, the gross and net Sharpe ratios of the single filter and of the long-short filter are therefore closed-form under the pure fractional process, because Corollaries \ref{col:sr} and \ref{col:sr_ls} depend on the process only through $\Phi_{\nu}$ evaluated at the filter spans.

\subsection{Verification Tests}\label{ssec:sr_mc}

Table \ref{tab:t6} verifies the closed-form Sharpe ratios within the full pipeline at the configuration of the empirical section: the long-short filter LS(250,20) under the long-memory calibration $d=0.1$, which makes the cost and tail effects resolvable at the reported MC precision. The Gaussian sub-panels pair the analytic values at $\kappa=0$ with Gaussian simulations. The Student-t sub-panels pair the analytic values at $\kappa=3$ with simulations under standardized Student-t innovations with six degrees of freedom, which keep the unit variance and add an excess kurtosis of $3$, close to the empirical median of $3.8$ across our futures universe. The MC net values charge the cost against the realized turnover of equation \eqref{eq:tur1}, and the analytic net values apply the signal-turnover proxy of Proposition \ref{pr:turnover2}. The ARFIMA drift rows label the raw drift parameter $\mu_{an}=0.50$ of equation \eqref{eq:lm2}, and the variance scale maps it to the normalized annual drift $\mu^{z}_{an} = \mu_{an}/\sqrt{V_{\phi,d}}$: $V_{0,0.1} = \Gamma(0.8)/\Gamma(0.9)^{2} \approx 1.019$ gives $0.495$, and $V_{-0.05,0.1} \approx 1.011$ gives $0.497$. The white-noise rows carry the same convention, with $\mu^{z}_{an}=\mu_{an}$ under unit variance. The analytic and the MC values both reflect the normalized drift. Both the MC and the analytic net ratios divide the cost-adjusted mean by the gross annualized volatility, matching the definition in equation \eqref{eq:sr_net_def}. The AR-1 rows are near zero for both signs of $\phi$, because the two legs of the long-short filter carry the identical loading on the first lag, which cancels in the signal: the covariance in equation \eqref{eq:sr_ls} reduces to $q\phi^{2}(\nu_{1}-\nu_{2})/\left((1-\nu_{1}\phi)(1-\nu_{2}\phi)\right)$, which is positive for either sign of $\phi$ and of the second order in $\phi$. The single filter retains the first-order term and its Sharpe ratio changes sign with $\phi$, as equation \eqref{eq:sr_ar1_approx} shows. Figure \ref{fig:expected_return_arfima1} verifies the formulas at the empirically calibrated $d=0.02$.

The table supports three conclusions. First, the analytic values track the simulations within $0.05$ across all processes, with the gap driven by the dampening of the autocorrelation of $z_{t}$ by the lagged volatility estimator, which lowers the first-lag autocorrelation from $0.050$ to $0.046$ under the AR-1 calibration and attenuates the Sharpe ratio by $2\%$ to $8\%$ in relative terms. Second, the heavy tails lower the pooled gross Sharpe ratio by at most $0.009$, so the kurtosis effect is second order within the pipeline. Third, the cost of $20$bp lowers the Sharpe ratio by at most $0.02$, while the simulated turnover exceeds the proxy by about a factor of two at the volatility span of $33$ days, as Section \ref{sc:eft} discusses. The formulas with population moments therefore predict the realized performance of the empirical section's filter, gross and net, under realistic tails.

\begin{table}[!htb]
	\footnotesize
	\centering
	\resizebox{\textwidth}{!}{
		\begin{tabular}{l|cc|cc|cc|cc}
			\hline
			& \multicolumn{4}{c|}{Gross Sharpe ratio} & \multicolumn{4}{c}{Net Sharpe ratio at $c=20$bp} \\
			& \multicolumn{2}{c|}{Gaussian} & \multicolumn{2}{c|}{Student-t(6)} & \multicolumn{2}{c|}{Gaussian} & \multicolumn{2}{c}{Student-t(6)} \\
			Process & Analytic & MC & Analytic & MC & Analytic & MC & Analytic & MC \\
			\hline
			White noise, $\mu_{an}=0.25$ & 0.062 & 0.065 & 0.062 & 0.061 & 0.050 & 0.048 & 0.050 & 0.042 \\
			White noise, $\mu_{an}=0.50$ & 0.227 & 0.234 & 0.227 & 0.225 & 0.216 & 0.217 & 0.216 & 0.207 \\
			AR-1, $\phi=+0.05$ & 0.001 & -0.000 & 0.001 & -0.002 & -0.011 & -0.018 & -0.011 & -0.021 \\
			AR-1, $\phi=-0.05$ & 0.000 & 0.002 & 0.000 & 0.003 & -0.012 & -0.016 & -0.012 & -0.016 \\
			ARFIMA, $d=0.1$ & 0.696 & 0.670 & 0.696 & 0.664 & 0.689 & 0.654 & 0.689 & 0.647 \\
			ARFIMA, $d=0.1$, $\phi=-0.05$ & 0.666 & 0.640 & 0.666 & 0.635 & 0.659 & 0.624 & 0.659 & 0.619 \\
			ARFIMA, $d=0.1$, $\mu_{an}=0.50$ & 0.809 & 0.781 & 0.809 & 0.776 & 0.802 & 0.765 & 0.802 & 0.759 \\
			ARFIMA, $d=0.1$, $\phi=-0.05$, $\mu_{an}=0.50$ & 0.784 & 0.763 & 0.784 & 0.757 & 0.777 & 0.747 & 0.777 & 0.740 \\
			\hline
	\end{tabular}}
	\caption{Trading costs at $20$bp lower the long-short Sharpe ratio by at most $0.02$, heavy tails by at most $0.009$, and the analytic formulas track the simulations within $0.05$. The table reports annualized gross and net Sharpe ratios of the European TF system with the filter LS(250,20). Analytic values apply Corollary \ref{col:sr_ls} with $\kappa=0$ in the Gaussian sub-panels and $\kappa=3$ in the Student-t sub-panels, where the kurtosis correction stays below $0.001$, and charge the cost of $c=20$bp against the signal-turnover proxy of Proposition \ref{pr:turnover2}. MC values are pooled estimators over $1000$ paths of $50$ years through the pipeline of Section \ref{sec:sp}, with standardized Student-t(6) innovations in the Student-t sub-panels, and charge the cost against the realized turnover of equation \eqref{eq:tur1}. Their $95\%$ confidence intervals range from $\pm 0.008$ to $\pm 0.015$, and the full grid is reproducible from the companion repository}
	\label{tab:t6}
\end{table}

\section{Implementation of TF Systems}\label{sc:impl}

We apply the three TF systems to the universe of $84$ futures contracts in Table \ref{fig:universe}, which covers the most liquid contracts across global equity, bond, short-rate, currency, and commodity markets, with the volume-based trading costs of Table \ref{fig:costs}, which follow Exhibit B1 in \citet{Hurst2017} and split by period and asset class. These costs cover rebalancing and rolling, and in our experience they are representative for a trend-following system at a medium to slow trading frequency. The proportional cost model excludes market impact, which scales nonlinearly with trade size and becomes binding for large programs.

\begin{table}[th]
	\small
	\centering
	\begin{tabular}{lcl}
		\hline
		Asset class & Number of contracts & Earliest contract \\
		\hline
		Equities & $21$ & January 1962 \\
		Bonds & $16$ & January 1962 \\
		STIR & $4$ & June 1989 \\
		FX & $10$ & September 1971 \\
		Energy & $7$ & July 1986 \\
		Metals & $6$ & January 1975 \\
		Agriculture & $20$ & July 1959 \\
		\hline
		Total & $84$ & July 1959 \\
		\hline
	\end{tabular}
	\caption{Universe of $84$ liquid futures contracts split into asset classes. The table reports the number of contracts and the start of the earliest contract in each class. STIR stands for short-term interest rate, and FX stands for foreign exchange}
	\label{fig:universe}
\end{table}

\begin{table}[th]
	\small
	\centering
	\begin{tabular}{lccc}
		\hline
		Asset class & Until Dec 1992 & Jan 1993 to Dec 2002 & From Jan 2003 \\
		\hline
		Equities & $34$bp & $11$bp & $6$bp \\
		Bonds & $6$bp & $2$bp & $1$bp \\
		STIR & $6$bp & $2$bp & $1$bp \\
		FX & $18$bp & $6$bp & $3$bp \\
		Energy, Metals, Agriculture & $58$bp & $19$bp & $10$bp \\
		\hline
	\end{tabular}
	\caption{One-way volume-based transaction costs applied in the backtests, in basis points, by asset class and period, following Exhibit B1 of \citet{Hurst2017}. The Energy, Metals, and Agriculture classes in Table \ref{fig:universe} share the commodity estimate, and the STIR class takes the bond estimate}
	\label{fig:costs}
\end{table}

\subsection{Grid Backtests}\label{sc:grid}

We start with grid backtests of the TF systems, to analyze the parameter sensitivity and to select the model parameters, simulating the performance from 1 January 1997 to 10 July 2026. Each system is applied to all available futures contracts with individual profit-and-loss aggregated on the portfolio level, without portfolio risk limits. We estimate EWMA volatility using equation \eqref{eq:ewmavol} for European and TSMOM TF systems with a span of $33$ days, and the ATR using equation \eqref{eq:int5} for American TF with the period of $33$ days. We assume that futures contracts trade in fractions, with an initial trade level of $1$ USD and net profits reinvested. We consider three characteristics. The Sharpe ratio applies the arithmetic daily-moment convention of equation \eqref{eq:sr_def} to the total excess performance without management or performance fees, and returns on futures contracts are excess returns, so the Sharpe ratio requires no adjustment by the risk-free rate. The Bear Sharpe ratio conditions the excess performance on the bear regime. The annualized average costs are realized by portfolio rebalancing and computed with the volume costs in Table \ref{fig:costs}, and the costs reflect the turnover of each configuration. We define the bear regime as the $16\%$ of worst quarters of the long-only benchmark portfolio, because the $16\%$ tail corresponds to the one-sigma quantile of the normal distribution. For the long-only benchmark, we use the 60/40 equity/bond portfolio computed using S\&P 500 index futures and the 10-year UST futures, respectively, with quarterly rebalancing. The grid contains $64$ configurations for each system. We read the grid as a sensitivity analysis rather than a strategy selection, because the maximum Sharpe ratio over a grid overstates the out-of-sample value (\citealp{Sullivan1999}).

For the European TF system, the two key parameters include long and short spans in days for the long-short EWMA filter in equation \eqref{eq:ewma3}, where we set $\nu_{1}=1-2/(\mathrm{long}+1)$ and $\nu_{2}=1-2/(\mathrm{short}+1)$. The top row of Figure \ref{fig:grid_backtests} shows the corresponding grid backtest. Fast TF with a small long span leads to high turnover and realized costs. We fix the long and short spans to 250 and 20, which provide reasonable long-term performance at moderate turnover, with realized costs of $1.7\%$, in line with a typical real-world CTA system.

For the American TF system in Definition \ref{as:am}, the core parameters are the spans of the long and short filters in equation \eqref{eq:as1} with $\nu_{1}=1-2/(\mathrm{long}+1)$ and $\nu_{2}=1-2/(\mathrm{short}+1)$. We set the entry-point scale parameter $\omega=5$ in equations \eqref{eq:as3} and \eqref{eq:as4} and set the stop-loss width parameter $p=5$ in equations \eqref{eq:as3a} and \eqref{eq:as4a}. The middle row of Figure \ref{fig:grid_backtests} shows a pattern similar to the European system, with the best total and Bear Sharpe ratios around the long and short values of $250$ and $20$.

For the generalized TSMOM system in equation \eqref{eq:tsmom2}, the two key parameters are the period length $L$ in days and the number of periods $M$ for computing the aggregate signal. The bottom row of Figure \ref{fig:grid_backtests} shows the optimal values around $M=L=10$.

\begin{figure}[!htb]
	\begin{center}\vspace*{-\baselineskip} 
		\includegraphics[width=0.85\textwidth, angle=0] {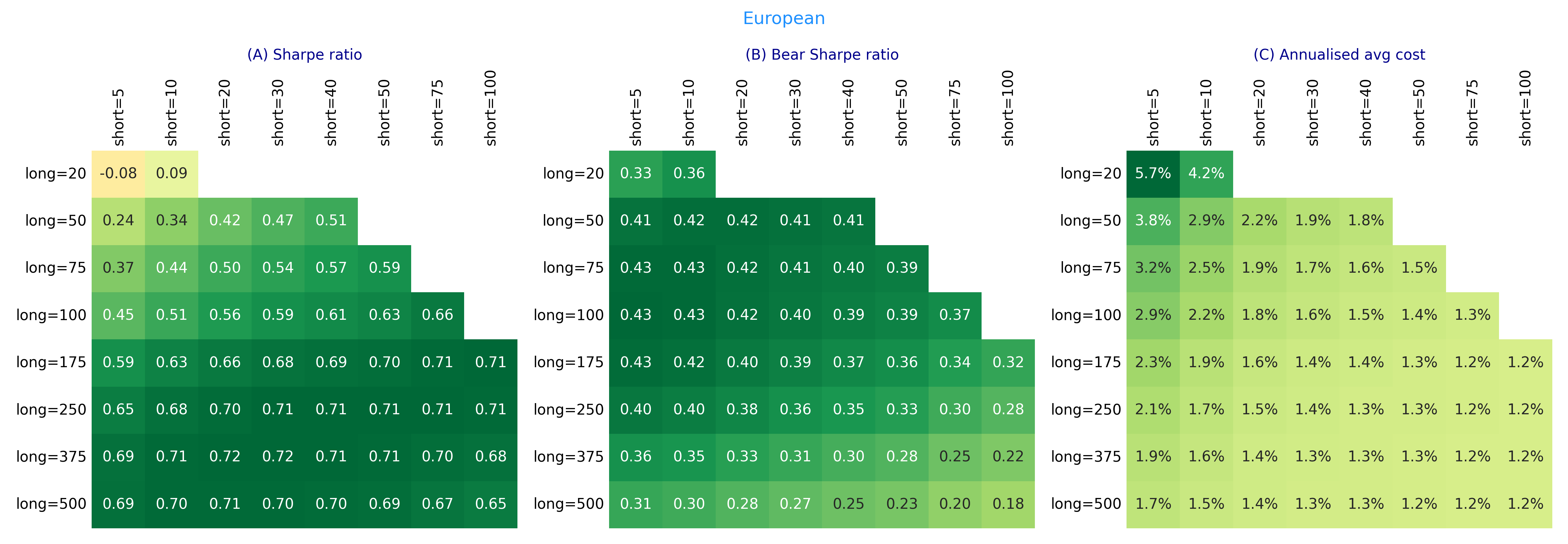}
		\includegraphics[width=0.85\textwidth, angle=0] {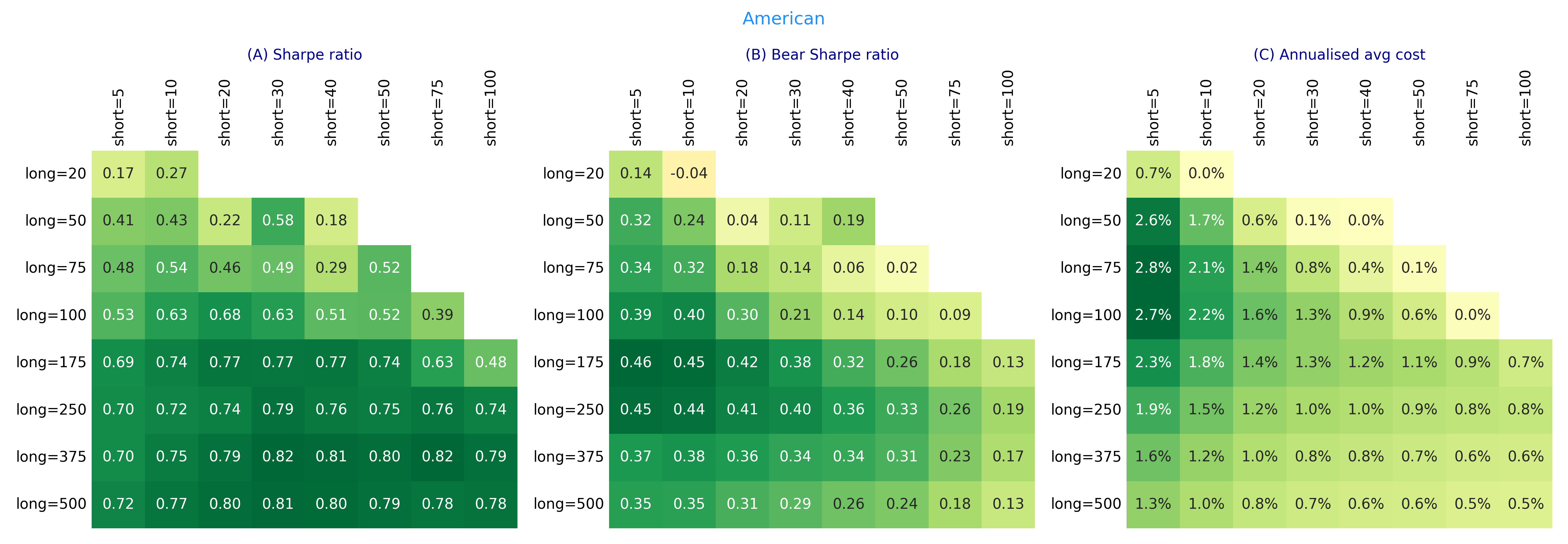}
		\includegraphics[width=0.85\textwidth, angle=0] {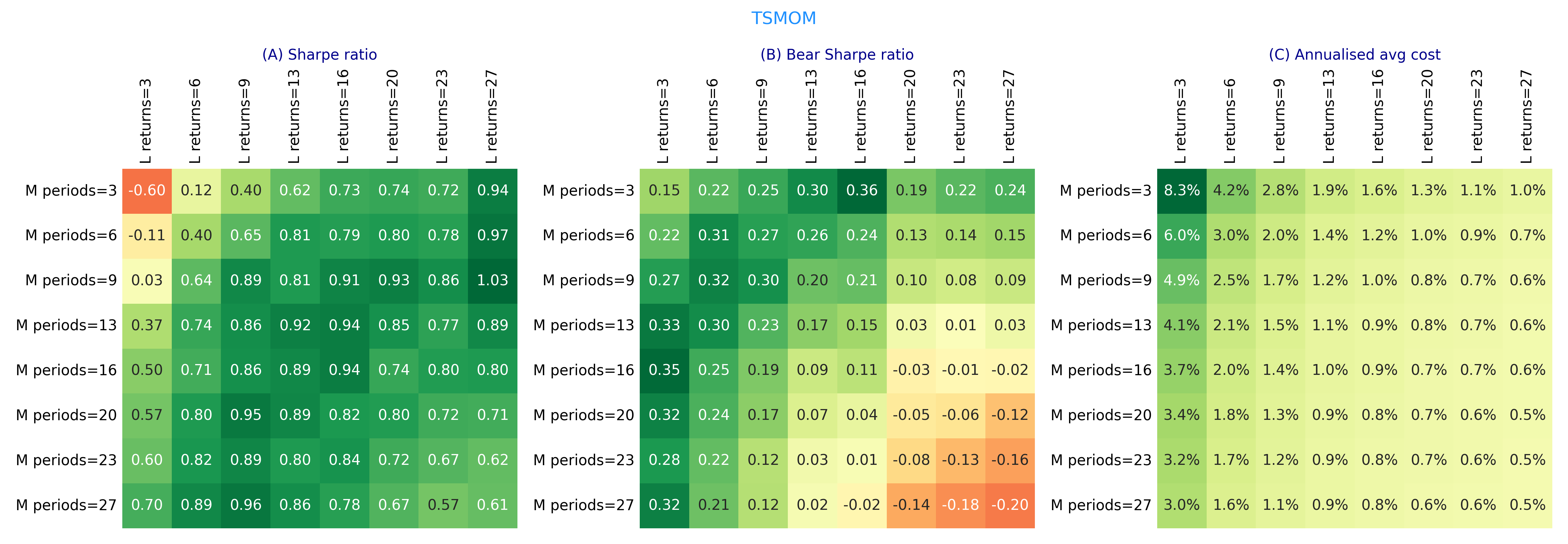}
	\end{center}
	\caption{Grid backtests of the three TF systems with excess returns realized from 1 January 1997 to 10 July 2026. The top row shows the European TF system on the grid of long and short spans in days for the long-short EWMA filter in equation \eqref{eq:ewma3} with $\nu_{1}=1-2/(\mathrm{long}+1)$ and $\nu_{2}=1-2/(\mathrm{short}+1)$. The middle row shows the American TF system in Definition \ref{as:am} on the same grid of spans, with the entry-point and stop-loss scales set by $\omega=p=5$. The bottom row shows the TSMOM system in equation \eqref{eq:tsmom2} on the grid of the period length $L$ and the number of periods $M$. In each row, (A) is the Sharpe ratio of excess returns, (B) is the Bear Sharpe ratio realized in the $16\%$ of worst quarters of the 60/40 equity/bond portfolio, and (C) is the realized transaction costs using the volume-based costs in Table \ref{fig:costs}}
	\label{fig:grid_backtests}
\end{figure}

\subsection{Comparison with SG Trend Index}\label{sec:comp}

We compare our three implementations of TF systems with the SG Trend Index, which tracks the performance of the 10 largest CTAs with annual rebalancing. To make a fair comparison, we apply $2\%/20\%$ management/performance annual fees to performances of our TF systems. We report the excess performance of the TF systems after transaction costs and fees. The index tracks funded programs and includes their interest income, which we cannot strip out, while we keep the TF systems and the 60/40 benchmark on excess returns, so the comparison modestly favors the index.

In Figure \ref{fig:tf_sg_backtest_paper}, we show the performance of the three simulated TF systems using specifications and established optimal parameters as in Section \ref{sc:grid}. For European and American TF systems, we set the spans of long and short EWMA filters to $250$ and $20$, respectively. We set the static volatility target and position risk so that all three systems generate the same level of volatility of their returns. For the TSMOM system, we set the period length to $L=10$ days and the number of periods to $M=10$. Panels (A1), (A2i), and (B1) show that, for this choice of TF system parameters, all three systems generate similar performance in the sample period from 31 December 1999 to 10 July 2026. 

From panel (C1), the three systems correlate with the SG Trend Index at $80\%$ on average, and the European and American systems correlate at $95\%$, which indicates the robustness of different implementations with similar lookbacks. From panels (B2) and (C2), the American TF system generates about $50\%$ lower turnover and rolling costs than the European and TSMOM systems. The core distinction is that the American TF system generates discrete trades, whereas the European and TSMOM systems generate continuous signals that require more frequent rebalancing. For practical implementation of TF systems using any methodology, it is beneficial to introduce some discretization rules.

We test the Sharpe ratio differences with the robust test of \cite{LedoitWolf2008} on monthly returns net of transaction costs and net of fees, which is the basis of Figure \ref{fig:tf_sg_backtest_paper} and the basis on which the index reports, from 31 December 1999 to 30 June 2026. The European, American, and TSMOM systems deliver annualized Sharpe ratios of $0.47$, $0.50$, and $0.55$ against $0.47$ for the SG Trend Index. The differences carry p-values of $0.96$, $0.82$, and $0.62$, so the test does not reject the equality of the Sharpe ratios for any of the systems. The failure to reject is consistent with the replication claim, although it is not evidence of equivalence.

From panel (A2ii), all three systems and the SG Trend Index generate net annual returns with small dispersion, so all three specifications capture strong market trends across regimes. TF systems are weakly correlated to the 60/40 portfolio over longer periods, so TF systems are natural diversifiers of long-only portfolios.

From panel (B2), the volatility-normalized turnover ranges between $300\%$ and $400\%$ for the European and TSMOM systems and between $125\%$ and $200\%$ for the American system. A turnover of $300\%$ corresponds to a long-short equity strategy rebalancing about $25\%$ per month, while the absolute annualized turnover is close to $2000\%$, dominated by the low-volatility bond and rates contracts.

\begin{figure}[!htb]\hspace*{-3\baselineskip}
	\begin{center}\hspace*{-3\baselineskip}\vspace*{-\baselineskip} 
		\includegraphics[width=0.9\textwidth, angle=0] {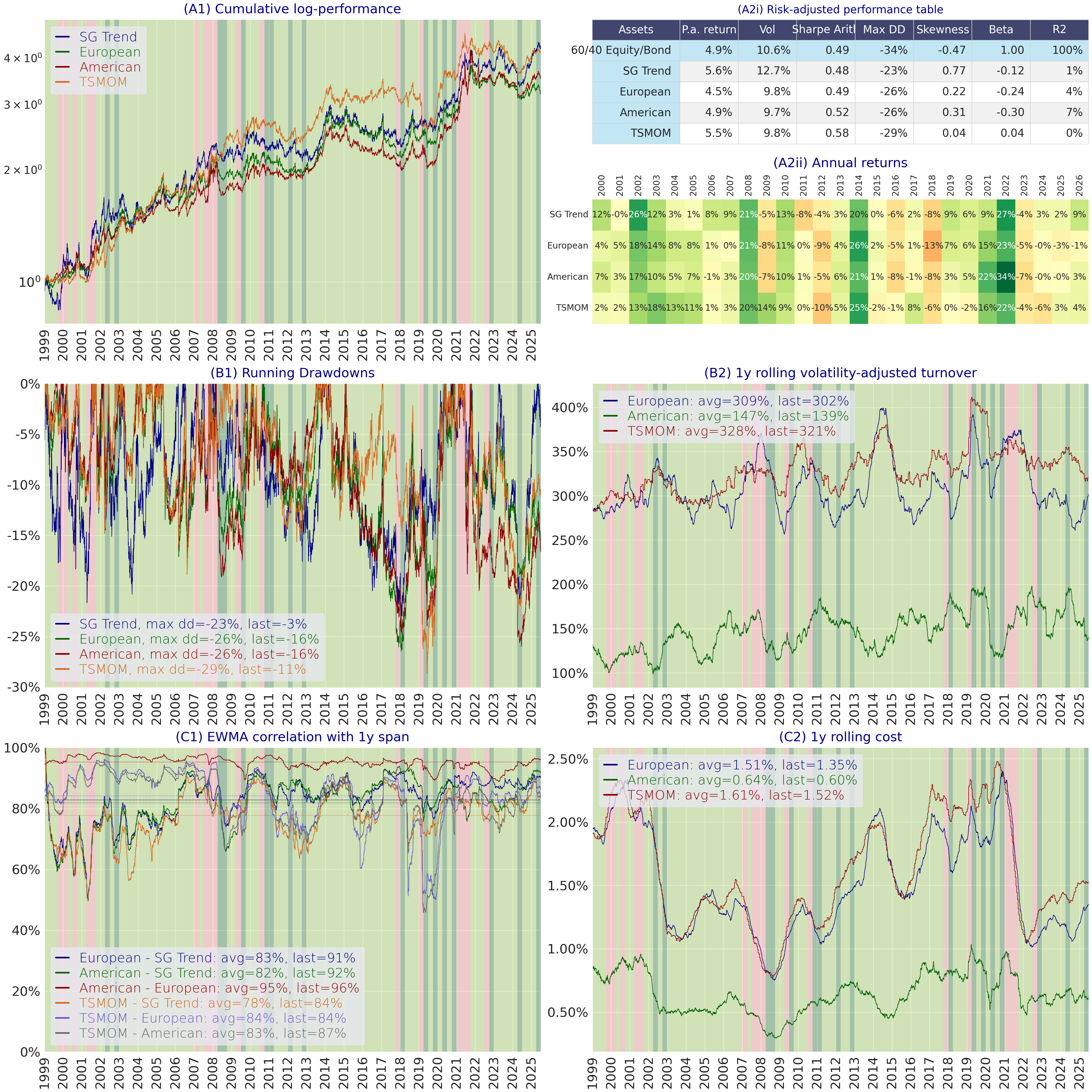}\hspace*{-2.5\baselineskip} 
	\end{center}
	\caption{Simulated performance of European, American and TSMOM systems along with the historical performance of SG Trend Index. Panels (A1), (B1), and (C1) show the cumulative log-performance, running drawdown, and EWMA correlations with one year span. Panel (A2i) shows the risk-adjusted performance table: P.a. return is the annualized return (CAGR), Vol is the annualized volatility of daily log returns, Sharpe Arith is the annualized Sharpe ratio of daily excess returns in the arithmetic convention of equation \eqref{eq:sr_def}, Max DD is the maximum drawdown, Skew is the skewness of quarterly log returns, and $\beta$ and $R^{2}$ are the slope and the $R^{2}$ of the regression of monthly returns on the 60/40 equity/bond portfolio. Panel (A2ii) shows annual returns. Panels (B2) and (C2) show one year rolling volatility-normalized turnover and cost, respectively. The background color orders the quarterly returns of the benchmark 60/40 portfolio and splits the $16\%$ of worst returns into the bear regime (pink), the $16\%$ of best returns into the bull regime (dark green), and the remaining quarters into the normal regime (light green). The period of performance measurement is from 31 December 1999 to 10 July 2026}
	\label{fig:tf_sg_backtest_paper}
\end{figure}

\subsection{Attribution of Realized Sharpe Ratios to Autocorrelation and Drift}\label{sec:attr}

Corollary \ref{col:sr} expresses the Sharpe ratio of the European TF system through two inputs: the autocorrelation function and the drift of the volatility-normalized returns. We estimate both inputs from the futures data and test the formula against the realized backtests of all three TF systems. For each of the 84 futures contracts, we compute the volatility-normalized returns $z_{t}$ with the EWMA span of 33 days, skip the first 250 days, and require at least six years of history. The signal filters initialize at zero at inception. The missing fraction of the expected response to a constant mean is $\nu^{t}$, while the standard-deviation scale under independent zero-mean input is $\sqrt{1-\nu^{2t}}$. At the 520-day span on the first included day, the missing mean response is $38\%$ and the standard-deviation attenuation is $8\%$, declining to $5.6\%$ and below $0.2\%$ within two further years and averaging out over the six-year-minimum samples. The liquidity and history screens may tilt the estimates toward heavily traded markets. We estimate the sample autocorrelation function of $z_{t}$ to 780 lags and the sample mean of $z_{t}$, both over the full instrument sample.

The prediction applies equation \eqref{eq:sr_main} with the sample autocorrelation function, the sample variance $\hat{\gamma}_{T}(0)$, the variance-normalized annualized drift $\hat{\mu}^{z}_{an} = \sqrt{a}\,\overline{z}_{T}/\sqrt{\hat{\gamma}_{T}(0)}$, and the kurtosis loading $\kappa=0$, which Section \ref{ssec:sr_mc} justifies as a second-order effect. At the 520-day span, the truncation weight is $\nu^{780} \approx 0.05$, so the truncation may omit a small residual long-memory contribution at the slowest span. The validation is nonparametric, so no parametric return model is estimated. We compute two predictions for each filter span: the autocorrelation channel alone, with the drift set to zero, and the total, which adds the sample drift. Treating the population variance as known, the squared sample drift carries an upward bias of $a\,\Var[\overline{z}_{T}]/\vartheta$, which equals $a/T \approx 0.17$ in the squared annualized drift at the six-year minimum sample under independent unit-variance returns. With the sample variance in the denominator, the expression is a first-order approximation, because the variance estimate carries its own error, and serial dependence replaces the variance of the sample mean by its long-run counterpart. The sample drift square is therefore a valid component of the in-sample decomposition and not an unbiased estimate of the population drift channel, and a forward-looking application should debias or shrink it. We compare the predictions with the Sharpe ratios of the gross backtests on the same sample, which isolates the accuracy of the formula from the accuracy of the turnover approximation. The comparison is in-sample by design: we test whether the closed form reproduces the realized performance on the same sample, including the pipeline effects that the population formulas do not capture. The out-of-sample application is a subject for future research.

Figure \ref{fig:attribution_scatter} pools the total predictions against the realized Sharpe ratios across instruments and spans. For the European system in panel (A), the pooled correlation is $0.99$ and the regression slope is $0.96$, so the formula reproduces the realized Sharpe ratio at the instrument level. Panels (B) and (C) apply the same European predictions to the realized Sharpe ratios of the American and TSMOM systems at matched lookbacks. We set the American short span to two days and the TSMOM period length to one day with the number of periods equal to the span, which a parameter search on the same sample identifies as the closest discretized counterparts of the European filter, with details in the companion repository, so the cross-system comparison is an in-sample calibrated exercise rather than an independent validation.

The TSMOM system tracks the predictions with a correlation of $0.89$ and a slope of $0.73$ across all spans. At the leading order, the Gaussian sign-moment identity gives a benchmark slope of $\sqrt{2/\pi}\approx 0.80$ for the sign filter, so the measured TSMOM slope of $0.73$ sits close to the Gaussian benchmark\footnote{For a jointly Gaussian pair with a zero-mean unit-variance signal, $\mathbb{E}\left[\sgn\left(S\right)z\right]=\mathbb{E}\left[Sz\right]\mathbb{E}\left[\left|S\right|\right]=\sqrt{2/\pi}\,\mathbb{E}\left[Sz\right]$, which holds exactly and is the Bussgang theorem for the hard limiter. The sign transform scales the autocorrelation channel by $\sqrt{2/\pi}$ and leaves the variance at one to the leading order, because $\sgn^{2}\left(S\right)=1$, so the leading order refers to the Sharpe ratio mapping rather than to the moment identity.}.

The American system tracks the predictions with a correlation of $0.92$ and a slope of $0.61$ at the spans above one month, while the trailing stops and the entry buffers disengage the system from the signal at faster spans. The slopes below one measure a combined implementation discount, which mixes the nonlinear signal transformation, the lookback mapping, the entry buffers and stops, and the different risk denominators.

Figure \ref{fig:attribution_medians} shows the cross-sectional medians by span with the decomposition of the prediction. The autocorrelation channel produces median Sharpe ratios from $0.55$ at the 5-day span to $0.33$ at the two-year span. The drift channel adds up to $0.08$, and the contribution grows with the span as equation \eqref{eq:sr_wn} predicts. The realized European medians run from $0.48$ to $0.39$ and sit $4\%$ to $14\%$ below the total prediction, with the largest gaps at short spans, where the error of the volatility estimator has the largest effect on fast signals, consistent with the attenuation of Section \ref{ssec:sr_mc}. The realized TSMOM medians run parallel to the prediction across all spans at the discount of the sign filter. The realized American medians converge to the prediction as the span grows. Thus both the TSMOM and the American TF systems can be viewed as a slow-span implementation of the same alpha.

\begin{figure}[!htb]
	\begin{center}
		\includegraphics[width=1.0\textwidth, angle=0]{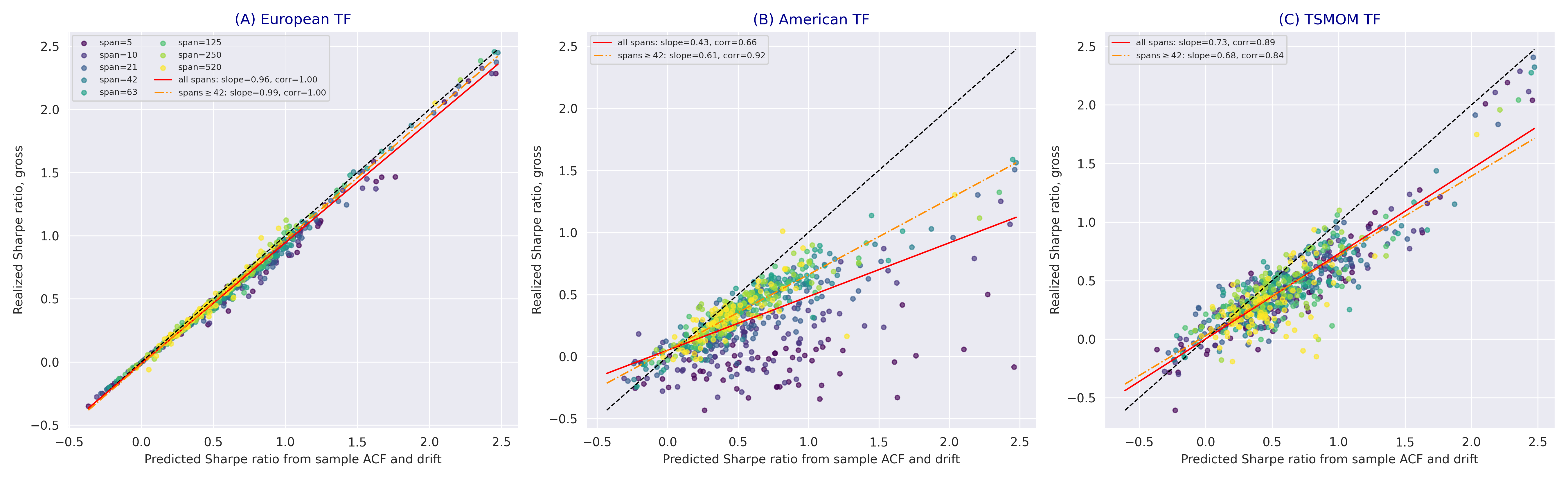}
	\end{center}
	\caption{The European closed form predicts realized Sharpe ratios across instruments, spans, and system designs. Each panel compares the Sharpe ratio predicted from the sample autocorrelation function and the sample drift of $z_{t}$ in equation \eqref{eq:sr_main} with the realized Sharpe ratio of the gross backtest, for the 84 futures contracts and the filter spans from 5 to 520 days, with points colored by span. Panel (A) shows the European single-filter system, panel (B) the American system with the short span of two days and the buffer and stop parameters of Section \ref{sc:grid}, and panel (C) the TSMOM system with daily signs and the number of periods equal to the span. The dashed line is the diagonal, the red line is the pooled regression over all spans, and the orange line is the regression over the spans of 42 days and above. Samples run from the instrument inception plus 250 warmup days to 10 July 2026, with a minimum of six years, and the points are dependent across spans within an instrument. Sharpe ratios are annualized with $a=260$ and use gross excess returns of the packaged dataset of the companion repository}\label{fig:attribution_scatter}
\end{figure}

\begin{figure}[!htb]
	\begin{center}
		\includegraphics[width=1.0\textwidth, angle=0]{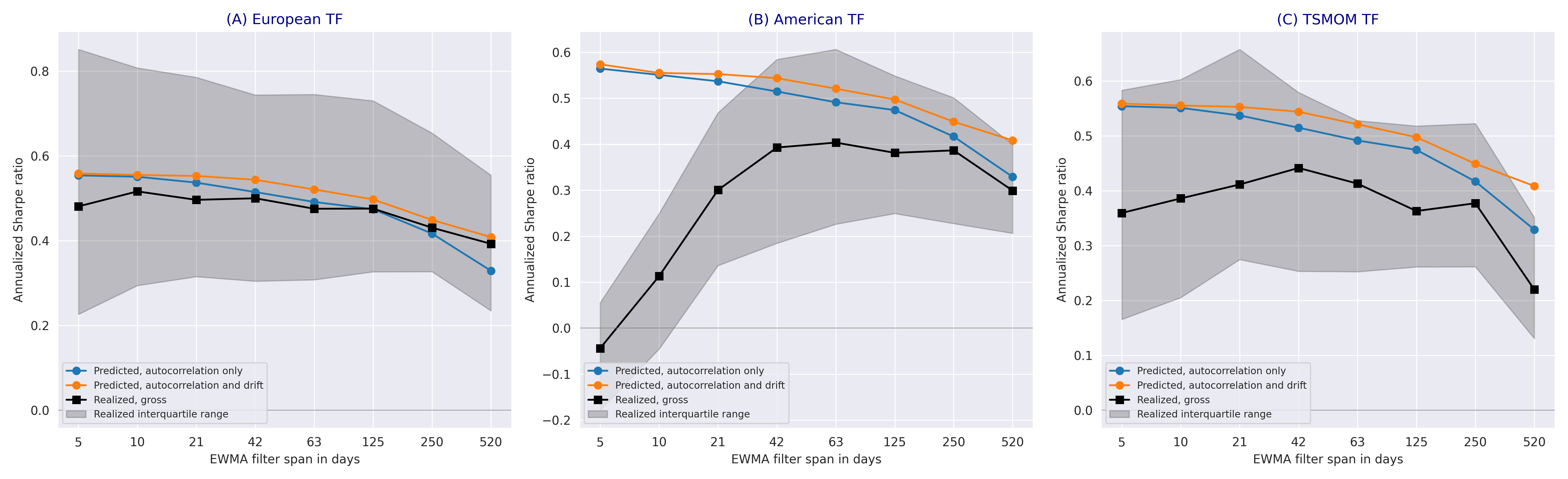}
	\end{center}
	\caption{The trend-following alpha shifts from autocorrelation at short spans to drift at long spans, and each system design monetizes it at its implementation discount. Each panel shows the cross-sectional medians of the predicted and realized Sharpe ratios by filter span, with the interquartile band of the realized values, for the systems and samples of Figure \ref{fig:attribution_scatter}. The predicted medians decompose into the autocorrelation-only channel, which sets the drift to zero, and the total prediction, which adds the sample drift. The predicted lines are identical across the panels, because all panels apply the European closed form of equation \eqref{eq:sr_main}}\label{fig:attribution_medians}
\end{figure}

The attribution supports the central claim of the paper with market data. Trend-following performance on liquid futures is primarily an autocorrelation phenomenon, and, in the realized in-sample decomposition, the drift channel contributes a secondary share that grows with the span. The same closed form ranks the performance across all three system designs, so the American and TSMOM systems monetize a closely related signal, and the closed form transfers to these nonlinear designs empirically rather than mathematically. The in-sample correlation validates the moment reconstruction and the limited size of the boundary, initialization, and pipeline effects, while the out-of-sample forecasting value of the estimated inputs remains to be tested. The EWMA lag-1 autocorrelation of $z_{t}$ across the universe declines from about $0.04$ in the 1990s to about $0.01$ after 2010, with the estimation reproducible from the companion repository. Because the autocorrelation input varies slowly across decades, the estimated autocorrelation function is a candidate input for conditional span selection, and we recommend estimating the inputs on rolling windows with a debiased drift, because the formula applies conditionally on the prevailing autocorrelation function.

\subsection{Skewness of Aggregated Returns}\label{ssec:skew}

The product structure of the daily return generates positive skewness under time aggregation. The daily return multiplies the lagged signal by the current return, so the $T$-day cumulative return loads on the realized autocovariance of the volatility-normalized returns at the filter lags. This loading makes the aggregated return a convex payoff on the realized trend, in line with the convexity of TF portfolios documented by \cite{Bouchaud2017}, and \cite{Potters2006} document the resulting positive skewness of the exponential-average system on a random walk. Appendix~\ref{app:skew} derives the third moment of the aggregated return in closed form under white noise from a single pairing identity. Proposition~\ref{pr:skew} gives the skewness of the $T$-day return of the single-filter system in equation \eqref{eq:skew_t}: the skewness is zero at $T=1$, positive at every horizon $T\geq 2$ although the expected return is zero, hump-shaped with the maximum near half the filter span, and independent of the loading and of the volatility target.

\begin{figure}[ht]
	\begin{center}
		\includegraphics[width=0.99\textwidth, angle=0] {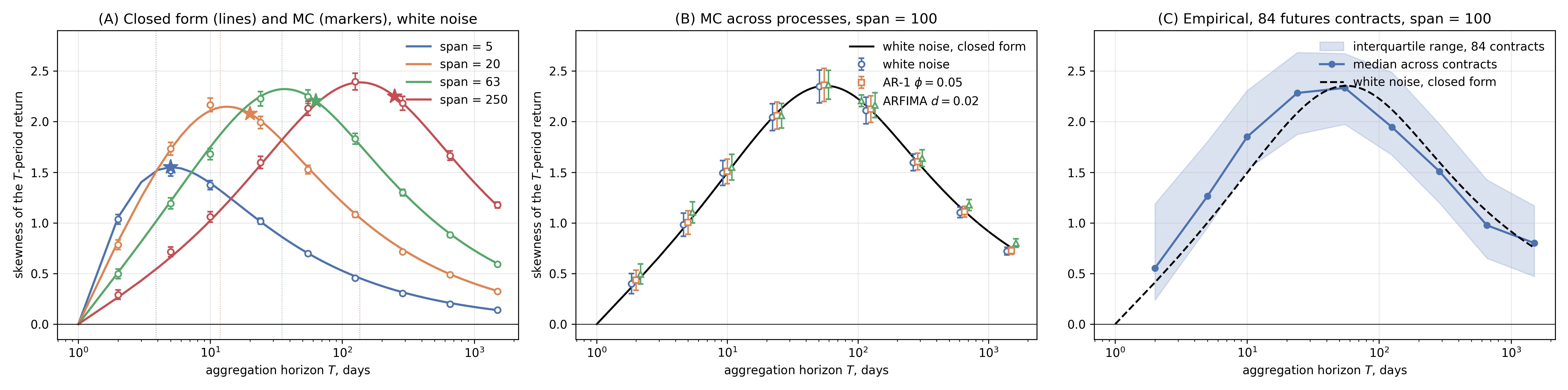}
	\end{center}
	\caption{The skewness of aggregated TF returns is positive at every horizon beyond one day and peaks near half the filter span, in closed form, in simulation, and in the data. Panel (A) reports the closed-form skewness of equation \eqref{eq:skew_t} under white noise for the filter spans of $5$, $20$, $63$, and $250$ days (lines), MC estimates over $400,000$ paths with $95\%$ confidence intervals (markers with error bars), the peak horizons of the closed form (dotted vertical lines), and the horizon equal to the filter span (stars), with the third central moment normalized by the known variance $T$. Panel (B) reports MC estimates at the span of $100$ days for the three processes of Section \ref{sec:sp} with zero drift, white noise, the AR-1 process with $\phi=0.05$, and the ARFIMA process with $d=0.02$ and $\phi=0$, over $100,000$ paths with the standardized sample skewness, together with the white-noise closed form (line). Panel (C) reports the empirical skewness of the gross single-filter European system with the span of $100$ days on the $84$ futures contracts of Table \ref{fig:universe}: the cross-sectional median and interquartile range of the standardized sample skewness of overlapping $T$-day sums, for contracts with at least six non-overlapping lengths of history at each horizon, against the white-noise closed form. Samples run from the instrument inception plus warmup to 10 July 2026 and use the gross excess returns of the packaged dataset of the companion repository}
	\label{fig:aggregated_skewness}
\end{figure}

Panels (A) and (B) of Figure~\ref{fig:aggregated_skewness} verify and extend the closed form by simulation. Panel (A) confirms equation \eqref{eq:skew_t} at every span and horizon. Panel (B) places the three processes of Section \ref{sec:sp} against the white-noise closed form: the AR-1 estimates track the closed form at every horizon, long memory raises the skewness at horizons short relative to the span, and at long horizons the estimates of all three processes converge to the closed form. The skewness of aggregated TF returns is therefore structural: the filter design drives the profile more than the data-generating process.

The approach is analytic beyond white noise. Under the AR-1 process, every covariance in the pairing expansion of the third moment is geometric, so the aggregated skewness admits a closed form, and the expansion generalizes to the ARFIMA process. We leave these derivations, together with the asymptotics of the horizon profile, to a follow-up paper.

The closed form is also robust to the distribution of the innovations. For independent innovations with zero mean, unit variance, and zero third moment, every surviving pairing in the third moment uses only second moments, so equation \eqref{eq:skew_t} holds without Gaussianity and without a finite-kurtosis condition. A non-zero third moment $\gamma_{3}$ of the innovations adds the single term $T\gamma^{2}_{3}\sum_{j\geq 1}w^{3}_{j}$ to the third moment of the aggregated return, with the filter weights $w_{j}$ of Appendix~\ref{app:skew}, and this term decays in the skewness as $T^{-1/2}$. This robustness supports the application of the formula to the heavy-tailed futures returns of panel (C).

Panel (C) takes the profile to the data. For each of the $84$ contracts, we compute the gross daily returns of the single-filter European system with the span of $100$ days and zero costs, aggregate them into overlapping $T$-day sums, and estimate the standardized sample skewness at each horizon. Costs shift the mean of the aggregated return and leave the skewness almost unchanged, because the skewness is invariant to the mean and the cost term is small relative to the daily return: in simulation, a cost rate of $100$bp per unit of volatility-normalized turnover changes the white-noise skewness by less than $0.002$ at every horizon. The cross-sectional median tracks the closed form across the full horizon range and attains $2.33$ at the horizon of $55$ days against the closed-form value of $2.35$, with the peak of the closed form near $55$ days. The interquartile range lies above zero at every horizon, so positive aggregated skewness is a property of the cross-section rather than of the median contract. The median sits above the closed form at short horizons, where the heavy tails and the autocorrelations of the data enter, and converges to the closed form at long horizons, in line with panel (B). Non-overlapping $T$-day blocks show the same pattern, with the medians within $0.06$ of the overlapping estimates at every horizon with sufficient data.

The proposition and the panels turn a documented property into a structural one. The positive skewness requires no autocorrelation, no drift, and no alpha, because it holds exactly under white noise where the expected return of the system is zero. The horizon profile explains why daily TF returns appear symmetric while monthly and quarterly returns are right-skewed, and it places the strongest asymmetry at horizons near half the filter span. The positive quarterly skewness of the three systems in Figure~\ref{fig:tf_sg_backtest_paper}, against the negative skewness of the balanced benchmark, is consistent with this profile. The quarterly skewness of the systems in panel (A2i) of Figure~\ref{fig:tf_sg_backtest_paper} sits below the single-contract profile for two reasons that the framework quantifies. The long-short filter shifts the profile toward longer horizons, with the white-noise skewness of the LS(250,20) system at $1.7$ at the quarterly horizon against $2.3$ for the single filter, and the post-2000 sample realizes a quarterly skewness of $0.5$ against $1.8$ over the full history of the same portfolio on our data. Cross-sectional aggregation itself preserves the skewness through the common trend component, while portfolio-level volatility targeting, which the systems of Figure~\ref{fig:tf_sg_backtest_paper} do not apply, would reduce it further.

\section{Conclusions}

We have classified trend-following systems into European, American, and Time Series Momentum designs, and we have developed a unified analytical framework around the European system, which sizes positions with an EWMA filter of volatility-normalized returns. We have derived an exact sample-path identity for the cumulative return and closed-form gross and leading-order net Sharpe ratios under a generic stationary autocorrelation function within the linear-process class.

The closed forms separate the sources of performance. Under white noise with drift, the net Sharpe ratio under the signal-turnover proxy converges to the absolute Sharpe ratio of the instrument at long spans for any proportional cost. Under the AR-1 process, the gross Sharpe ratio is approximately $2\phi\sqrt{a/\spn}$, the break-even cost has the span-invariant limit $c^{*}_{\infty}=\sqrt{\pi/(2a)}\,\phi/(1-\phi)$, and costs determine the viability of a short-memory alpha to first order. Under the ARFIMA process with memory parameter $d$, the net Sharpe ratio attains an interior cost-optimal span, which scales as $c^{1/(2d)}$ in the large-cost limit with the net-to-gross ratio converging to $4d/(1+2d)$.

We have verified the framework by simulation and on market data. Within the full pipeline, the formulas track the simulated gross and net Sharpe ratios within $0.05$. On 84 liquid futures contracts, the sample autocorrelation function and drift reproduce the realized Sharpe ratios of the European system with a pooled correlation of $0.99$ and a slope of $0.96$, and they rank the TSMOM and American systems at design-specific implementation discounts. Trend-following performance is primarily an autocorrelation phenomenon with a drift share that grows with the span. At realized costs of $40$bp to $60$bp per unit of turnover, the proxy break-even cost implied by the observed short-lag autocorrelation lies below the realized cost level, while the long-memory component supports positive net performance at the one-to-three-month spans. We prove that under white noise the skewness of the aggregated return admits the closed form of Proposition~\ref{pr:skew}, positive at every horizon beyond one day and peaking near half the filter span. We show that under general empirical and theoretical data generation processes that the TF skewness is positive, so the defensive right tail of TF returns is structural.

The framework connects trend-following to the theory of stationary processes: a TF system is a spectral instrument, whose filter span selects the bandwidth of a Poisson kernel and whose alpha is the excess spectral mass of the volatility-normalized returns at low frequencies. The cost-optimal span solves a bandwidth-selection problem against the spectral shape and the cost level, with the AR-1 and ARFIMA asymptotics as its two tractable instances. The sign-based designs point to an extension to nonlinear filters, and we develop the diversification benefits of TF portfolios for long-only investors in a companion paper.

\paragraph{Acknowledgments}

The authors are thankful to Yoav Git for useful discussions and constructive suggestions. The authors are also grateful to Rupert Goodwin and members of the London Quant Group (LQG) for useful feedback at the LQG seminar on 13 May 2025. An early draft of this paper was circulated under the title ``Trend-Following Strategies for Tail-Risk Hedging and Alpha Generation''.

\section*{Declarations}

\paragraph{Funding}
No funding was received for conducting this study.

\paragraph{Competing interests}
The authors have no competing interests to declare that are relevant to the content of this article.

\paragraph{Data availability}
The futures dataset used in the empirical analysis is included in the companion repository. All analytical and Monte Carlo results require no data.

\paragraph{Code availability}
The Python package \texttt{trendfollowing}, which reproduces all analytical results and exhibits of this paper, is available at \url{https://github.com/ArturSepp/TrendFollowingSystems}.

\paragraph{Use of generative artificial intelligence}
The authors used generative artificial intelligence tools in the preparation of this article, including the checking and numerical verification of derivations, AI-assisted review reports on the mathematical content, the design and testing of the companion code, and the editing of the text. The authors verified all derivations, numerical results, and text. The authors assume responsibility for all content.
\appendix

\section{American and TSMOM TF Systems}\label{app:tfs}

The American TF system applies the average true range (ATR) of daily prices for the signal buffer, the position size, and the stop-loss levels.

\begin{definition}[Average True Range (ATR)]
	The daily true range is computed by:
	\begin{equation} \label{eq:int4} 
		\begin{split}
			tr_{t} & = \max\left\{\left|hl_{t} \right|, \left|hc_{t} \right| , \left|lc_{t} \right|  \right\}, \\
			hl_{t}& = s^{high}_{t} - s^{low}_{t}, \ hc_{t} = s^{high}_{t} - s^{close}_{t-1}, \ lc_{t} = s^{low}_{t} - s^{close}_{t-1},
		\end{split}
	\end{equation}
	where $s^{high}_{t}, s^{low}_{t}, s^{close}_{t}$ are daily high, low, and close prices at the $t$-th day.
	
	ATR over a period of $N$ days is computed as the average of daily true ranges by:
	\begin{equation} \label{eq:int5}
		ATR_{t} =  \frac{1}{N}\sum^{N-1}_{n=0} tr_{t-n}.
	\end{equation}
\end{definition}

ATR is a dimensional, price-measured variable, similar to the price volatility in equation \eqref{eq:pricevol}. To construct an ATR-based estimator analogous to the EWMA returns-volatility estimator in equation \eqref{eq:ewmavol}, we introduce the relative true range and its EWMA filter.

\begin{definition}[Relative True Range (RTR)]
	Daily RTR is computed by:
	\begin{equation} \label{eq:int6} 
		rtr_{t}  = \frac{tr_{t}}{s^{close}_{t-1}}.
	\end{equation}
	
	EWMA of RTR is denoted by ERTR and computed using EWMA filter:
	\begin{equation} \label{eq:int8}
		ERTR_{t} =   \mathcal{L}^{(\nu_{etr})}\left(\frac{tr_{t}}{s^{close}_{t-1}}\right),
	\end{equation}
	where $\nu_{etr}$ is the smoothing parameter. The ERTR provides the volatility input of the American TF system, as the EWMA volatility of equation \eqref{eq:ewmavol} serves the European system.
	
\end{definition}

\subsection{American TF System}

For the American TF system, we work with the sequence of daily prices $\{s_{t}\}$ as follows.

\begin{definition}[American TF System]\label{as:am} ~
	
	\begin{itemize}
		
		\item The long (or slow) filter $\overline{s}^{(\nu_{1})}$ and the short (or fast) filter $\overline{s}^{(\nu_{2})}$ are computed using the EWMA filter in equation \eqref{eq:ewma1} applied to the continuous prices of futures contracts $s_{t}$, respectively:
		\begin{equation} \label{eq:as1} 
			\overline{s}^{(\nu_{1})} = \mathcal{L}^{(\nu_{1})} (s_{t}), \ \overline{s}^{(\nu_{2})} = \mathcal{L}^{(\nu_{2})} (s_{t}),
		\end{equation}
		where $\nu_{1} > \nu_{2}$.

		\item If the current position is zero, the long position is initiated when the fast filter exceeds the slow filter plus the signal buffer as follows:
		\begin{equation} \label{eq:as3} 
			\overline{s}^{(\nu_{2})} > \overline{s}^{(\nu_{1})} + \omega \cdot ATR_{t},
		\end{equation}
		where $\omega$ is the entry-point width parameter. The position size $p^{long}$ and the stop loss $sl^{long}_{t}$ are initiated respectively by:
		\begin{equation} \label{eq:as3a} 
			p^{long} = R \frac{s_{t}}{ATR_{t}}, \  sl^{long}_{t} = s_{t} - p \cdot ATR_{t},
		\end{equation}
		where $R$ is the risk multiple and $p$ is the stop-loss width parameter.

		The short position is initiated when the fast filter drops below the slow filter minus the buffer:
		\begin{equation} \label{eq:as4} 
			\overline{s}^{(\nu_{2})} < \overline{s}^{(\nu_{1})} - \omega \cdot ATR_{t},
		\end{equation}
		where the entry-point scale $\omega$ matches the long side in equation \eqref{eq:as3} for robustness. The position size $p^{short}$ and the stop loss are initiated by:
		\begin{equation} \label{eq:as4a} 
			p^{short} =  - R \frac{s_{t}}{ATR_{t}}, \  sl^{short}_{t} = s_{t} + p \cdot ATR_{t},
		\end{equation}
		with the risk multiple $R$ and the stop-loss width $p$ of equation \eqref{eq:as3a}.
		
		\item If the current position is non-zero, the long position is closed when the stop-loss is breached and the signal in equation \eqref{eq:as3} is off\footnote{A stop breach while the signal of equation \eqref{eq:as3} remains on would re-create the position on the next trading day, so the exit requires both the stop breach and the signal turning off.}:
		\begin{equation} \label{eq:as5} 
			s_{t} < sl^{long}_{t-1}  \ \& \ \overline{s}^{(\nu_{2})}  \leq \overline{s}^{(\nu_{1})} + \omega \cdot ATR_{t},
		\end{equation}
		otherwise the stop-loss is updated as follows:
		\begin{equation} \label{eq:as5a} 
			sl^{long}_{t} = \max\left\{ sl^{long}_{t-1},  s_{t} - p \cdot ATR_{t} \right\}.
		\end{equation}

		The short position is closed if the stop-loss is breached and the condition in equation \eqref{eq:as4} is off:
		\begin{equation} \label{eq:as6} 
			s_{t} > sl^{short}_{t-1}  \ \& \ \overline{s}^{(\nu_{2})}  \geq \overline{s}^{(\nu_{1})} - \omega \cdot ATR_{t},
		\end{equation}
		otherwise the stop-loss is updated by the analogous minimum:
		\begin{equation} \label{eq:as6a} 
			sl^{short}_{t} = \min\left\{ sl^{short}_{t-1},  s_{t} + p \cdot ATR_{t} \right\}.
		\end{equation}
		
	\end{itemize}
	
\end{definition}

The risk multiple $R$ sets the notional risk budget per trade: the initial stop distance $p \cdot ATR_{t}/s_{t}$ times the weight $p^{long}$ equals $R\,p$ of the notional.

The American system differs from the European system in two ways. The position size is fixed at trade inception using the current ATR, and the exit is a trailing stop-loss rather than a signal reversal, so a continuing trend raises the stop and defers the exit. The position quantity $p^{long}$ in equation \eqref{eq:as3a} is a dimensionless notional weight, so we state the per-trade P\&L per unit of notional capital. The long position exits at the first day $\tau^{long}$ on which both conditions of equation \eqref{eq:as5} hold, and the short position at the first day $\tau^{short}$ on which both conditions of equation \eqref{eq:as6} hold, so the P\&L of the trade is exact:
\begin{equation} \label{eq:as7} 
	\text{P\&L}^{long} = p^{long} \left(\frac{s_{\tau^{long}}}{s_{0}} - 1  \right), \qquad \text{P\&L}^{short} = p^{short} \left(\frac{s_{\tau^{short}}}{s_{0}} - 1 \right),
\end{equation}
where $s_{0}$ is the entry price at the day $t_{0}$. The recursion in equation \eqref{eq:as5a} unwinds over the life of the trade to $sl^{long}_{t} = \max_{t_{0} \leq u \leq t}\left\{ s_{u} - p \cdot ATR_{u} \right\}$, so the exit price is close to $s_{max} - p \cdot ATR$ for the long position and to $s_{min} + p \cdot ATR$ for the short position, with the maximum price $s_{max}$ and the minimum price $s_{min}$ over the life of the trade. Both expressions are heuristics, and they require three conditions: the entry signal is off already at the first breach of the stop, by equation \eqref{eq:as3} for the long position and by equation \eqref{eq:as4} for the short position, the ATR is approximately constant over the life of the trade, and the price gaps and the execution slippage are ignored. The backtest applies the notional weights and accrues the P\&L through the weighted returns, with the filters, the stops, and the ATR of the day $t$ determining the position that earns the return of the day $t+1$.

\subsection{Time Series Momentum (TSMOM) System}

\cite{Moskowitz2012}, \cite{Hurst2013}, and \cite{Baltas2015} construct the TSMOM signal as the average sign of rolling 12 monthly returns:
\begin{equation} \label{eq:tsmom1} 
	w_{t} = \frac{1}{M} \sum^{M}_{m=1} \sgn\left(\frac{s_{t-m+1}}{s_{t-m-11}}-1\right) \frac{\sigma_{\mathrm{target}}}{\sigma^{an}_{t-1}},
\end{equation}
where $\sgn(x)$ is the sign function, $\{s_{t}\}$ is the set of month-end prices, $\sigma_{\mathrm{target}}$ is the annualized volatility target, and $\sigma^{an}_{t}$ is the annualized volatility estimate. \cite{Dudler2015} extend this approach to averaging the signs of volatility-normalized returns.

We generalize this approach as follows. We fix a period length of $L$ days and define the rebalancing grid $\{t'\}$ as the multiples of $L$, $t' \equiv 0 \ \text{(mod L)}$. We set the lookback to $M$ periods, so the signal aggregates $ML$ daily observations.  We replace the sign of the aggregated return with the sum of the signs of the period returns, which yields the exact unit-variance normalization below.

\begin{definition}[TSMOM] The signal is the normalized sum of the signs of daily volatility-normalized returns over the lookback of $M$ periods with $L$ days each:
	\begin{equation} \label{eq:tsmom2} 
		w_{t'} = \frac{\sigma_{\mathrm{target}}}{\sigma^{an}_{t'-1}}\, S_{t'} , \quad
		S_{t'}  = \frac{1}{\sqrt{ML}} \sum^{M}_{m=1} b_{t'-(m-1)L}, \quad b_{t''} = \sum^{t''}_{t=t''-L+1} \sgn\left(\frac{r_{t}}{\sigma^{an}_{t-1}}\right),
	\end{equation}
	where $\sigma^{an}_{t}$ is the annualized volatility of daily returns. We extrapolate $w_{t'}$ to the daily grid $\{t\}$ by forward fill. The weight $w_{t'}$ uses returns through day $t'$ and applies from the return of day $t'+1$ onward, so the forward fill introduces no same-day look-ahead.
\end{definition}

In equation \eqref{eq:tsmom2}, the double sum contains $ML$ daily signs. Under serially independent returns whose distribution is continuous at zero with zero median, the signs are independent Rademacher variables with unit variance, so the sum has variance $ML$ and the normalized signal $S_{t'}$ is a $z$-score. An atom at zero would lower the variance of the sign. The annualized volatility of the system return for each instrument is therefore approximately $\sigma_{\mathrm{target}}$ independently of $M$ and $L$, and the sign function removes the volatility estimator from the signal because $\sgn(r_{t}/\sigma^{an}_{t-1})=\sgn(r_{t})$. The double sum reindexes to a single sum of $ML$ daily signs over disjoint blocks, so the signal value depends on the meta parameters $(L, M)$ only through the product $ML$, while $L$ sets the rebalancing frequency.

\section{Proofs}\label{app:proofs}

\subsection{Proof of Proposition \ref{prop1}}\label{sec:prop1}

We consider a sequence of serially independent random variables $\{y_t\}$ such that $\mathbb{E}[y_t]=\mu_{0}$ and $\Var[y_t]=\vartheta_{0}$. Using the EWMA filter in equation \eqref{eq:ewma1}, we obtain for the expected value
\begin{equation} \label{eq:e3a} 
	\mathbb{E}\left[ \mathcal{L}^{(\nu)} (y_{n}) \right]   =  \left(1-\nu \right)\sum^{\infty}_{m=0} \nu^{m} \mathbb{E}\left[y_{n-m}\right] =\mu_{0}
\end{equation}
and for the variance:
\begin{equation} \label{eq:e3c} 
	\Var\left[ \mathcal{L}^{(\nu)} (y_{n}) \right]   = \left( 1-\nu\right)^{2} \sum^{\infty}_{m=0} \nu^{2m} \Var\left[y_{n-m}\right] =   \frac{\left(1-\nu \right)^{2} }{ 1-\nu^{2}} \vartheta_{0}=   \frac{ 1-\nu   }{ 1+\nu} \vartheta_{0}.
\end{equation}

Thus, for the variance-preserving filter in equation \eqref{eq:ewma2}, $\Var[\tilde{\mathcal{L}}^{(\nu)} (y_{t})] = \vartheta_0$.

\subsection{Proof of Proposition \ref{prop2}}\label{sec:prop2}

Assuming a sequence of serially independent random variables $\{y_t\}$ as in Appendix \ref{sec:prop1}, we compute the variance of the long-short EWMA filter in equation \eqref{eq:ewma3} as follows:
\begin{equation} \label{eq:ls2} 
	\begin{split}
		& \Var\left[ \widetilde{\mathcal{LS}}^{(\nu_{1}, \nu_{2})}(y_{n})  \right] \equiv \Var\left[ \widetilde{l}_{1} \widetilde{\mathcal{L}}^{(\nu_{1})} (y_{n}) - \widetilde{l}_{2}\widetilde{\mathcal{L}}^{(\nu_{2})} (y_{n}) \right]\\
		&   = \widetilde{l}^{2}_{1}\vartheta_{0}+  \widetilde{l}^{2}_{2}\vartheta_{0} - 2  \widetilde{l}_{1}\widetilde{l}_{2}\sqrt{\frac{ 1+\nu_{1}   }{ 1-\nu_{1}}} \sqrt{\frac{ 1+\nu_{2}   }{ 1-\nu_{2}}} \Cov\left[ \mathcal{L}^{(\nu_{1})} (y_{n}),\mathcal{L}^{(\nu_{2})} (y_{n}) \right]\\
		&   = \widetilde{l}^{2}_{1}\vartheta_{0}+  \widetilde{l}^{2}_{2}\vartheta_{0} - 2  \widetilde{l}_{1}\widetilde{l}_{2}\sqrt{ 1-\nu^{2}_{1}}\sqrt{1-\nu^{2}_{2}} \frac{1}{1-\nu_{1}\nu_{2}}\vartheta_{0},
	\end{split}
\end{equation}
where we apply the covariance of the two raw filters under serial independence:
\begin{equation} \label{eq:ls3} 
	\begin{split}
		&\Cov\left[ \mathcal{L}^{(\nu_{1})} (y_{n}),\mathcal{L}^{(\nu_{2})} (y_{n}) \right]   = \left( 1-\nu_{1}\right) \left( 1-\nu_{2}\right) \sum^{\infty}_{m=0} \nu^{m}_{1} \nu^{m}_{2} \Var\left[y_{n-m}\right] \\
		& \qquad =\left( 1-\nu_{1}\right) \left( 1-\nu_{2}\right) \frac{1}{1-\nu_{1}\nu_{2}}\vartheta_{0}.
	\end{split}
\end{equation}

Substituting the loadings in equation \eqref{eq:ewma3a} with the variable $q$ into equation \eqref{eq:ls2} and setting the variance to $\vartheta_{0}$, we obtain the equation for $q$:
\begin{equation} \label{eq:ls5} 
	\left( \frac{1}{ 1-\nu^{2}_{1}}  + \frac{1}{ 1-\nu^{2}_{2} } - 2 \frac{1}{1-\nu_{1}\nu_{2}} \right)q^{2} = 1,
\end{equation}
with the solution given in equation \eqref{eq:ewma3a}.

\subsection{Proof of Proposition~\ref{pr:cumreturn}}\label{sec:cumreturn}

Using EWMA filter in equation \eqref{eq:ewma1}, we present $\mathcal{E}_{t}$ in equation \eqref{eq:e3} as follows:
\begin{equation} \label{eq:e3b} 
	\begin{split}
		\mathcal{E}_{t} & =      \nu z_{t}    \left(1-\nu \right)  \left(  z_{t-1} + \nu^{1} z_{t-2} + \nu^{2} z_{t-3}+... \right)   =  \left(1-\nu \right)  \sum^{\infty}_{m=0} \nu^{m} z_{t} z_{t-m} - \left(1-\nu \right)  z^{2}_{t},
	\end{split}
\end{equation}
where we add and subtract the $m=0$ term $\left(1-\nu\right)z^{2}_{t}$.

Applying the above result to the cumulative return $E(T)$ in equation \eqref{eq:e4a}
\begin{equation} \label{eq:e5} 
	E(T) = \sum^{T}_{t=1}  \mathcal{E}_{t}   =   \left(1-\nu \right)   \sum^{T}_{t=1}\sum^{\infty}_{m=0} \nu^{m} z_{t} z_{t-m} - \left(1-\nu \right)  \sum^{T}_{t=1} z^{2}_{t}.   
\end{equation}
The finite second moments make the first term absolutely convergent, because $\mathbb{E}\left|z_{t} z_{t-m}\right| \leq \vartheta + \mu^{2}$ by the Cauchy--Schwarz inequality and the weights $\nu^{m}$ are summable, so the convergence covers the ARFIMA processes with $d>0$, whose plain autocorrelation sum diverges, and justifies exchanging the two sums in equation \eqref{eq:e5}. Next, we apply the adjustment by the average $\overline{z}_{T}$ as follows:
\begin{equation} \label{eq:e7} 
	\begin{split}
		& \sum^{T}_{t=1} z_{t} z_{t-m}  = \sum^{T}_{t=1} \left( z_{t} -\overline{z}_{T}+\overline{z}_{T}\right) \left(z_{t-m} -\overline{z}_{T}+\overline{z}_{T}\right)\\
		& = \sum^{T}_{t=1} \left( z_{t} -\overline{z}_{T} \right) \left(z_{t-m} -\overline{z}_{T} \right)  + \overline{z}_{T}\Delta_{m} +\sum^{T}_{t=1} (\overline{z}_{T})^{2},
	\end{split}
\end{equation}
where the cross term $\sum_{t}\left(z_{t}-\overline{z}_{T}\right)\overline{z}_{T}$ vanishes by the definition of the sample mean, while the cross term $\sum_{t}\overline{z}_{T}\left(z_{t-m}-\overline{z}_{T}\right)$ leaves the boundary difference $\Delta_{m}=\sum^{T}_{t=1} z_{t-m} - T\overline{z}_{T}$ of the window shifted by $m$ lags, with $\Delta_{0}=0$.

For the first term, the geometric weights sum to one, so the weighted sum of squared means equals $\sum^{T}_{t=1} (\overline{z}_{T})^{2}$. The mean-adjustment for the second term in equation \eqref{eq:e5} is obtained using equation \eqref{eq:e7} with $m=0$, and its sum of squared means carries the weight $\left(1-\nu\right) \sum^{T}_{t=1} (\overline{z}_{T})^{2}$. Substituting equation \eqref{eq:e7} into equation \eqref{eq:e5}, we obtain:
\begin{equation} \label{eq:e10} 
	\begin{split}
		E(T) & =   \left(1-\nu \right)  \sum^{\infty}_{m=0} \nu^{m} \sum^{T}_{t=1}  \left( z_{t} -\overline{z}_{T} \right) \left(z_{t-m} -\overline{z}_{T} \right) \\
		& \qquad -  \left(1-\nu \right) \sum^{T}_{t=1} \left(z_{t}-\overline{z}_{T} \right)^{2} +  \left( 1-(1-\nu) \right)T(\overline{z}_{T})^{2} + \left(1-\nu \right) \sum^{\infty}_{m=1} \nu^{m}\, \overline{z}_{T}\Delta_{m}.
	\end{split}
\end{equation}

Rearranging the terms and collecting the boundary sum into $R_{T}$ of equation \eqref{eq:cumreturn_rt}, we obtain equation \eqref{eq:cumreturn}.

\subsection{Proof of Proposition \ref{pr:turnover1}}\label{sec:ta1}

Using equation \eqref{eq:eu3} for the weight, we consider the daily turnover as follows
\begin{equation} \label{eq:tur3}  
	\begin{split}
		U_{t} &  \equiv \sqrt{a}\sigma_{t}\left| w_{t} -  w_{t-1}\right| \\
		& = \sigma_{\mathrm{target}}\sigma_{t}\left|  \frac{1}{\sigma_{t}}\tilde{\mathcal{L}}^{(\nu)} (z_{t}) -  \frac{1}{\sigma_{t-1}}\tilde{\mathcal{L}}^{(\nu)} (z_{t-1})   \right| \\
		& \leq  \sigma_{\mathrm{target}}\left|  \tilde{\mathcal{L}}^{(\nu)} (z_{t}) - \tilde{\mathcal{L}}^{(\nu)} (z_{t-1})   \right| +  \sigma_{\mathrm{target}}\left| \tilde{\mathcal{L}}^{(\nu)} (z_{t-1}) \left(1 -  \frac{\sigma_{t}}{\sigma_{t-1}} \right)  \right|,
	\end{split}
\end{equation}
by adding $\pm \tilde{\mathcal{L}}^{(\nu)} (z_{t-1})$ in the second line: the first term is the signal-increment component and the second term is the volatility-update component.

For the first term, we obtain:
\begin{equation} \label{eq:tur4}  
	U^{(signal)}_{t}  = \sigma_{\mathrm{target}}\left|  \tilde{\mathcal{L}}^{(\nu)} (z_{t}) -  \tilde{\mathcal{L}}^{(\nu)} (z_{t-1})   \right| = \sigma_{\mathrm{target}}  \left|  (1-\nu)\sqrt{ \frac{ 1+\nu}{ 1-\nu} }z_{t}  + (\nu -1)\tilde{\mathcal{L}}^{(\nu)} (z_{t-1})  \right|. 
\end{equation}

Under the assumption that volatility-normalized returns $z_{t}$ are serially independent with variance $\vartheta_{0}$, we compute the variance of the signed weight increment $\Delta^{(signal)}_{t}$, with $U^{(signal)}_{t} = |\Delta^{(signal)}_{t}|$, as follows:
\begin{equation} \label{eq:tur5}  
	\begin{split}
		\Var \left[ \Delta^{(signal)}_{t}  \right] & =\sigma^{2}_{\mathrm{target}}  \left(  (1-\nu)^{2}\frac{ 1+\nu}{ 1-\nu} \vartheta_{0}  + (1-\nu)^{2}  \vartheta_{0} \right) =  2 \sigma^{2}_{\mathrm{target}}  \left(  1-\nu \right)   \vartheta_{0},
	\end{split}
\end{equation}
where we use equation \eqref{eq:e3c} for the variance of $\tilde{\mathcal{L}}(z_{t})$.

The increment $\Delta^{(signal)}_{t}$ is Gaussian with zero mean, so the Gaussian absolute-moment formula gives:
\begin{equation} \label{eq:tur6}  
	\mathbb{E}\left[ U^{(signal)}_{t}  \right] = \frac{2\sigma_{\mathrm{target}}}{\sqrt{\pi}}  \sqrt{ (1-\nu) \vartheta_{0} }  =  \frac{2\sigma_{\mathrm{target}}}{\sqrt{\pi}}  \sqrt{ (1-\nu) }, 
\end{equation}
where we apply the fact that the variance $\vartheta_{0}$ of normalized returns is $1$.

The leading-order approximation retains the first term in equation \eqref{eq:tur3} and drops the second term, which is proportional to the absolute change in the volatility estimate $\left|\sigma_{t}- \sigma_{t-1}\right|$. The omitted volatility-update term is not negligible in general, and Section \ref{sc:eft} quantifies its contribution within the full pipeline.

\subsection{Proof of Proposition \ref{pr:turnover2}}\label{sec:ta2}

The long-short signal in equation \eqref{eq:ewma3} equals $S_{n} = l_{1}\mathcal{L}^{(\nu_{1})}(z_{n}) - l_{2}\mathcal{L}^{(\nu_{2})}(z_{n})$ in terms of the raw filters, with $l_{1}(1-\nu_{1}) = l_{2}(1-\nu_{2}) = q$ by equation \eqref{eq:ewma3a}. The recursion $\mathcal{L}^{(\nu)}(z_{n}) = (1-\nu)z_{n} + \nu \mathcal{L}^{(\nu)}(z_{n-1})$ gives the filter increment $(1-\nu)\left(z_{n} - \mathcal{L}^{(\nu)}(z_{n-1})\right)$, so the terms in $z_{n}$ cancel in the signal increment:
\begin{equation} \label{eq:tur_ls1}
	S_{n} - S_{n-1} = q\left( \mathcal{L}^{(\nu_{2})}(z_{n-1}) - \mathcal{L}^{(\nu_{1})}(z_{n-1}) \right).
\end{equation}
Under serially independent $z_{t}$ with unit variance, the variance of the increment follows from the filter variances and the cross-covariance in equation \eqref{eq:ls3}:
\begin{equation} \label{eq:tur_ls2}
	\Var\left[ S_{n} - S_{n-1} \right] = q^{2}\left( \frac{1-\nu_{1}}{1+\nu_{1}}  +  \frac{1-\nu_{2}}{1+\nu_{2}}  - \frac{2\left( 1-\nu_{1}\right) \left( 1-\nu_{2}\right)}{1-\nu_{1}\nu_{2}}\right) = 2\zeta.
\end{equation}
The increment is Gaussian with zero mean, so $\mathbb{E}\left|S_{n} - S_{n-1}\right| = \sqrt{(2/\pi)\, \Var\left[ S_{n} - S_{n-1} \right]} = 2\sqrt{\zeta/\pi}$. The assumption of equation \eqref{eq:tur6} then gives equation \eqref{eq:tur_ls}.

\subsection{Proof of Lemma~\ref{lem:isserlis}}\label{sec:isserlis}

We write $X=\mu_{X}+x$ and $Y=\mu_{Y}+y$ with the centered variables $x=\sum_{s}a_{s}\epsilon_{t-s}$ and $y=\sum_{s}b_{s}\epsilon_{t-s}$. The third-order moments $\mathbb{E}[x^{2}y]$ and $\mathbb{E}[xy^{2}]$ vanish, because they are proportional to the third moment of the innovations, and the multilinearity of cumulants over the independent innovations gives $\mathrm{cum}\left(x,x,y,y\right)=\kappa\sum_{s}a^{2}_{s}b^{2}_{s}$. The moment identity $\mathbb{E}[x^{2}y^{2}]=\sigma^{2}_{X}\sigma^{2}_{Y}+2\sigma^{2}_{XY}+\mathrm{cum}\left(x,x,y,y\right)$ expresses the fourth moment through cumulants of order two and four. Expanding $\mathbb{E}[(XY)^{2}]$ and subtracting $\left(\mathbb{E}[XY]\right)^{2}=\left(\sigma_{XY}+\mu_{X}\mu_{Y}\right)^{2}$ yields equation \eqref{eq:sr_isserlis}. For $\kappa=0$, the result is the Isserlis theorem (\cite{Isserlis1918}).

\subsection{Proof of Proposition~\ref{prop:dailyvar}}\label{sec:dailyvar}
We write $\Lambda_{t-1}=\mathcal{L}^{(\nu)}(z_{t-1})$ for the lagged filter. Equation \eqref{eq:e2} then reads $f_{t}=\left(l\sigma_{\mathrm{target}}/\sqrt{a}\right) z_{t}\Lambda_{t-1}$. The filter weights sum to one, so $\mathbb{E}[\Lambda_{t-1}]=\mu$. The covariance with the next return follows by summing the autocovariances:
\begin{equation} \label{eq:sr_cov}
	\Cov\left[ z_{t}, \Lambda_{t-1} \right] = \left(1-\nu\right)\sum^{\infty}_{k=0}\nu^{k}\, \vartheta\rho(k+1) = \vartheta \frac{1-\nu}{\nu}\Psi_{\nu} = \vartheta A_{\nu}.
\end{equation}
For the variance of the filter, we sum over both indices and split the double sum into the diagonal $j=k$ and the two off-diagonal parts with $m=\left|j-k\right|\geq 1$:
\begin{equation} \label{eq:sr_varL}
	\Var\left[ \Lambda_{t-1} \right] = \left(1-\nu\right)^{2} \vartheta \sum^{\infty}_{j,k=0}\nu^{j+k}\rho(\left|j-k\right|) = \left(1-\nu\right)^{2} \vartheta\, \frac{1+2\Psi_{\nu}}{1-\nu^{2}} = \vartheta B_{\nu}.
\end{equation}
The mean of $f_{t}$ follows from $\mathbb{E}\left[z_{t}\Lambda_{t-1}\right]=\Cov\left[z_{t},\Lambda_{t-1}\right]+\mu^{2}$. Under Assumption \ref{as:lp}, the centered return is $z_{t}-\mu=\sqrt{\vartheta}\sum_{s\geq 0}\psi_{s}\epsilon_{t-s}$. The filter applies the weights $(1-\nu)\nu^{m}$ to the lagged returns, so the centered filter is $\Lambda_{t-1}-\mu=\sqrt{\vartheta}\sum_{u\geq 0}g_{u}\epsilon_{t-1-u}$ with $g_{u}$ of equation \eqref{eq:sr_k}. The innovation $\epsilon_{t-s}$ enters the return with the loading $\sqrt{\vartheta}\psi_{s}$ and enters the filter with the loading $\sqrt{\vartheta}g_{s-1}$ for $s\geq 1$. The fourth-moment term of Lemma \ref{lem:isserlis} equals $\kappa\vartheta^{2}\sum_{s\geq 1}\psi^{2}_{s}g^{2}_{s-1}=\kappa\vartheta^{2}K_{\nu}$. The variance follows by applying Lemma \ref{lem:isserlis} with $\sigma^{2}_{X}=\vartheta$, $\sigma^{2}_{Y}=\vartheta B_{\nu}$, $\sigma_{XY}=\vartheta A_{\nu}$, $\mu_{X}=\mu_{Y}=\mu$, the fourth-moment term $\kappa\vartheta^{2}K_{\nu}$, and collecting the terms.

\section{Asymptotics of the Cost-Optimal Span}\label{app:asym}

This appendix formalizes the asymptotic statements of Section \ref{sec:sp} on the cost-optimal span. Throughout, we consider the zero-drift case with $\kappa=0$ and the single variance-preserving filter, we write $\eta = 1-\nu = 2/(\spn+1)$, we set $c_{0}=2ac/\sqrt{\pi}$, and $SR^{net}$ denotes the leading-order form of equation \eqref{eq:sr_net}, which charges the cost against the signal-turnover proxy of Proposition \ref{pr:turnover1}. The break-even and optimal-span results therefore describe the independence-based proxy cost model, and Section \ref{ssec:sr_mc} quantifies the pipeline turnover in excess of it. The gross Sharpe ratio and the cost drag share the factor $1/\sqrt{B_{\nu}+A^{2}_{\nu}}$, so the net Sharpe ratio in equation \eqref{eq:sr_net} takes the form:
\begin{equation}\label{eq:asym1}
	SR^{net}(\eta) = \frac{\sqrt{\eta}\, M(\eta)}{\sqrt{W(\eta)}}, \qquad
	M(\eta) = \frac{\sqrt{a}\,A_{\nu}}{\eta} - \frac{c_{0}}{\sqrt{2-\eta}}, \qquad
	W(\eta) = \frac{B_{\nu} + A^{2}_{\nu}}{\eta}.
\end{equation}

\begin{proposition}[Break-even cost and optimal span under AR-1]\label{pr:asym_ar1}
	Let $\{z_{t}\}$ follow the zero-drift AR-1 process with $\phi \in (0,1)$.
	
	(i) The net Sharpe ratio vanishes at the span with $\eta = 2/(\spn+1)$ if and only if $c = c^{*}(\eta)$ with:
	\begin{equation}\label{eq:asym2}
		c^{*}(\eta) = \sqrt{\frac{\pi}{2a}}\, \frac{\phi\sqrt{1-\eta/2}}{1-\phi+\eta\phi}.
	\end{equation}
	The break-even cost strictly increases in the span with the supremum $c^{*}_{\infty}$ in equation \eqref{eq:sr_net_ar1}, and the net Sharpe ratio is positive at some span if and only if $c < c^{*}_{\infty}$.
	
	(ii) Let $c < c^{*}_{\infty}$ and $v = c/c^{*}_{\infty}$. For $v$ sufficiently close to one, an interior maximiser $\spn^{*}(c)$ of the net Sharpe ratio exists, and any such maximiser satisfies:
	\begin{equation}\label{eq:asym3}
		\lim_{v \to 1} \left( 1 - v \right) \spn^{*}(c) = \frac{6\phi}{1-\phi} + \frac{3}{2},
	\end{equation}
	and equation \eqref{eq:sr_net_span} restates this limit at finite $v$.
\end{proposition}

\begin{proof}
	Under the AR-1 process, $\Psi_{\nu} = \nu\phi/(1-\nu\phi)$ gives $A_{\nu} = \eta\phi/(1-\phi+\eta\phi)$, so $M(\eta) = \sqrt{a}\phi/(1-\phi+\eta\phi) - c_{0}/\sqrt{2-\eta}$.
	
	(i) Both $\eta$ and $W(\eta)$ are positive on $(0,1)$, so the net Sharpe ratio vanishes if and only if $M(\eta)=0$, which rearranges to $c = c^{*}(\eta)$ in equation \eqref{eq:asym2}. The logarithmic derivative of $c^{*}$ equals $-1/(2(2-\eta)) - \phi/(1-\phi+\eta\phi) < 0$, so $c^{*}$ strictly decreases in $\eta$ from the supremum $c^{*}(0+) = c^{*}_{\infty}$ to the infimum $c^{*}_{1} \equiv c^{*}(1^{-}) = \phi\sqrt{\pi/(4a)}$, and strictly increases in the span. The sign of $M(\eta)$ equals the sign of $c^{*}(\eta) - c$, so three regimes follow. For $c \geq c^{*}_{\infty}$, the net Sharpe ratio is negative at every span. For $c^{*}_{1} < c < c^{*}_{\infty}$, the strictly decreasing function $M$ has a unique root $\eta_{0} \in (0,1)$, and the net Sharpe ratio is positive exactly on $(0, \eta_{0})$. For $0 \leq c \leq c^{*}_{1}$, we have $c^{*}(\eta) > c$ at every $\eta \in (0,1)$, so the net Sharpe ratio is positive at every span, and at $c = c^{*}_{1}$ the root occurs only at the excluded boundary $\eta = 1$.
	
	(ii) By part (i), an interior root $\eta_{0} \in (0,1)$ exists if and only if $v > c^{*}_{1}/c^{*}_{\infty} = (1-\phi)/\sqrt{2}$, and the net Sharpe ratio is then positive exactly on $(0, \eta_{0})$. In the regime $v \to 1$ of the proposition, this condition holds and the root exists. Expanding equation \eqref{eq:asym2} at $\eta_{0}=0$ gives $1 - v = \eta_{0}\left( \phi/(1-\phi) + 1/4 \right) + O(\eta^{2}_{0})$, so $\eta_{0} = 4(1-\phi)(1-v)/(1+3\phi) + O\left((1-v)^{2}\right)$, and any maximiser $\eta^{*} \in (0, \eta_{0})$ tends to zero as $v \to 1$. The net Sharpe ratio is continuous and positive on $(0,\eta_{0})$ and vanishes at both endpoints, so an interior maximiser exists and satisfies the first-order condition:
	\begin{equation}\label{eq:asym4}
		\frac{1}{2\eta^{*}} + \frac{M'(\eta^{*})}{M(\eta^{*})} - \frac{W'(\eta^{*})}{2W(\eta^{*})} = 0.
	\end{equation}
	The term $W'/(2W)$ is bounded on $[0, 1/2]$ by a constant that depends only on $\phi$, so equation \eqref{eq:asym4} gives $M(\eta^{*}) = 2\eta^{*}|M'(\eta^{*})|\left(1+O(\eta^{*})\right)$. The mean value theorem gives $M(\eta^{*}) = M(0) - |M'(\xi)|\eta^{*}$ with $\xi \in (0, \eta^{*})$, and both $|M'(\xi)|$ and $|M'(\eta^{*})|$ converge to $|M'(0)|$ as $\eta^{*} \to 0$. Combining the two displays yields $\eta^{*} = M(0)/\left(3|M'(0)|\right)\left(1+o(1)\right)$. Substituting $M(0) = \sqrt{a}\phi(1-v)/(1-\phi)$, $|M'(0)| = \sqrt{a}\phi\left(\phi/(1-\phi) + v/4\right)/(1-\phi)$, and $\spn^{*} = 2/\eta^{*} - 1$ gives the limit in equation \eqref{eq:asym3}.
\end{proof}

\begin{proposition}[Cost-optimal span under long memory]\label{pr:asym_arfima}
	Let $\{z_{t}\}$ follow the zero-drift ARFIMA$(0,d,0)$ process with $d \in (0, 1/2)$, and let $C_{\psi} = \Gamma(1-d)\Gamma(2d)/\Gamma(d)$. As $c \to \infty$, any maximiser of the net Sharpe ratio satisfies:
	\begin{equation}\label{eq:asym5}
		\spn^{*}(c) = 2 \left( \sqrt{\frac{2a}{\pi}}\, \frac{(1+2d)\, c}{C_{\psi} (1-2d)} \right)^{1/(2d)} \left(1 + o(1)\right).
	\end{equation}
	At the optimum, the ratio of the cost drag to the gross Sharpe ratio converges to $(1-2d)/(1+2d)$, so the net-to-gross ratio converges to $4d/(1+2d)$.
\end{proposition}

\begin{proof}
	The hypergeometric connection formula applied to equation \eqref{eq:sr_agf_hyp} gives the exact decomposition:
	\begin{equation}\label{eq:asym6}
		\Phi_{\nu} = C_{\psi}\, \eta^{-2d} (1-\eta)^{d} + C_{\gamma}\, F(d, 1, 1+2d; \eta), \qquad
		C_{\gamma} = \frac{\Gamma(1-d)\,\Gamma(-2d)}{\Gamma(1-2d)\,\Gamma(-d)},
	\end{equation}
	because the first hypergeometric factor in the connection formula reduces through $F(1-2d, -d, 1-2d; \eta) = (1-\eta)^{d}$. The second term in equation \eqref{eq:asym6} is bounded on $[0, 1/2]$, so $\Psi_{\nu} = \Phi_{\nu} - 1 = C_{\psi}\eta^{-2d}\left(1 + O(\eta^{\alpha})\right)$ with $\alpha = \min(2d, 1-2d)$. Substituting into $A_{\nu} = \eta\Psi_{\nu}/\nu$ and $B_{\nu} = \eta(1+2\Psi_{\nu})/(2-\eta)$, and using $A^{2}_{\nu}/B_{\nu} = O(\eta^{1-2d})$, we obtain the uniform expansion on $(0, 1/2]$:
	\begin{equation}\label{eq:asym7}
		SR^{net}(\eta) = P \eta^{p}\left(1 + r_{1}(\eta)\right) - Q \eta^{q}\left(1 + r_{2}(\eta)\right), \qquad |r_{i}(\eta)| \leq C\eta^{\alpha},
	\end{equation}
	with $p = (1-2d)/2$, $q = (1+2d)/2$, $P = \sqrt{a C_{\psi}}$, and $Q = c_{0}/\sqrt{2 C_{\psi}}$.
	
	The leading two-power profile $G(\eta) = P\eta^{p} - Q\eta^{q}$ attains its unique maximum at $\eta_{0} = \left(Pp/(Qq)\right)^{1/(2d)}$, where the drag-to-alpha ratio equals $Q\eta_{0}^{q}/(P\eta_{0}^{p}) = p/q$. Positivity of equation \eqref{eq:asym7} requires $\eta^{2d} \leq (P/Q)(1+O(\eta^{\alpha}))$, so any maximiser satisfies $\eta^{*} \leq C'(P/Q)^{1/(2d)}$ and tends to zero as $c \to \infty$, because $Q$ grows linearly in $c$. Writing $G(\eta) = G(\eta_{0}) H(\eta/\eta_{0})$ with the universal profile $H(s) = (q s^{p} - p s^{q})/(q-p)$, the error bounds in equation \eqref{eq:asym7} give $H(\eta^{*}/\eta_{0}) \geq 1 - C''\eta_{0}^{\alpha}$, because the maximiser value must weakly exceed the value at $\eta_{0}$. The profile $H$ is continuous and strictly unimodal with $H(1)=1$, so $\eta^{*}/\eta_{0} \to 1$, which gives equation \eqref{eq:asym5} through $\spn^{*} = 2/\eta^{*} - 1$. The ratio of the cost drag to the gross Sharpe ratio at the optimum equals $(Q/P)(\eta^{*})^{2d}(1+O((\eta^{*})^{\alpha})) \to p/q = (1-2d)/(1+2d)$.
\end{proof}

\begin{remark}[White noise with drift]\label{rem:asym_wn}
	Under the zero-autocorrelation process with drift $\mu^{z}_{an} \neq 0$, direct differentiation of equation \eqref{eq:sr_wn} shows that the gross Sharpe ratio increases in the span with the limit $|\mu^{z}_{an}|$, while the proxy cost drag declines to zero. The net Sharpe ratio under the signal-turnover proxy therefore converges to $|\mu^{z}_{an}|$ for every proportional cost, which formalizes the buy-and-hold limit of Section \ref{sec:sp}, up to a deterministic scale. The limit describes the proxy cost model: the mean of the variance-preserving signal grows as $\mu^{z}_{an}\sqrt{\spn}$, so, with an estimated volatility, the omitted volatility-update turnover need not vanish at long spans.
\end{remark}

\section{Skewness of Aggregated Returns under White Noise}\label{app:skew}

Let the volatility-normalized returns $z_{t}$ be independent standard normal. Write the unit-variance signal as $S_{t-1}=\sum_{j\geq 1}w_{j}z_{t-j}$ with $\sum_{j\geq 1}w^{2}_{j}=1$ and the signal autocorrelation $R(h)=\sum_{j\geq 1}w_{j}w_{j+h}$. The daily system return is $f_{t}=S_{t-1}z_{t}$ up to the loading and the volatility target, which scale $f_{t}$ linearly and cancel in the skewness. For the single filter, $w_{j}=\sqrt{1-\nu^{2}}\,\nu^{j-1}$ and $R(h)=\nu^{h}$.

\begin{lemma}[Pairing identity for third-order moments]\label{lem:skew_pair}
	For $h\geq 1$, $\mathbb{E}\left[f^{2}_{t}f_{t-h}\right]=2w_{h}R(h)$. All other third-order moments vanish: $\mathbb{E}\left[f^{3}_{t}\right]=0$, $\mathbb{E}\left[f^{2}_{t}f_{t+h}\right]=0$, and $\mathbb{E}\left[f_{t}f_{s}f_{u}\right]=0$ for distinct $t$, $s$, $u$.
\end{lemma}

\begin{proof}
	The return $z_{t}$ is independent of all signals and returns up to time $t$. In $\mathbb{E}\left[f^{3}_{t}\right]=\mathbb{E}\left[S^{3}_{t-1}\right]\mathbb{E}\left[z^{3}_{t}\right]$, both factors vanish for centred Gaussians. In $\mathbb{E}\left[f^{2}_{t}f_{t+h}\right]$ and in $\mathbb{E}\left[f_{t}f_{s}f_{u}\right]$ with $t>s>u$, the return with the largest time index enters the product exactly once, so the expectation vanishes under the Isserlis theorem. For the surviving moment, the independence of $z_{t}$ gives $\mathbb{E}\left[f^{2}_{t}f_{t-h}\right]=\mathbb{E}\left[S^{2}_{t-1}S_{t-h-1}z_{t-h}\right]$. We decompose $S_{t-1}=w_{h}z_{t-h}+R_{h}$, where $R_{h}$ collects the remaining innovations and is independent of $z_{t-h}$. Expanding the square, the terms with an odd total power of $z_{t-h}$ vanish, and the cross term gives $2w_{h}\,\mathbb{E}\left[z^{2}_{t-h}\right]\mathbb{E}\left[R_{h}S_{t-h-1}\right]=2w_{h}R(h)$, because $S_{t-h-1}$ contains no $z_{t-h}$, so that $\mathbb{E}\left[R_{h}S_{t-h-1}\right]=\mathbb{E}\left[S_{t-1}S_{t-h-1}\right]=R(h)$.
\end{proof}

\begin{proposition}[Skewness of the aggregated return]\label{pr:skew}
	Let $F_{T}=\sum^{T}_{t=1}f_{t}$. Then $\Var\left[F_{T}\right]=T$, and
	\begin{equation}\label{eq:skew_gen}
		\mathbb{E}\left[F^{3}_{T}\right]=6\sum^{T-1}_{h=1}\left(T-h\right)w_{h}R(h).
	\end{equation}
	For the single filter, the skewness $\varsigma(T)=\mathbb{E}\left[F^{3}_{T}\right]/T^{3/2}$ admits the closed form
	\begin{equation}\label{eq:skew_t}
		\varsigma(T) = \frac{6\nu\left(T\left(1-\nu^{2}\right)-1+\nu^{2T}\right)}{\left(1-\nu^{2}\right)^{3/2}\,T^{3/2}}.
	\end{equation}
	It satisfies $\varsigma(1)=0$ and $\varsigma(T)>0$ for all $T\geq 2$.
\end{proposition}

\begin{proof}
	The covariances $\mathbb{E}\left[f_{t}f_{s}\right]$ vanish for $t\neq s$ by the largest-index argument of Lemma \ref{lem:skew_pair}, and $\Var\left[f_{t}\right]=\mathbb{E}\left[S^{2}_{t-1}\right]\mathbb{E}\left[z^{2}_{t}\right]=1$, which gives $\Var\left[F_{T}\right]=T$. The multinomial expansion of $F^{3}_{T}$ contains the diagonal terms, the paired terms, and the terms with three distinct indices. Lemma \ref{lem:skew_pair} annihilates all but the paired terms with the later index squared, so $\mathbb{E}\left[F^{3}_{T}\right]=3\sum_{t>s}\mathbb{E}\left[f^{2}_{t}f_{s}\right]=6\sum^{T-1}_{h=1}\left(T-h\right)w_{h}R(h)$. For the single filter, $w_{h}R(h)=\sqrt{1-\nu^{2}}\,\nu^{2h-1}$, and the identity $\sum^{T-1}_{h=1}\left(T-h\right)x^{h}=x\left(T\left(1-x\right)-1+x^{T}\right)/\left(1-x\right)^{2}$ with $x=\nu^{2}$ yields equation \eqref{eq:skew_t}. For positivity, the function $g(x)=T\left(1-x\right)-1+x^{T}$ satisfies $g(1)=0$ and $g^{\prime}(x)=-T\left(1-x^{T-1}\right)<0$ on $\left(0,1\right)$, so $g\left(\nu^{2}\right)>0$ for $T\geq 2$, while $g$ vanishes identically at $T=1$.
\end{proof}

The long-short filter satisfies Lemma \ref{lem:skew_pair} with its own weights and autocorrelation, so its third moment follows from equation \eqref{eq:skew_gen} with the long-short $w_{h}$ and $R(h)$. Under a general stationary autocorrelation function, additional pairings of the order of the autocorrelations enter the third-order moments, so the white-noise profile gives the leading-order skewness for weakly autocorrelated returns. Under the AR-1 process, the additional pairings form geometric sums, so the third moment remains closed-form, and the expansion generalizes to the ARFIMA process. We develop the general case and the asymptotics of the profile in the follow-up paper.

\end{document}